\ifdefined \issubmit \else \def \issubmit{0} \fi

\newif \ifsubmit {} \if \issubmit 0 \submitfalse \else \submittrue \fi

\ifsubmit
\documentclass[acmsmall]{acmart}
\else
\documentclass[acmsmall,review,nonacm]{acmart} % remove nonacm if copyright has to go in
\fi

\PassOptionsToPackage{prologue,dvipsnames}{xcolor}
\usepackage[T5,T1]{fontenc}
\usepackage{xspace}
\usepackage[frozencache=true,cachedir=.]{minted}
\setminted{fontsize=\footnotesize,
  escapeinside=\#\#,
  mathescape=true}
\setmintedinline{fontsize=\footnotesize,
  escapeinside=\&\&,
  mathescape=true}
\usepackage{microtype}                % Better typesetting and layout algorithms

\usepackage{float}                    % Float layout placement
\usepackage{placeins}                 % and management

\usepackage{hyperref}                 % Hyperlinked references
\hypersetup{
    colorlinks,
    linkcolor={red!50!black},
    citecolor={blue!50!black},
    urlcolor={blue!80!black}
}

\usepackage{amsmath}                  % AMS math tools
\usepackage{amsfonts}                 % ...
\usepackage{amsthm}                   % ...
\usepackage{mathtools}                % ...

\usepackage{cleveref}

\newtheorem*{fact*}{Fact}                       % ...
\newtheorem{definition}{Definition}[section]    % ...
\newtheorem*{definition*}{Definition}           % ...
  % ...
\newtheorem*{proposition*}{Proposition}         % ...
\newtheorem{theorem}{Theorem}[section]          % ...
\newtheorem*{theorem*}{Theorem}                 % ...
\newtheorem{lemma}[theorem]{Lemma}              % ...
\newtheorem*{lemma*}{Lemma}                     % ...
        % ...
\newtheorem*{sublemma*}{Sublemma}               % ...
      % ...
\newtheorem*{corollary*}{Corollary}             % ...

\usepackage[dvipsnames]{xcolor}       % Colors

\usepackage{stmaryrd}                 % More symbols
\usepackage{pifont}                   % ...
\usepackage{textcomp}                 % ...

\usepackage{centernot}                % slashing things

\usepackage{stackengine}              % Symbol Stacking
\stackMath

\usepackage{nicefrac}                 % Nice Fractions

\usepackage{array}                    % Array Layout
\usepackage{tabularx}                 % Improved Tabular Layout

\usepackage{framed}                   % Framed Environment

\usepackage{lipsum}                   % Easy placeholder text generation

\usepackage{enumitem}                 % Better control over itemize/enumerate

\usepackage{mathpartir}               % Flexible math display layout

\usepackage{galois}                   % Galois connections

\usepackage{tikz}

\usepackage{relsize}

\usepackage[normalem]{ulem}
\usepackage{comment}

\usepackage{colortbl}

\newcommand{\mtext}[1]{\ifmmode\operatorname{#1}\else\textnormal{#1}\fi}
\newcommand{\mtexttt}[1]{\ifmmode\operatorname{\mathtt{#1}}\else\textnormal{\texttt{#1}}\fi}
\newcommand{\mtextit}[1]{\ifmmode\operatorname{\mathit{#1}}\else\textnormal{\textit{#1}}\fi}
\newcommand{\mtextbf}[1]{\ifmmode\operatorname{\mathbf{#1}}\else\textnormal{\textbf{#1}}\fi}
\newcommand{\mtextsc}[1]{\ifmmode\operatorname{\textsc{\smaller #1}}\else\textnormal{\textsc{\smaller #1}}\fi}

\definecolor{cbsafeABright}{RGB}{8,72,145}
\definecolor{cbsafeADark}{RGB}{109,36,150}

\definecolor{cbsafeBBright}{RGB}{85,119,13}
\definecolor{cbsafeBDark}{RGB}{109,70,17}

\definecolor{cbsafeCBright}{RGB}{150,48,89}
\definecolor{cbsafeCDark}{RGB}{70,24,48}

\newcommand{\colorMATHA}{cbsafeABright}
\newcommand{\colorSYNTAXA}{cbsafeADark!80!black}
\newcommand{\colorMATHB}{cbsafeBBright}
\newcommand{\colorSYNTAXB}{cbsafeBDark}
\newcommand{\colorMATHC}{cbsafeCBright}
\newcommand{\colorSYNTAXC}{cbsafeCDark}

\newcommand{\colorTEXT}{black}

\newcommand{\colorMATH}{\colorMATHA}
\newcommand{\colorSYNTAX}{\colorSYNTAXA}

\renewcommand{\paragraph}[1]{\vspace{5pt}\noindent\textbf{#1}}

\definecolor{mygray}{gray}{0.97}
\setminted{bgcolor=white}
\usepackage{etoolbox}
\makeatletter
\patchcmd{\minted@colorbg}{\medskip}{}{}{}
\patchcmd{\endminted@colorbg}{\medskip}{}{}{}
\makeatother

\newcommand{\solo}{\textsc{Solo}\xspace}

\newcommand{\duet}{\textsc{Duet}\xspace}
\newcommand{\fuzz}{\textsc{Fuzz}\xspace}
\newcommand{\pinq}{\textsc{PINQ}\xspace}
\newcommand{\dduo}{\textsc{DDuo}\xspace}
\newcommand{\dfuzz}{\textsc{DFuzz}\xspace}
\newcommand{\fuzzi}{\textsc{Fuzzi}\xspace}
\newcommand{\dpella}{\textsc{DPella}\xspace}

\newcommand{\apRHL}{\textsc{apRHL}\xspace}
\newcommand{\apRHLplus}{\textsc{apRHL$^+$}\xspace}

\newcommand{\inline}[1]{\mintinline{haskell}{#1}}

\author{Chike Abuah}
\email{cabuah@uvm.edu}
\affiliation{%
  \institution{University of Vermont}
}

\author{David Darais}
\email{darais@galois.com}
\affiliation{%
  \institution{Galois, Inc.}
}

\author{Joseph P. Near}
\email{jnear@uvm.edu}
\affiliation{%
  \institution{University of Vermont}
}

\title{\solo: A Lightweight Static Analysis for Differential Privacy}

\begin{abstract}
Existing approaches for statically enforcing differential privacy in higher
order languages use either linear or relational refinement types.
A barrier to adoption for these approaches is the lack of support for expressing
these ``fancy types'' in mainstream programming languages.
For example, no mainstream language supports relational refinement types, and
although Rust and modern versions of Haskell both employ some linear typing
techniques, they are inadequate for embedding enforcement of differential
privacy, which requires ``full'' linear types a la Girard.
We propose a new type system that enforces differential privacy, avoids the use
of linear and relational refinement types, and can be easily embedded in
mainstream richly typed programming languages such as Scala, OCaml and Haskell.
We demonstrate such an embedding in Haskell, demonstrate its expressiveness on
case studies, and prove soundness of our type-based enforcement of differential privacy.
\end{abstract}

\makeatletter
\let\@authorsaddresses\@empty
\makeatother

\begin{document}
\maketitle
\thispagestyle{plain}
\pagestyle{plain}

\section{Introduction}
\label{sec:intro}

Differential privacy has become the standard for protecting the privacy of individuals with formal guarantees of \emph{plausible deniability}. It has been adopted for use at several high-profile institutions such as Google \cite{rappor}, Facebook \cite{nayak_2020}, and the US Census Bureau \cite{abowd2018}. However, experience has shown that implementation mistakes are easy to make---and difficult to catch---in differentially private algorithms~\cite{lyu2017understanding}. Verifying that differentially private programs \emph{actually} ensure differential privacy is thus an important problem, given the sensitive nature of the data processed by these programs.

Recent work has made significant progress towards techniques for static verification of differentially private programs. Existing techniques typically define novel programming languages that incorporate  specialized static type systems (linear types~\cite{reed2010distance,near2019duet}, relational types~\cite{barthe2015higher}, dependent types~\cite{gaboardi2013linear}, etc.). However, there remains a major challenge in bringing these techniques to practice: the specialized features they rely on do not exist in mainstream programming languages.

We introduce \solo, a novel type system for static verification of differential privacy that does \emph{not} rely on linear types, and present a reference implementation \emph{as a Haskell library}. \solo is similar to \fuzz~\cite{reed2010distance} and its descendants in expressive power, but \solo can be implemented entirely in Haskell with no additional language extensions. In particular, \solo's sensitivity and privacy tracking mechanisms are compatible with higher-order functions, and leverage Haskell's type inference system to minimize the need for additional type annotations.

In differential privacy, the \emph{sensitivity} of a computation determines how much noise must be added to its result to achieve differential privacy. \fuzz-like languages track sensitivity relative to program variables, using a linear typing discipline. The key innovation in \solo is to track sensitivity relative to a set of global \emph{data sources} instead, which eliminates the need for linear types. Compared to prior work on static verification of differential privacy, our system can be embedded in existing programming languages without support for linear types, and supports advanced variants of differential privacy like {{\color{\colorMATH}\ensuremath{(\epsilon , \delta )}}}-differential privacy and R\'enyi differential privacy.

We describe our approach using the Haskell implementation of \solo, and demonstrate its use to verify differential privacy for practical algorithms in four case studies. We formalize a subset of \solo's sensitivity analysis and prove \emph{metric preservation}, the soundness property for this analysis.

\paragraph{Contributions.}
In summary, we make the following contributions:

\begin{itemize}[topsep=1mm,leftmargin=6mm]
\item We introduce \solo, a novel type system for the static verification of differential privacy without linear types (\S\ref{sec:solo-example}).
\item We present a reference implementation of \solo as a Haskell library, which retains support for type inference and does not require additional language extensions (\S\ref{sec:sensitivity}, \S\ref{sec:privacy}).
\item We formalize a subset of \solo's type system and prove its soundness (\S\ref{sec:formalism}).
\item We demonstrate the applicability of the \solo library in four case studies (\S\ref{sec:case}).
\end{itemize}

\section{Background}
\label{sec:background}

This section provides a summary of the fundamentals of differential privacy.
Differential privacy~\cite{dwork2006calibrating} affords a notion of \emph{plausible deniability} at the individual level to participants in aggregate data analysis queries. In principle, a differentially private algorithm {{\color{\colorMATH}\ensuremath{{\mathcal{K}}}}} over several individuals must include enough random noise to make the participation (removal/addition) of any one individual statistically unrecognizable. While this guarantee is typically in terms of a \emph{symmetric difference} of one individual, formally a distance metric between datasets {{\color{\colorMATH}\ensuremath{{\textit{d}}}}} is specified.
\begin{definition}[Differential privacy]
  For a distance metric {{\color{\colorMATH}\ensuremath{{\textit{d}}_{A} \in  A \times  A \rightarrow  {\mathbb{R}}}}}, a randomized \emph{mechanism} {{\color{\colorMATH}\ensuremath{{\mathcal{K}} \in  A \rightarrow  B }}} is ({{\color{\colorMATH}\ensuremath{\epsilon , \delta }}})-differentially private if {{\color{\colorMATH}\ensuremath{ \hspace*{0.33em}\forall  x, x^{\prime} \in  A }}} s.t. {{\color{\colorMATH}\ensuremath{{\textit{d}}_{A}(x, x^{\prime}) \leq  1}}}, considering any set S of possible outcomes, we have that: {{\color{\colorMATH}\ensuremath{ {\mtext{Pr}}[{\mathcal{K}}(x) \in  S] \leq  e^{\epsilon } {\mtext{Pr}}[{\mathcal{K}}(x^{\prime}) \in  S] + \delta  }}}.
\end{definition}
We say that two inputs {{\color{\colorMATH}\ensuremath{x}}} and {{\color{\colorMATH}\ensuremath{x^{\prime}}}} are \emph{neighbors} when {{\color{\colorMATH}\ensuremath{{\textit{d}}_{A}(x, x^{\prime}) = 1}}}. To provide meaningful privacy protection, two neighboring inputs are normally considered to differ in the data of a single individual. Thus, the definition of differential privacy ensures that the probably distribution over {{\color{\colorMATH}\ensuremath{{\mathcal{K}}}}}'s outputs will be roughly the same, whether or not the data of a single individual is included in the input. The strength of the guarantee is parameterized by the \emph{privacy parameters} {{\color{\colorMATH}\ensuremath{\epsilon }}} and {{\color{\colorMATH}\ensuremath{\delta }}}. The case when {{\color{\colorMATH}\ensuremath{\delta =0}}} is often called \emph{pure} {{\color{\colorMATH}\ensuremath{\epsilon }}}-differential privacy; the case when {{\color{\colorMATH}\ensuremath{\delta  > 0}}} is often called \emph{approximate} or {{\color{\colorMATH}\ensuremath{(\epsilon , \delta )}}}-differential privacy. When {{\color{\colorMATH}\ensuremath{\delta  > 0}}}, the {{\color{\colorMATH}\ensuremath{\delta }}} parameter can be thought of as a \emph{failure probability}: with probability {{\color{\colorMATH}\ensuremath{1-\delta }}}, the mechanism achieves pure {{\color{\colorMATH}\ensuremath{\epsilon }}}-differential privacy, but with probability {{\color{\colorMATH}\ensuremath{\delta }}}, the mechanism makes no guarantee at all (and may violate privacy arbitrarily). The {{\color{\colorMATH}\ensuremath{\delta }}} parameter is therefore set very small---values on the order of {{\color{\colorMATH}\ensuremath{10^{-5}}}} are often used. Typical values for {{\color{\colorMATH}\ensuremath{\epsilon }}} are in the range of {{\color{\colorMATH}\ensuremath{0.1}}} to {{\color{\colorMATH}\ensuremath{1}}}.

\paragraph{Sensitivity.}
The core mechanisms for differential privacy (described below) rely on the notion of \emph{sensitivity}~\cite{dwork2006calibrating} to determine how much noise is needed to achieve differential privacy. Intuitively, function sensitivity describes the rate of change of a function's output relative to its inputs, and is a scalar value that bounds this rate, in terms of some notion of distance. Formally:
\begin{definition}[Global Sensitivity]
\label{def:sensitivity}
  Given distance metrics {{\color{\colorMATH}\ensuremath{{\textit{d}}_{A}}}} and {{\color{\colorMATH}\ensuremath{{\textit{d}}_{B}}}}, a function {{\color{\colorMATH}\ensuremath{f \in  A \rightarrow  B }}} is said to be {{\color{\colorMATH}\ensuremath{s}}}\emph{-sensitive} if {{\color{\colorMATH}\ensuremath{\forall  s^{\prime} \in  {\mathbb{R}},\hspace*{0.33em} (x,y) \in  A.\hspace*{0.33em} {\textit{d}}_{A}(x,y) \leq  s^{\prime} \implies  {\textit{d}}_{B}(f(x),f(y)) \leq  s^{\prime}\mathord{\cdotp }s}}}.
\end{definition}
For example, the function {{\color{\colorMATH}\ensuremath{\lambda  x \mathrel{:} {\mathbb{R}} .\hspace*{0.33em} x + x}}} is {{\color{\colorMATH}\ensuremath{2}}}-sensitive, because its output is twice its input. Determining tight bounds on sensitivity is often the key challenge in ensuring differential privacy for complex algorithms.

\paragraph{Core Mechanisms.}
The core mechanisms that are often utilized to achieve differential privacy are the  \emph{Laplace mechanism}~\cite{dwork2014algorithmic} and the \emph{Gaussian mechanism}~\cite{dwork2014algorithmic}. Both mechanisms are defined for scalar values as well as vectors; the Laplace mechanism requires the use of the {{\color{\colorMATH}\ensuremath{L_{1}}}} distance metric and satisfies {{\color{\colorMATH}\ensuremath{\epsilon }}}-differential privacy, while the Gaussian mechanism requires the use of the {{\color{\colorMATH}\ensuremath{L_{2}}}} distance metric (which is often much smaller than {{\color{\colorMATH}\ensuremath{L_{1}}}} distance) and satisfies {{\color{\colorMATH}\ensuremath{(\epsilon , \delta )}}}-differential privacy (with {{\color{\colorMATH}\ensuremath{\delta  > 0}}}).
\begin{definition}[Laplace Mechanism]
\label{def:laplace}
Given a function {{\color{\colorMATH}\ensuremath{f \mathrel{:} A \rightarrow  {\mathbb{R}}^{d}}}} which is {{\color{\colorMATH}\ensuremath{s}}}-sensitive under the {{\color{\colorMATH}\ensuremath{L_{1}}}} distance metric {{\color{\colorMATH}\ensuremath{{\textit{d}}_{{\mathbb{R}}}(x, x^{\prime}) = \mathrel{\| }x - x^{\prime}\mathrel{\| }_{1}}}} on the function's output, the Laplace mechanism releases {{\color{\colorMATH}\ensuremath{ f(x) + Y_{1}, \ldots , Y_{d} }}}, where each of the values {{\color{\colorMATH}\ensuremath{Y_{1}, \ldots , Y_{d}}}} is drawn iid from the Laplace distribution centered at {{\color{\colorMATH}\ensuremath{0}}} with scale {{\color{\colorMATH}\ensuremath{\frac{s}{\epsilon }}}}; it satisfies {{\color{\colorMATH}\ensuremath{\epsilon }}}-differential privacy.
\end{definition}
\begin{definition}[Gaussian Mechanism]
Given a function {{\color{\colorMATH}\ensuremath{f \mathrel{:} A \rightarrow  {\mathbb{R}}^{d}}}} which is {{\color{\colorMATH}\ensuremath{s}}}-sensitive under the {{\color{\colorMATH}\ensuremath{L_{2}}}} distance metric {{\color{\colorMATH}\ensuremath{{\textit{d}}_{{\mathbb{R}}}(x, x^{\prime}) = \mathrel{\| }x - x^{\prime}\mathrel{\| }_{2}}}} on the function's output, the Gaussian mechanism releases {{\color{\colorMATH}\ensuremath{ f(x) + Y_{1}, \ldots , Y_{d} }}}, where each of the values {{\color{\colorMATH}\ensuremath{Y_{1}, \ldots , Y_{d}}}} is drawn iid from the Gaussian distribution centered at {{\color{\colorMATH}\ensuremath{0}}} with variance {{\color{\colorMATH}\ensuremath{\sigma ^{2} = \frac{2s ^{2}\ln (1.25/\delta )}{\epsilon ^{2}}}}}; it satisfies ({{\color{\colorMATH}\ensuremath{\epsilon ,\delta }}})-differential privacy for {{\color{\colorMATH}\ensuremath{\delta  > 0}}}.
\end{definition}

\paragraph{Composition.}
Multiple invocations of a privacy mechanism on the same data degrade in an additive or compositional manner. For example, the law of \emph{sequential composition} states that:
\begin{theorem}[Sequential Composition]\label{thm:sequential-composition}
If two mechanisms {{\color{\colorMATH}\ensuremath{{\mathcal{K}}_{1}}}} and {{\color{\colorMATH}\ensuremath{{\mathcal{K}}_{2}}}} with privacy costs of {{\color{\colorMATH}\ensuremath{(\epsilon _{1}, \delta _{1})}}} and {{\color{\colorMATH}\ensuremath{(\epsilon _{2}, \delta _{2})}}} respectively are executed on the same data, the total privacy cost of running both mechanisms is {{\color{\colorMATH}\ensuremath{(\epsilon _{1}+\epsilon _{2}, \delta _{1}+\delta _{2})}}}.
\end{theorem}
For iterative algorithms, \emph{advanced composition}~\cite{dwork2014algorithmic} can yield tighter bounds on total privacy cost. Advanced variants of differential privacy, like R\'enyi differential privacy~\cite{mironov17} and zero-concentrated differential privacy~\cite{bun2016concentrated}, provide even tighter bounds on composition. We discuss composition in detail in Section~\ref{sec:privacy}.

\paragraph{Type Systems for Differential Privacy.}
The first static approach for verifying differential privacy in the context of higher-order programming constructs was \fuzz~\cite{reed2010distance}. \fuzz uses linear types to verify both sensitivity and privacy properties of programs, even in the context of higher-order functions. Conceptual descendents of \fuzz include \dfuzz~\cite{gaboardi2013linear}, Adaptive \fuzz~\cite{Winograd-CortHR17}, \fuzzi~\cite{zhang2019fuzzi}, \duet~\cite{near2019duet}, and the system due to Azevedo de Amorim et al.~\cite{de2019probabilistic}. Approaches based on linear types combine a high degree of automation with support for higher-order programming, but require the host language to support linear types, so none has yet been implemented in a mainstream programming language.

Our work is closest to \dpella~\cite{lobo2020programming}, a Haskell library that uses the Haskell type system for sensitivity analysis. \dpella implements a custom dynamic analysis of programs to compute privacy and accuracy information. \solo goes beyond \dpella by supporting calculation of privacy costs using Haskell's type system, in addition to sensitivity information, and we have formalized its soundness.
See Section~\ref{sec:related} for a complete discussion of related work.

\section{Overview of \solo}
\label{sec:overview}

%\paragraph{Static Analysis via Haskell}

\solo is a static analysis for differential privacy, which can be implemented as a library in Haskell. Its analysis is completely static, and it does not impose any runtime overhead. \solo requires special type annotations, but in many cases these types can be inferred, and typechecking is aided by the flexibilty of parametric polymorphism in Haskell. \solo retains many of the strengths of linear typing approaches to differential privacy, while taking a light-weight approach capable of being embedded in mainstream functional languages. Specifically, \solo:
\begin{enumerate}
  \item is capable of sensitivity analysis for general-purpose programs in the context of higher order programming.
  \item implements a privacy verification approach with separate privacy cost analysis for multiple program inputs using ideas from \duet.
  \item leverages type-level dependency on values via Haskell singleton types, allowing verification of private programs with types that reference symbolic parameters
  \item features verification of several recent variants of differential privacy including {{\color{\colorMATH}\ensuremath{(\epsilon , \delta )}}} and R\'enyi differential privacy.
\end{enumerate}
However, \solo is not intended for the verification of low-level privacy mechanisms such as the core mechanisms described previously, the exponential mechanism~\cite{dwork2014algorithmic}, or the sparse vector technique~\cite{dwork2014algorithmic}.

\paragraph{A Departure From Linear Types.}
Linear types have previously been used to track the consumption of finite resources, such as memory,
in computer programs. They have also seen popular use in differential privacy analysis to track program
sensitivity and the privacy budget expenditure. Linear types are attractive for such applications because
they provide a strategy rooted in type theory and linear logic for tracking resources throughout the
semantics of a core lambda calculus. However, while linear types are a natural fit for differential
privacy analysis, implementations of linear type systems are not commonly available in mainstream
programming languages, and when available are usually not sophisticated enough to support differential privacy analysis.
In order to facilitate an approach to static language-based privacy analysis in mainstream programming
languages, we have chosen to depart from a linear types based strategy, instead favoring an approach similar to static taint analysis.

This design decision has one huge advantage: it \textbf{enables verifying differential privacy in languages without linear types}, such as Haskell. It also brings several drawbacks, outlined below and detailed later in the paper:

\begin{itemize}[leftmargin=5mm, itemsep=4pt]
\item \emph{Functions}: Linear typing provides an explicit type for sensitive functions, indicating
the resource expenditure incurred if the function is called with certain arguments. Without linear
types baked into a programming language, it is usually impossible to annotate function types in
the required manner. However, as we will see later on, it is possible to bypass this limitation using
polymorphism (see Section~\ref{sec:functions}).

\item \emph{Recursion}: In addition to resource tracking for function introduction, linear type systems also provide a
strategy for tracking resource usage during function elimination while accounting for self-referential
functions (recursion). One example of this is a verified implementation of the \inline{map} function.
However, without linear types we must rely on trusted primitives in order to perform looping
on our private programs (see Section~\ref{sec:recursion}).

\item \emph{Decisions \& Branching}: Programs with linear type systems use annotated sum
types and modified typing rules for \inline{case} branching in order to preserve soundness. While working
from the outside, building an analysis system as a library on top of a mainstream language, we are unable
to modify the typing of \inline{case} statements, and instead impose constraints on branching. Specifically,
we disallow branching on sensitive information (which does not restrict the set of private programs
we can write) and a case analysis which returns sensitive information (or a non-deterministic value due to invocation of a privacy mechanism) must have the same sensitivity (or privacy cost) in each \inline{case} alternative (see Section~\ref{sec:conditionals}).

\end{itemize}

% Pros:

% \begin{itemize}

% \item We are able to perform a fully static analysis of differential privacy in mainstream programming
% languages.
% \item Our analysis can account for multiple sensitive program inputs.
% \item Our analysis is fully automated and supports automatic type inference.

% \end{itemize}

\paragraph{The Challenge of Sensitivity Analysis without Linear Types.}
Linear type systems track resources by attaching resource usages to individual program variables in type derivations. Without linear types, program variables are not typically available in function types---so without linear types, \emph{where do we attach sensitivities?}
Previous \emph{dynamic} sensitivity analyses~\cite{mcsherry2009, ebadi2015featherweight, zhang2018ektelo, abuah2021dduo} have attached sensitivities to \emph{values}. This approach works extremely well in a dynamic analysis, where functions are effectively inlined, so higher-order programming is easy to support.

Our static setting is more complicated. We embed sensitivities in base types---the static equivalent of the dynamic strategy of attaching sensitivities to values This approach stands in contrast to the linear-types strategy of embedding sensitivities in function types. A naive implementation of our approach effectively prevents higher-order programming, since it is impossible to give sufficiently general types to sensitive functions. Our solution involves a careful combination of type system features in the implementation language, including:
\begin{enumerate}
\item Type-level parameters to represent sensitivities symbolically
\item Type-level computation to compute symbolic sensitivity expressions
\item Parametric polymorphism to generalize types over sensitivity parameters
\end{enumerate}
Fortunately, recent versions of Haskell support all of these; our approach is also possible in other languages with sufficiently expressive type systems.

\paragraph{Threat Model.}
The threat model for \solo is ``honest but fallible''---that is, we assume the programmer \emph{intends} to write a differentially private program, but may make mistakes. \solo is intended as a tool to help the programmer implement correct differentially private programs in this context. Our approach implements a sound analysis for sensitivity and privacy, but its embedding in a larger system (Haskell) may result in weak points that a malicious programmer could exploit to subvert \solo's guarantees (unsoundness in Haskell's type system, for example). The \solo library can be used with Safe Haskell~\cite{terei2012} to address this issue; \solo exports only a set of safe primitives which are designed to enforce privacy preserving invariants that adhere to our metatheory. However, \solo's protection against malicious programmers are only as strong as the guarantees made by Safe Haskell. Our guarantees against malicious programmers are therefore similar to those provided by language-based information flow control libraries that also utilize Safe Haskell (e.g. \cite{russo2008}).

\paragraph{Soundness.}
We formalize our privacy analysis in terms of a metric preservation metatheory and prove its soundness in Section \ref{sec:formalism} via a step-indexed logical relation {w.r.t.} a step-indexed big-step semantics relation. A consequence of metric preservation is that well-typed pure functions are semantically \emph{sensitive} functions, and that well-typed monadic functions are semantically \emph{differentially private} functions. Our model includes two variants of pair and list type connectives---one sensitive and the other non-sensitive---as well as recursive functions.

\section{Avoiding Linear Types: from \fuzz to \solo}
\label{sec:solo-example}
This section introduces the usage of \solo based on code examples written in our Haskell reference implementation, and compares \solo to related techniques based on linear types.

\paragraph{Sensitivity Analysis.}
Consider the function {{\color{\colorMATH}\ensuremath{\lambda  x \mathrel{:} {\mathbb{R}} .\hspace*{0.33em} x + x}}} from Section~\ref{sec:background}, which is 2-sensitive in its argument {{\color{\colorMATH}\ensuremath{x}}}. The \fuzz language gives this function the type {{\color{\colorMATH}\ensuremath{{\mathbb{R}} \multimap _{2} {\mathbb{R}}}}}, which encodes its sensitivity directly via an annotation on the linear function connective {{\color{\colorMATH}\ensuremath{\multimap }}}. The linear type systems of \fuzz, \dfuzz, \fuzzi, \duet, and Amorim et al. contain typing rules like the following:
\begingroup\color{\colorMATH}\begin{gather*}
\inferrule*[lab={\mtextsc{ t-var}}
]{ s \geq  1
   }{
   \Gamma , x \mathrel{:} _{s} \tau  \vdash  x \mathrel{:} \tau 
}
\hfill\hspace{0pt}\hspace*{1.00em}\hspace*{1.00em}
\inferrule*[lab={\mtextsc{ t-splus}}
]{ \Gamma _{1} \vdash  e_{1} \mathrel{:} {\mathbb{R}}
\\ \Gamma _{2} \vdash  e_{2} \mathrel{:} {\mathbb{R}}
   }{
   \Gamma _{1} + \Gamma _{2} \vdash  e_{1} + e_{2} \mathrel{:} {\mathbb{R}}
}
\hfill\hspace{0pt}\hspace*{1.00em}\hspace*{1.00em}
\inferrule*[lab={\mtextsc{ t-lam}}
]{ \Gamma , x \mathrel{:}_{s} \tau _{1} \vdash  e \mathrel{:} \tau _{2}
   }{
   \Gamma  \vdash  \lambda  x \mathrel{:} \tau _{1} .\hspace*{0.33em} e \mathrel{:} \tau _{1} \multimap _{s} \tau _{2}
}
\end{gather*}\endgroup
The {\mtextsc{ t-var}} rule says that each \emph{use} of a program variable incurs a ``cost'' of 1 to total sensitivity, and the {\mtextsc{ t-lam}} rule translates the sensitivity analysis results on the function's body into a sensitivity annotation on the function type. Here, the context {{\color{\colorMATH}\ensuremath{\Gamma }}} maps program variables to types \emph{and sensitivities}. In linear type systems for differential privacy, the context {{\color{\colorMATH}\ensuremath{\Gamma }}} acts as both a type environment and a \emph{sensitivity environment}. Rules like {\mtextsc{ t-splus}} add together the sensitivity environments of their subexpressions---an operation that sums each variable's sensitivities pointwise (so {{\color{\colorMATH}\ensuremath{\{ x \mathrel{:}_{1} {\mathbb{R}}\}  + \{ x \mathrel{:}_{1} {\mathbb{R}}\}  = \{ x \mathrel{:}_{2} {\mathbb{R}}\} }}}). Using these rules, we can write down the following derivation for the function {{\color{\colorMATH}\ensuremath{\lambda  x \mathrel{:} {\mathbb{R}} .\hspace*{0.33em} x + x}}}:
\begingroup\color{\colorMATH}\begin{gather*}
\inferrule*[lab=
]{ \inferrule*[lab=
   ]{ \{  x\mathrel{:}_{1} {\mathbb{R}} \}  \vdash  x \mathrel{:} {\mathbb{R}}
   \\ \{  x\mathrel{:}_{1} {\mathbb{R}} \}  \vdash  x \mathrel{:} {\mathbb{R}}
      }{
      \{  x\mathrel{:}_{2} {\mathbb{R}} \}  \vdash  x + x \mathrel{:} {\mathbb{R}}
   }
   }{
   \{ \}  \vdash  \lambda  x \mathrel{:} {\mathbb{R}} .\hspace*{0.33em} x + x \mathrel{:} {\mathbb{R}} \multimap _{2} {\mathbb{R}}
}
\end{gather*}\endgroup
Sensitivity tracking in a linear type system is fundamentally linked to sensitivity environments mapping program variables to sensitivities, and is modeled as a co-effect (i.e. sensitivity environments are part of the context {{\color{\colorMATH}\ensuremath{\Gamma }}}).

\paragraph{Sensitivity in \solo.}
In \solo, we instead attach sensitivity environments to \emph{base types}. Our sensitivity environments associate sensitivities with \emph{data sources} (a set of global variables specified by the programmer, detailed in Section~\ref{sec:sens_environments}). For example, we can define a function that doubles its argument as follows:
\begin{minted}{haskell}
dbl :: SDouble 'Diff senv -> SDouble 'Diff (Plus senv senv)
dbl x = x <+> x
\end{minted}
We define \emph{sensitive base types}, like \inline{SDouble}, which augment base types with a distance metric (\S\ref{sec:sens_environments}) and a sensitivity environment (\inline{senv}).
The sensitivity environment \inline{senv} in the type of \inline{dbl} plays the same role as the sensitivity annotations in the context {{\color{\colorMATH}\ensuremath{\Gamma }}} in the linear typing rules above. As in the linear typing rules, the \inline{<+>} function adds the sensitivity environments of its arguments together (\inline{(Plus senv senv)}) at the type level. The \inline{<+>} function is built into \solo, and its type mirrors the {\mtextsc{ t-splus}} rule above:
\begin{minted}{haskell}
(<+>) :: SDouble 'Diff senv1 -> SDouble 'Diff senv2 -> SDouble 'Diff (Plus senv1 senv2)
\end{minted}
Functions in \solo have regular function types ({{\color{\colorMATH}\ensuremath{\tau _{1} \rightarrow  \tau _{2}}}}), and sensitivity environments are attached only to base types. In the absence of polymorphism, this difference leads directly to a significant loss of expressive power: without polymorphism, the type of \inline{dbl} would need to specify exactly what data sources are defined for the program, and how sensitive the function's \emph{input} is with respect to each one. Even \emph{with} polymorphism, there are some functions (like {{\color{\colorMATH}\ensuremath{{{\color{\colorSYNTAX}\mtexttt{map}}}}}}) for which linear-type-based approaches provide more general types. We detail the interaction between polymorphism and sensitivity environments in Section~\ref{sec:functions}.

% \noindent The type of the \inline{<+>} function indicates that it adds the sensitivity environments of its arguments together (\inline{(Plus a b)}), just like the typing rule for addition in \fuzz. The distance metric \inline{Diff} is described in Section~\ref{sec:sensitivity}. We can use \inline{<+>} to define a doubling function, as below. The type signature can be left off, and will be inferred automatically by Haskell. The type signature of \inline{dbl} describes its sensitivity---it is {{\color{\colorMATH}\ensuremath{2}}}-sensitive in its argument.

% \begin{minted}{haskell}
% dbl :: SDouble 'Diff senv -> SDouble 'Diff (Plus senv senv)
% dbl x = x <+> x
% \end{minted}

\paragraph{Privacy Analysis.}
Sensitivity tells us how much noise we need to add to a particular value to achieve the definition of differential privacy. To determine the total privacy cost of a complete program, we need to use the sequential composition property of differential privacy. Languages based on linear types include a language fragment for operations on differentially private values (often in the form of a \emph{privacy monad}), with typing rules like the following:
\begingroup\color{\colorMATH}\begin{gather*}
\inferrule*[lab={\mtextsc{ t-laplace}}
]{ \Gamma  \vdash  e \mathrel{:} \tau 
\\ \Gamma  \sqsubseteq  {}\rceil \Gamma \lceil {}^{s}
   }{
   {}\rceil \Gamma \lceil {}^{\epsilon } \vdash  {{\color{\colorSYNTAX}\mtexttt{laplace}}}[s, \epsilon ](e) \mathrel{:} {\scriptstyle \bigcirc } \tau 
}
\hfill\hspace{0pt}\hspace*{1.00em}\hspace*{1.00em}
\inferrule*[lab={\mtextsc{ t-bind}}
]{ \Gamma _{1} \vdash  e_{1} \mathrel{:} {\scriptstyle \bigcirc } \tau _{1}
\\ \Gamma _{2}, x \mathrel{:}_{\infty } \tau _{1} \vdash  e_{2} \mathrel{:} {\scriptstyle \bigcirc } \tau _{2}
   }{
   \Gamma _{1} + \Gamma _{2} \vdash  x \leftarrow  e_{1} \mathrel{;} e_{2} \mathrel{:} {\scriptstyle \bigcirc } \tau _{2}
}
\end{gather*}\endgroup
The notation {{\color{\colorMATH}\ensuremath{{\scriptstyle \bigcirc }\tau }}} denotes differentially private values. The {\mtextsc{ t-laplace}} rule says that the {{\color{\colorMATH}\ensuremath{{{\color{\colorSYNTAX}\mtexttt{laplace}}}}}} function (\S\ref{sec:background}, Definition~\ref{def:laplace}) satisfies {{\color{\colorMATH}\ensuremath{\epsilon }}}-differential privacy, and returns a differentially private value. {{\color{\colorMATH}\ensuremath{{}\rceil \Gamma \lceil {}^{s}}}} denotes the \emph{truncation} of the context {{\color{\colorMATH}\ensuremath{\Gamma }}} to the sensitivity {{\color{\colorMATH}\ensuremath{s}}} (i.e. replacing every sensitivity in {{\color{\colorMATH}\ensuremath{\Gamma }}} with {{\color{\colorMATH}\ensuremath{s}}}). {{\color{\colorMATH}\ensuremath{\Gamma  \sqsubseteq  {}\rceil \Gamma \lceil {}^{s}}}} encodes the requirement from Definition~\ref{def:laplace} that the argument to the Laplace mechanism must be at most {{\color{\colorMATH}\ensuremath{s}}}-sensitive, and {{\color{\colorMATH}\ensuremath{{}\rceil \Gamma \lceil {}^{\epsilon }}}} replaces each sensitivity in the context with the privacy cost {{\color{\colorMATH}\ensuremath{\epsilon }}}. The {\mtextsc{ t-bind}} rule encodes sequential composition (\S\ref{sec:background}, Theorem~\ref{thm:sequential-composition}), adding up the privacy costs of both computations. For example, the rules above can show that the program {{\color{\colorMATH}\ensuremath{{{\color{\colorSYNTAX}\mtexttt{laplace}}}[2, \epsilon ](x + x)}}} satisfies {{\color{\colorMATH}\ensuremath{\epsilon }}}-differential privacy:
\begingroup\color{\colorMATH}\begin{gather*}
\inferrule*[lab=
]{ \inferrule*[lab=
   ]{ \{  x\mathrel{:}_{1} {\mathbb{R}} \}  \vdash  x \mathrel{:} {\mathbb{R}}
   \\ \{  x\mathrel{:}_{1} {\mathbb{R}} \}  \vdash  x \mathrel{:} {\mathbb{R}}
      }{
      \{  x\mathrel{:}_{2} {\mathbb{R}} \}  \vdash  x + x \mathrel{:} {\mathbb{R}}
   }
   }{
   \{ x \mathrel{:}_{\epsilon } {\mathbb{R}}\}  \vdash  {{\color{\colorSYNTAX}\mtexttt{laplace}}}[2, \epsilon ](x + x) \mathrel{:} {\scriptstyle \bigcirc } {\mathbb{R}}
}
\end{gather*}\endgroup
In \fuzz's privacy monad, the context {{\color{\colorMATH}\ensuremath{\Gamma }}} associates \emph{privacy costs} (rather than sensitivities) with program variables. In our example, the contexts in the first and second rows of the derivation contain sensitivities, while the context in the bottom row contains privacy costs. Linear function types can also encode privacy costs; the Laplace mechanism, for example, can be given the type {{\color{\colorMATH}\ensuremath{{\mathbb{R}} \multimap _{\epsilon } {\scriptstyle \bigcirc }{\mathbb{R}}}}}. Conflating sensitivity and privacy this way works well for pure {{\color{\colorMATH}\ensuremath{\epsilon }}}-differential privacy, but does not work for variants like {{\color{\colorMATH}\ensuremath{(\epsilon ,\delta )}}}-differential privacy; recent linear type systems that support these variants (e.g.~\cite{near2019duet, de2019probabilistic}) are more complex as a result.

\paragraph{Privacy in \solo.}
In \solo, we take the same approach to avoiding linear types for privacy as we did for sensitivity. We attach privacy costs (in the form of \emph{privacy environments}) to monadic values. We define a privacy monad in Haskell (\inline{EpsPrivacyMonad}, detailed in Section~\ref{sec:privacy}) for which the \inline{bind} operator adds privacy environments in the same way as the {\mtextsc{ t-bind}} rule above. We give {{\color{\colorMATH}\ensuremath{{{\color{\colorSYNTAX}\mtexttt{laplace}}}}}} the following type:
\begin{minted}{haskell}
laplace :: forall eps senv m. (TL.KnownNat (MaxSens senv),TL.KnownNat eps) =>
  Proxy eps -> SDouble senv m -> EpsPrivacyMonad (TruncateSens eps senv) Double
\end{minted}
Here, \inline{MaxSens} is a type-level operation corresponding to {{\color{\colorMATH}\ensuremath{\Gamma  \sqsubseteq  {}\rceil \Gamma \lceil {}^{s}}}} in the {\mtextsc{ t-laplace}} rule above (i.e. it ensures the maximum sensitivity of the mechanism's input is {{\color{\colorMATH}\ensuremath{s}}}), and \inline{TruncateSens} is a type-level operation corresponding to {{\color{\colorMATH}\ensuremath{{}\rceil \Gamma \lceil {}^{\epsilon }}}} (i.e. it converts the sensitivity environment \inline{senv} to a privacy environment).
The following function takes a \inline{SDouble} as input, doubles it, and applies the Laplace mechanism:
\begin{minted}{haskell}
simplePrivacyFunction :: SDouble 'Diff '[ '(o, 1) ] -> EpsPrivacyMonad '[ '(o, 2) ] Double
simplePrivacyFunction x = laplace @2 Proxy (dbl x)
\end{minted}
The type \inline{EpsPrivacyMonad '[ '(o, 2) ] Double} indicates that the function satisfies {{\color{\colorMATH}\ensuremath{\epsilon }}}-differential privacy for {{\color{\colorMATH}\ensuremath{\epsilon  = 2}}}, where \inline{o} is a type-level symbol representing a source of sensitive data. As in the previous example, Haskell is able to infer the type if the annotation is left off. Note that the maximum sensitivity of the argument to the Laplace mechanism is automatically calculated (using \inline{MaxSens}), and does not need to be specified by the programmer.

As with sensitivity analysis, we rely heavily on polymorphism to produce general types for functions that guarantee differential privacy (e.g. the \inline{laplace} function). Absent polymorphism, the type for the Laplace mechanism would need to specify an exact set of data sources and a concrete privacy cost associated with each one.

% We can use the Laplace and Gaussian mechanisms introduced in Section~\ref{sec:background} to add noise to sensitive values and satisfy differential privacy. \solo tracks the total privacy cost of multiple uses of these mechanisms using a \emph{privacy monad}, which is similar to the one used in \fuzz and related systems. \solo implements privacy monads for several different privacy variants, with conversion operations between them. These monads are described in detail in Section~\ref{sec:privacy}.

\section{Sensitivity Analysis}
\label{sec:sensitivity}

\begin{figure}

\begin{framed}
\begin{minted}{haskell}
import qualified GHC.TypeLits as TL

-- $\color{black}{\mbox{\textbf{Sources \& Sensitivity Environments} (\S\ref{sec:sens_environments})}}$
type Source = TL.Symbol                               -- sensitive data sources
data Sensitivity = InfSens | NatSens TL.Nat           -- sensitivity values
type SEnv = [(Source, Sensitivity)]                   -- sensitivity environments

-- $\color{black}{\mbox{\textbf{Distance Metrics} (\S\ref{sec:sens_environments})}}$
data NMetric = Diff | Disc                            -- distance metrics for numeric types
SDouble :: NMetric -> SEnv -> *                       -- sensitive doubles

-- $\color{black}{\mbox{\textbf{Pairs} (\S\ref{sec:pairs-lists})}}$
data CMetric = L1 | L2 | LInf                         -- metrics for compound types
SPair    :: CMetric -> (SEnv -> *) -> (SEnv -> *) -> SEnv -> *
L1Pair   = SPair L1                                   -- $\otimes $-pairs in Fuzz
L2Pair   = SPair L2                                   -- Not in Fuzz
LInfPair = SPair LInf                                 -- $\&$-pairs in Fuzz

-- $\color{black}{\mbox{\textbf{Lists} (\S\ref{sec:pairs-lists})}}$
SList    :: CMetric -> (SEnv -> *) -> SEnv -> *       -- sensitive lists
L1List   = SList L1                                   -- $\tau \hspace*{0.33em}{\mtext{list}}$ in Fuzz
L2List   = SList L2                                   -- Not in Fuzz
LInfList = SList LInf                                 -- $\tau \hspace*{0.33em}{\mtext{alist}}$ in Fuzz
\end{minted}
\vspace*{-1em}
\end{framed}

\caption{Sensitivity Types in \solo.}
\label{fig:formal_syntax_types}
\end{figure}

Prior type-based analyses for sensitivity analysis~\cite{reed2010distance, gaboardi2013linear, Winograd-CortHR17, near2019duet} focus on \emph{function sensitivity} with respect to \emph{program variables}.
%In systems like \fuzz~\cite{reed2010distance}, for example, an {{\color{\colorMATH}\ensuremath{s}}}-sensitive function has the type {{\color{\colorMATH}\ensuremath{\tau _{1} \multimap _{s} \tau _{2}}}}.
%
\solo's type system, in contrast, associates sensitivity with \emph{base types} (not functions), and these sensitivities are determined with respect to \emph{data sources} (not program variables).
This difference represents a significant departure from previous systems, and is the key design feature that enables embedding \solo's type system in a language (like Haskell) without linear types.
Figure~\ref{fig:formal_syntax_types} presents the types for the sensitivity analysis in the \solo system. The rest of this section describes types in \solo and how they can be used to describe the sensitivity of a program.
We describe the privacy analysis in Section~\ref{sec:privacy}, and we formalize both analyses in Section~\ref{sec:formalism}.

\subsection{Types, Metrics, and Environments}
\label{sec:sens_environments}

This section describes Figure~\ref{fig:formal_syntax_types} in detail.
We begin with sources (written {{\color{\colorMATH}\ensuremath{o}}}), environments ({{\color{\colorMATH}\ensuremath{\Sigma }}}), metrics ({{\color{\colorMATH}\ensuremath{m}}} and {{\color{\colorMATH}\ensuremath{w}}}), types ({{\color{\colorMATH}\ensuremath{\tau }}}), and sensitive types ({{\color{\colorMATH}\ensuremath{\sigma }}}).

% \paragraph{Data Sources.}
% Our approach makes use of the idea that a static privacy analysis of a program can be centered around a set of sensitive \emph{data sources} which the analyst wants to preserve privacy for. A data source may be represented by some identifier such as a string value, which represents some sensitive program input such as raw numeric data, a file or an IO stream. \solo's data sources are inspired by ideas from static taint analysis---we ``taint'' the program's data sources with sensitivity annotations that are tracked and modified throughout type-checking.

% To enable the static privacy analysis of a program, we track privacy information for data sources at the type-level. Because our analysis is based on the information flow of distinguished data source values throughout a program, we are able to perform a fully static analysis without precise tracking of variable usage within functions, and without a specialized linear type system.
% %
% In \solo, data sources are created for sensitive inputs external to the program. For example, the \inline{readDoubleFromIO} function reads in a sensitive double value from the user:

% \begin{minted}{haskell}
% readDoubleFromIO :: #$\forall $# m o. IO (SDouble m '[ '(o, 1) ])
% \end{minted}

\paragraph{Sources \& Environments.}
Our approach makes use of the idea that a static privacy analysis of a program can be centered around a global set of sensitive \emph{data sources} which the analyst wants to preserve privacy for. Data sources are represented by type-level symbols, each of which represents a single sensitive program input (e.g. raw numeric data, a file or an IO stream). In the most common case, when data is read from a file, the source is identified by the data's filename. \solo's data sources are inspired by ideas from static taint analysis---we ``taint'' the program's data sources with sensitivity annotations that are tracked and modified throughout type-checking.
\solo tracks sensitivity \emph{relative} to data sources (i.e. \solo assumes that data sources have an ``absolute sensitivity'' of 1). In \solo, like in \fuzz, \emph{sensitivities} can be either a number or {{\color{\colorMATH}\ensuremath{\infty }}}. In \solo, numeric sensitivities are represented using type-level natural numbers.
A \emph{sensitivity environment} \inline{SEnv} is an association list of data sources and their sensitivities, and corresponds to the same concept in \fuzz.
% For example, the function {{\color{\colorMATH}\ensuremath{inputSensitiveReal}}} might read a single number from the user, and produce a sensitivity environment of {{\color{\colorMATH}\ensuremath{\Sigma  = \{ o \mapsto  1\} }}} to indicate that the value is 1-sensitive in its input. In the \solo Haskell library, we write this environment \inline{'['(o, 1)]}.

% Sensitive values (whose types are prefixed with "S" such as \inline{SDouble}) represent base values that have been tagged with sensitivity tracking information---specifically, a \emph{sensitivity environment} \inline{&$\sigma $&}. \inline{readDoubleFromIO} instantiates the sensitivity environment to \inline{'[ '(o, 1) ]}, which indicates that values with this type are {{\color{\colorMATH}\ensuremath{1}}}-sensitive with respect to the input read from the sensitive source. The distance metric identifier \inline{m} specifies the metric used to measure distance (as described in Definition~\ref{def:sensitivity}).

% Sensitive values (like \inline{SDouble}) are encapsulated in order to restrict their usage to only privacy preserving operations. The constructors of sensitive data types are hidden, and they are manipulated solely through trusted primitive operations provided in our implementation.

\paragraph{Distance Metrics \& Metric-Carrying Types.}
Interpreting sensitivity requires describing how to measure distances between values (as described in Definition~\ref{def:sensitivity}); different metrics for this measurement produce different privacy properties. \solo provides support for several distance metrics including those commonly used in differentially private algorithms. The \emph{base metrics} listed in Figure~\ref{fig:formal_syntax_types} (\inline{BMetric}) are distance metrics for base types. The \emph{sensitive base types} (\inline{SBase}) are metric-carrying base types (i.e. every sensitive type must have a distance metric). For example, the type of a sensitive \inline{Double} would be \inline{SBase Double m}, where \inline{m} is a metric.
%We use the alias \inline{SDouble} to refer to a sensitive double parameterized by the metric.
The base metrics are \inline{Diff}, the \emph{absolute difference metric}  ({{\color{\colorMATH}\ensuremath{d(x, y) = |x - y|}}}), and \inline{Disc}, the \emph{discrete metric} ({{\color{\colorMATH}\ensuremath{d(x, y) = 0\hspace*{0.33em}{\mtextit{if}}\hspace*{0.33em}x = y; 1\hspace*{0.33em}{\mtextit{otherwise}}}}}). Thus the types \inline{SBase Double Diff} and \inline{SBase Double Disc} mean very different things when interpreting sensitivity. The distance between two values {{\color{\colorMATH}\ensuremath{v_{1}, v_{2} \mathrel{:} \inline{SBase Double Diff}}}} is {{\color{\colorMATH}\ensuremath{|v_{1} - v_{2}|}}}, but the distance between two values {{\color{\colorMATH}\ensuremath{v_{3}, v_{4} \mathrel{:} \inline{SBase Double Disc}}}} is at most 1 (when {{\color{\colorMATH}\ensuremath{v_{3} \neq  v_{4}}}}).

Both of these metrics are useful in writing differentially private programs; basic mechanisms for differential privacy (like the Laplace mechanism) typically require their inputs to use the \inline{Diff} metric, while the distance between program inputs is often described using the \inline{Disc} metric. For example, we might consider a ``database'' of real numbers, each contributed by one individual; two neighboring databases in this setting will differ in exactly one of those numbers, but the change to the number itself may be unbounded. In this case, each number in the database would have the type \inline{SBase Double Disc}. \fuzz fixes the distance metric for numbers to be the absolute difference metric; \duet provides two separate types for real numbers, each with its own distance metric.

% \noindent With these definitions, {{\color{\colorMATH}\ensuremath{{{\color{\colorSYNTAX}\mtexttt{sreal}}}_{{{\color{\colorSYNTAX}\mtexttt{disc}}}}}}} is written \mintinline{haskell}{SReal Disc}, and {{\color{\colorMATH}\ensuremath{{{\color{\colorSYNTAX}\mtexttt{sreal}}}_{{{\color{\colorSYNTAX}\mtexttt{diff}}}}}}} is written \mintinline{haskell}{SReal Diff}.
%
% \solo also provides distance metrics for compound data types (pairs and lists), denoted {{\color{\colorMATH}\ensuremath{w}}}. We describe these metrics in Section~\ref{sec:pairs-lists}.

\paragraph{Types.}
A sensitive type in \solo carries both a metric and a sensitivity environment (e.g. \inline{SBase} has kind \inline{* -> BMetric -> SEnv -> *}). Thus, sensitivities are associated with \emph{values}, rather than with \emph{program variables} (as in \fuzz). For example, the type \inline{SDouble '[ '("sensitive_input", 1) ] 'Diff} from Section~\ref{sec:solo-example} is the type of a double value that is 1-sensitive with respect to the data source \emph{input} under the absolute difference metric. Adding such a value to itself results in the type \inline{SDouble '[ '("sensitive_input", 2) ] 'Diff}---encoding the fact that the sensitivity has doubled. In \fuzz, the same information is encoded by the sensitivities recorded in the context; but with respect to program variables rather than data sources.
Note that it is not possible to attach a sensitivity environment to a function type---\emph{only} the metric-carrying sensitive types may have associated sensitivity environments. \solo does not provide a ``sensitive function'' type connective (like \fuzz's {{\color{\colorMATH}\ensuremath{\multimap }}}); in \solo, function sensitivity must be stated in terms of the sensitivity of the function's arguments with respect to the program's data sources (more in Section~\ref{sec:functions}).

\paragraph{Operations on Sensitivity Environments.}
Section~\ref{sec:solo-example} describes several type-level functions on sensitivity environments in \solo, including \inline{Plus}, \inline{MaxSens}, and \inline{TruncateSens}. We implement these functions in \solo as Haskell type families. For example, the definition of \inline{MaxSens} appears below.
\begin{minted}{haskell}
type family MaxSens (s :: SEnv) :: TL.Nat where
  MaxSens '[] = 0
  MaxSens ('(_,n)':s) = MaxNat n (MaxSens s)
\end{minted}
The other operations are similarly defined as simple recursive functions at the type level, which mimic the mathematical definitions used earlier and in our formalism (\S\ref{sec:formalism}). The definition of \inline{Plus} is slightly more complicated, because \inline{Plus} must find matching sources in its two input environments and add their sensitivities. To make this possible, we ensure that sensitivity environments are ordered by their keys (the symbols representing data sources), and define operations like \inline{Plus} to maintain that ordering. These functions appear in Appendix~\ref{sec:type_level_defs} in the supplemental material.

% type family ScaleSens (n :: TL.Nat) (s :: SEnv) :: SEnv where
%   ScaleSens _ '[] = '[]
%   ScaleSens n1 ('(o,n2) ': s) = '(o,n1 TL.* n2) ': ScaleSens n1 s

% type family TruncateSens (n :: TL.Nat) (s :: SEnv) :: SEnv where
%   TruncateSens _ '[] = '[]
%   TruncateSens n1 ('(o,n2) ': s) = '(o,TruncateNat n1 n2) ': TruncateSens n1 s

% -- compute the sum of two sensitivity environments by traversing each
% -- association list, adding values that have the same key (third equation), and
% -- keeping things in order when keys don't overlap (fourth equation)
% type family (+++) (s1 :: SEnv) (s2 :: SEnv) :: SEnv where
%   '[]            +++ s2             = s2
%   s1             +++ '[]            = s1
%   ('(o,n1)':s1)  +++ ('(o,n2)':s2)  = '(o,n1 TL.+ n2) ': (s1 +++ s2)
%   ('(o1,n1)':s1) +++ ('(o2,n2)':s2) =
%     Cond (IsLT (TL.CmpSymbol o1 o2)) ('(o1,n1) ': (s1 +++ ('(o2,n2)':s2)))
%                                      ('(o2,n2) ': (('(o1,n1)':s1) +++ s2))

% \subsection{Data Sources}

% \todo{add this}

\subsection{Pairs and Lists}
\label{sec:pairs-lists}

The \fuzz system contains two connectives for pairs, {{\color{\colorMATH}\ensuremath{\otimes }}} and {{\color{\colorMATH}\ensuremath{\&}}}, which differ in their metrics. The distance between two {{\color{\colorMATH}\ensuremath{\otimes }}} pairs is the sum of the distances between their elements, while the distance between two {{\color{\colorMATH}\ensuremath{\&}}} pairs is the maximum of distances between their elements.
\solo provides a single pair type, \inline{SPair}, that can express both types by specifying a \emph{compound metric} \inline{CMetric}.

\paragraph{Compound Metrics.}
In \solo, metrics for compound types are derived from standard vector-space distance metrics. For example, a sensitive pair has the type \inline{SPair w} where \inline{w} is one of the compound metrics in Figure~\ref{fig:formal_syntax_types} (\inline{L1}, the {{\color{\colorMATH}\ensuremath{L_{1}}}} (or \emph{Manhattan}) distance; \inline{L2}, the {{\color{\colorMATH}\ensuremath{L_{2}}}} (or \emph{Euclidian}) distance; or \inline{LInf}, the {{\color{\colorMATH}\ensuremath{L_{\infty }}}} distance).
%
% \begin{itemize}
%   \item \inline{L1}, the {{\color{\colorMATH}\ensuremath{L_{1}}}} (or \emph{Manhattan}) distance, is the sum of the distances between corresponding elements: {{\color{\colorMATH}\ensuremath{d(x, y) = \sum \limits_{x_{i} \in  x, y_{i} \in  y} d_{i}(x_{i}, y_{i})}}}.
%   \item \inline{L2}, the {{\color{\colorMATH}\ensuremath{L_{2}}}} (or \emph{Euclidian}) distance, is the square root of the sum of squares of distances between corresponding elements: {{\color{\colorMATH}\ensuremath{d(x, y) = \sqrt {\sum \limits_{x_{i} \in  x, y_{i} \in  y} d_{i}(x_{i}, y_{i})^{2}}}}}
%   \item \inline{LInf}, the {{\color{\colorMATH}\ensuremath{L_{\infty }}}} distance, is the maximum of the distances between corresponding elements: {{\color{\colorMATH}\ensuremath{d(x, y) = \max_{x_{i} \in  x, y_{i} \in  y} d_{i}(x_{i}, y_{i})}}}
% \end{itemize}
%
Thus we can represent \fuzz's {{\color{\colorMATH}\ensuremath{\otimes }}} pairs in \solo using the \inline{SPair L1} type constructor, and \fuzz's {{\color{\colorMATH}\ensuremath{\&}}} pairs using \inline{SPair LInf}. We can construct pairs from sensitive values using the following two functions:
\begin{minted}{haskell}
makeL1Pair :: a m s#$_{1}$# -> b m s#$_{2}$# -> SPair L1 a b (Plus s#$_{1}$# s#$_{2}$#)        -- Fuzz's $\otimes $-pair
makeLInfPair :: a m s#$_{1}$# -> b m s#$_{2}$# -> SPair LInf a b (Join s#$_{1}$# s#$_{2}$#)    -- Fuzz's &-pair
\end{minted}
Here, the \inline{Plus} operator for sensitivity environments performs elementwise addition on sensitivities, and the \inline{Join} operator performs elementwise maximum.

\paragraph{Lists.}
\fuzz defines the list type {{\color{\colorMATH}\ensuremath{\tau \hspace*{0.33em}{{\color{\colorSYNTAX}\mtexttt{list}}}}}}, and gives types to standard operators over lists reflecting their sensitivities. In \solo, we define the \inline{SList} type to represent sensitive lists. Sensitive lists in \solo carry a metric, in the same way as sensitive pairs, and can only contain metric-carrying types. The type of a sensitive list of doubles with the {{\color{\colorMATH}\ensuremath{L_{1}}}} distance metric, for example, is \inline{SList L1 SDouble}; this type corresponds to \fuzz's {{\color{\colorMATH}\ensuremath{\otimes }}}-lists. The type \inline{SList LInf SDouble} corresponds to \fuzz's {{\color{\colorMATH}\ensuremath{\&}}}-lists. \fuzz does not provide the equivalent of \inline{SList L2 SDouble}, which uses the {{\color{\colorMATH}\ensuremath{L_{2}}}} distance metric.

The distance metrics available in \solo are useful for writing practical differentially private programs. For example, we might want to sum up a list of sensitive numbers drawn from a database.
The typical definition of neighboring databases tells us that the distance between two such lists is equal to the number of elements which differ---and those elements may differ by any amount. As a result, their sums may also differ by any amount, and the sensitivity of the computation is unbounded. To address this problem, differentially private programs often \emph{clip} (or ``top-code'') the input data, which enforces an upper bound on input values and results in bounded sensitivity. We can implement this process in a \solo program:
\begin{minted}{haskell}
db   :: L1List (SDouble Disc) '[ '( "input_db", 1 ) ]
clip :: L1List (SDouble Disc) senv -> L1List (SDouble Diff) senv
sum  :: L1List (SDouble Diff) senv -> SDouble Diff senv

summationFunction :: L1List (SDouble Disc) senv -> SDouble Diff senv
summationFunction = sum . clip

summationResult :: SDouble Diff '[ '( "input_db", 1 ) ]
summationResult = summationFunction db
\end{minted}
Here, the \inline{clip} function limits each element of the list to lie between 0 and 1, which allows changing the metric on the underlying \inline{SDouble} from the discrete metric to the absolute difference metric (which is the metric required by the \inline{sum} function). Without the use of \inline{clip} in \inline{summationFunction}, the metrics would not match, and the program would not be well-typed.

% Data clipping is a popular operation to manipulate distance metrics. Using the appropriate distance metric in many cases will lead to optimization of the privacy budget and higher accuracy. For example, the vector-valued Laplace mechanism requires the use of $L1$, while the Gaussian mechanism is compatible with either $L1$ or $L2$. For many applications $L2$ sensitivity is much lower than $L1$, however the Gaussian mechanism introduces a small failure probability.

% Sensitive values in \solo are indexed by a distance metric (usually seen as $m$ with subscripts) as well as a sensitivity environment (usually seen as $s$ with subscripts).

\subsection{Function Sensitivity \& Higher-Order Functions}
\label{sec:functions}

In \fuzz, an {{\color{\colorMATH}\ensuremath{s}}}-sensitive function is given the type {{\color{\colorMATH}\ensuremath{\tau _{1} \multimap _{s} \tau _{2}}}}. \solo does not have sensitive function types, but we have already seen examples of the approach used in \solo to bound function sensitivity: we write function types that are polymorphic over sensitivity environments.
In general, we can recover the notion of an {{\color{\colorMATH}\ensuremath{s}}}-sensitive function in \solo by writing a Haskell function type that scales the sensitivity environment of its input by a scalar $s$:
\begin{minted}{haskell}
s_sensitive :: SDouble senv m -> SDouble (ScaleSens senv s) m   -- An s-sensitive function
\end{minted}
Here, \inline{ScaleSens} is implemented as a type family that scales the sensitivity environment \inline{senv} by \inline{s}: for each mapping {{\color{\colorMATH}\ensuremath{o \mapsto  s_{1}}}} in \inline{senv}, the scaled sensitivity environment contains the mapping {{\color{\colorMATH}\ensuremath{o \mapsto  s\mathord{\cdotp }s_{1}}}}.
The common case of a {{\color{\colorMATH}\ensuremath{1}}}-sensitive (or linear) function can be represented by keeping the input's sensitivity environment unchanged (as in \inline{clip} and \inline{sum} in the previous section):
\begin{minted}{haskell}
one_sensitive :: SDouble senv m -> SDouble senv m               -- A 1-sensitive function
\end{minted}
%
% We may also represent a metric-carrying product or pair value \`a la \fuzz, with the following constructor type:

% As seen in \fuzz we can assign different metrics to product types which may come in useful under certain conditions in sensitivity analysis. For example, we may define the distance between two products to be either the \emph{sum} or \emph{maximum} of their elements.

% \vspace{1em}
% \begin{minted}{haskell}
% makeL1Pair :: (SDouble m s#$_{1}$#, SDouble m s#$_{2}$#) -> SPair L1 SDouble SDouble (s#$_{1}$# + s#$_{2}$#)
% makeLInfPair :: (SDouble m s#$_{1}$#, SDouble m s#$_{2}$#) -> SPair LInf SDouble SDouble (s#$_{1}$# /\ s#$_{2}$#)
% \end{minted}

% As can be seen in the inputs to the pair constructors above, a regular Haskell pair containing sensitive values is capable of containing more sensitivity information, and may be preferable in some applications.
% Conversely, it is always sound to recover a sensitive Haskell pair from a \fuzz style pair.

% \vspace{1em}
% \begin{minted}{haskell}
% recover_fpair :: SPair s m SDouble SDouble -> (SDouble 'AbsoluteM s, SDouble 'AbsoluteM s)
% \end{minted}

% The metric-carrying list is constructed as expected:

% \vspace{1em}
% \begin{minted}{haskell}
% SList :: Metric -> (SEnv -> *) -> SEnv -> *
% \end{minted}

\paragraph{Sensitive Higher-Order Operations}
% As can be seen in Section\ref{sec:solo-example}, operations on \wrapper values are designed to perform the equivalent operation on the underlying values, as well as reconstructing the resulting sensitivity environment on the output value appropriately. The strategy for calculating the new sensitivity environment may vary based on the nature of the operation, value type, and distance metric (discussed later).
An important goal in the design of \solo is support for sensitivity analysis for higher-order, general-purpose programs. For example, prior systems such as \fuzz and \duet encode the type for the higher-order {{\color{\colorMATH}\ensuremath{{\mtext{map}}}}} function as follows:
\begingroup\color{\colorMATH}\begin{gather*}
\begin{array}{rcl
} {\mtext{map}} &{}\mathrel{:}{}& (\tau _{1} \multimap _{s} \tau _{2}) \multimap _{\infty } {\mtext{list}}\hspace*{0.33em}\tau _{1} \multimap _{s} {\mtext{list}}\hspace*{0.33em}\tau _{2}
\end{array}
\end{gather*}\endgroup
This {{\color{\colorMATH}\ensuremath{{\mtext{map}}}}} function describes a computation that accepts as inputs: an {{\color{\colorMATH}\ensuremath{s}}}-sensitive unary function from values of type {{\color{\colorMATH}\ensuremath{\tau _{1}}}} to values of type {{\color{\colorMATH}\ensuremath{\tau _{2}}}}  ({{\color{\colorMATH}\ensuremath{{\mtext{map}}}}} is allowed to apply this function an unlimited number of times), and a list of values of type {{\color{\colorMATH}\ensuremath{\tau _{1}}}}. {{\color{\colorMATH}\ensuremath{{\mtext{map}}}}} returns a list of values of type {{\color{\colorMATH}\ensuremath{\tau _{2}}}} which is {{\color{\colorMATH}\ensuremath{s}}}-sensitive in the former list.
We can give an equivalent type to \inline{map} in \solo as follows, by explicitly scaling the appropriate sensitivity environments using type-level arithmetic:
\begin{minted}{haskell}
map :: #$\forall $# m s s1 a b. (#$\forall $# s'. a s' -> b (s * s')) -> SList m a s1 -> SList m b (s * s1)
\end{minted}

\paragraph{Polymorphism for Sensitive Function Types.}
Special care is needed for functions that close over sensitive values, especially in the context of higher-order functions like \inline{map}. Consider the following example:
% (\lambda  x. \lambda  y. x) z \mathrel{:} \tau _{1} \multimap  (\tau _{2} \multimap  \tau _{1})
\begin{minted}{haskell}
dangerousMap :: SDouble m1 s1 -> SList m2 (SDouble m1) s2 -> _
dangerousMap x ls = let f y = x in map f ls
\end{minted}
Note that \inline{f} is \emph{not} a function that is {{\color{\colorMATH}\ensuremath{s}}}-sensitive with respect to its input---instead, it is {{\color{\colorMATH}\ensuremath{s_{1}}}}-sensitive with respect to the closed-over value of \inline{x}. This use of \inline{map} is dangerous, because it may apply \inline{f} many times, creating duplicate copies of \inline{x} without accounting for the sensitivity effect of this operation. \fuzz assigns an infinite sensitivity for \inline{x} in this program.
\solo rejects this program as not well-typed. The type of \inline{f} is \inline{SDouble m1 s3 -> SDouble m1 s1}, but \inline{map} requires it to have the type \inline{(&$\forall $& s'. a s' -> b (s * s'))}---and these two types do not unify. Specifically, the scope of the sensitivity environment \inline{s'} is limited to \inline{f}'s type---but in the situation above, the environment \inline{s1} comes from \emph{outside} of that scope.

The use of parametric polymorphism to limit the ability of higher-order functions to close over sensitive values is key to our ability to support this kind of programming. Without it, we would not be able to give a type for \inline{map} that ensures soundness of the sensitivity analysis. The use of parametric polymorphism to aid in information flow analysis has been previously noted~\cite{10.1145/2784731.2784733}, and is also key to the treatment of sensitivity in \dpella~\cite{lobo2020programming}.

\subsection{Recursion}
\label{sec:recursion}

In \fuzz, it is possible not only to write the type of {{\color{\colorMATH}\ensuremath{map}}}, but to \emph{infer} it from the definition. The \fuzz type system contains general recursive datatypes, and its typing rules admit recursive programs over those datatypes (like {{\color{\colorMATH}\ensuremath{map}}}).

\solo's sensitive list types are less powerful than \fuzz's. In \solo, it is possible to give types to recursive functions over lists (like \inline{map}, as seen in Section~\ref{sec:functions}). However, it is not possible to typecheck the \emph{implementations} of these functions using \solo's types, since the structure of a sensitive list is opaque to programs written using the \solo library. Hence \inline{map} is provided as a trusted primitive with sound typing.
This restriction is shared by systems like \duet and \dpella; as demonstrated by our case studies in Section~\ref{sec:case}, it is not typically a barrier to writing differentially private programs in practice.

%\todo{is there any possible way to lift this restriction???}

\subsection{Conditionals}
\label{sec:conditionals}

Sensitivity analysis for conditionals requires care, whether or not linear types are used. The primary challenge is that branching on a sensitive value typically implies \emph{infinite} sensitivity with respect to its variables, since the resulting control flow reveals information about the condition. Systems like \fuzz handle this problem using a {\mtextsc{ case}} rule that scales the sensitivity environment used to typecheck condition by the number of times the result is used in the two branches. In practice, this approach often disallows branching on sensitive values, since the distance metric for boolean values says that true and false are infinitely far apart.

In \solo, as in many systems for static information flow control~\cite{myers1999}, we disallow branching on sensitive values altogether. This restriction is implemented implicitly by the opacity of sensitive types (e.g. it is not possible to compare two \inline{SDouble} values and obtain a boolean, so branching on \inline{SDouble} values is impossible). This restriction does not significantly limit the set of differentially private programs that we can write with \solo; except for special mechanisms like the Sparse Vector Technique~\cite{dwork2014algorithmic}, differentially private algorithms generally do not branch on sensitive values anyway, because doing so would violate privacy.

It \emph{is} possible (and desirable) to allow branching on differentially private values (i.e. noisy values). The basic mechanisms in \solo (e.g. \inline{laplace}) return regular Haskell values (e.g. \inline{Double} instead of \inline{SDouble}), and branching on these values can be done in the usual way.

\section{Privacy Analysis}
\label{sec:privacy}

The goal of static privacy analysis is to check that (1) the program adds the correct amount of noise for the sensitivity of underlying computations (i.e. that core mechanisms are used correctly), and (2) the program composes privacy-preserving computations correctly (i.e. the total privacy cost of the program is correct, according to differential privacy's composition properties). A well-typed program should satisfy both conditions. As described earlier, sensitivity analysis often supports privacy analysis, especially in systems based on linear types.

Previous work has taken several approaches to static privacy analysis; we provide a summary in the next section. \solo provides a \emph{privacy monad} that encodes privacy as an effect. As in our sensitivity analysis, the primary difference between \solo and previous work is that our privacy monad tracks privacy cost with respect to data sources, rather than program variables. This distinction allows the implementation of \solo's privacy monad in Haskell, and additionally enables our approach to describe variants of differential privacy without linear group privacy (e.g. {{\color{\colorMATH}\ensuremath{(\epsilon , \delta )}}}-differential privacy).

\subsection{Existing Approaches for Privacy Analysis}

The \fuzz language pioneered static verification of {{\color{\colorMATH}\ensuremath{\epsilon }}}-differential privacy, using a linear type system to track sensitivity of data transformations. In this approach, the linear function space can be interpreted as a space of {{\color{\colorMATH}\ensuremath{\epsilon }}}-differentially private functions by lifting into the probability monad.  However, more advanced variants of differential privacy such as {{\color{\colorMATH}\ensuremath{(\epsilon , \delta )}}} differential privacy do not satisfy the restrictions placed on the interpretation of the linear function space in this approach, and \fuzz cannot be easily extended to support these variants. Azevedo de Amorim et al.~\cite{de2019probabilistic} provide an extensive discussion of this challenge.

More recently, Lobo-Vesga et al. in DPella present an approach in Haskell which tracks sensitivity via data types which are indexed with their \emph{accumulated stability} i.e. sensitivity. Typically in privacy analysis we consider sensitivity to be a property of functions, however as they show, we can also represent sensitivity via the arguments to these functions. Their approach represents private computations via a monad value and monadic operations, similar to the approach in \fuzz. However, in the absence of true linear types, their approach relies on dynamic taint analysis and runtime symbolic execution.

The technique of separating sensitivity composition from privacy composition has been seen before, subsequent to \fuzz, in order to facilitate {{\color{\colorMATH}\ensuremath{(\epsilon , \delta )}}}-differential privacy. Azevedo de Amorim et al.~\cite{de2019probabilistic} introduce a \emph{path construction} technique which performs a \emph{parameterized comonadic lifting} of a metric space layer \`a la \fuzz to a separate relational space layer for {{\color{\colorMATH}\ensuremath{(\epsilon , \delta )}}} differential privacy. The \duet system~\cite{near2019duet} uses a dual type system, with dedicated systems for sensitive composition and privacy composition. In principle, this follows a combined effect/co-effect system approach~\cite{coeffects-thesis}, where one type system tracks the co-effect (in this case sensitivity) and another tracks the effect which is randomness due to privacy.

Our approach embodies the spirit of \duet and simulates coeffectful program behavior by embedding the co-effect (i.e. the entire sensitivity environment) as an index in comonadic base data types. We then track privacy composition via a special monadic type as an effect. As in \duet, the core privacy mechanisms such as Laplace and Gauss police the boundary between the two. Due to the nature of our co-effect oriented approach in which we track the full sensitivity context, our solution can be embedded in Haskell completely statically, without the need for runtime dynamic symbolic execution. We are also able to verify advanced privacy variants such as {{\color{\colorMATH}\ensuremath{(\epsilon , \delta )}}} and state-of-the-art composition theorems such as advanced composition and the moments accountant via a family of higher-order primitives.
%In addition to this, we are able to track the accuracy of data sources fully statically, as opposed to using a mixed dynamic/static approach as seen in prior work.

\paragraph{Monads \& Effect Systems.}
Effect systems are known for providing more detailed static type information than possible with monadic typing. They are the topic of a variety of research on enhancing monadic types with program effect information, in order to provide stricter static guarantees. Orchard et al \cite{orchard2014}, following up on initial work by Wadler and Thiemann \cite{wadler2003}, provide a denotational semantics which unify effect systems with a monadic-style semantics as an \emph{parametric effect monad}, establishing an isomorphism between indices of the denotations and the effect annotations of traditional effect systems. They present a formulation of parametric effect monads which generalize monads to include annotation of an effect with a strict monoidal structure. Below typing rules of the general parametric effect monad are shown:
\begingroup\color{\colorMATH}\begin{gather*}
\begin{tabularx}{\linewidth}{>{\centering\arraybackslash\(}X<{\)}}
  \hfill\hspace{0pt}
  \inferrule*[lab={\mtextsc{ Bind}}
  ]{ f \mathrel{:} a \rightarrow  {\begingroup\renewcommand\colorMATH{\colorMATHB}\renewcommand\colorSYNTAX{\colorSYNTAXB}{{\color{\colorMATH}\ensuremath{M}}}\endgroup }\hspace*{0.33em}{\begingroup\renewcommand\colorMATH{\colorMATHA}\renewcommand\colorSYNTAX{\colorSYNTAXA}{{\color{\colorMATH}\ensuremath{r}}}\endgroup }\hspace*{0.33em}b
  \\ g \mathrel{:} b \rightarrow  {\begingroup\renewcommand\colorMATH{\colorMATHB}\renewcommand\colorSYNTAX{\colorSYNTAXB}{{\color{\colorMATH}\ensuremath{M}}}\endgroup }\hspace*{0.33em}{\begingroup\renewcommand\colorMATH{\colorMATHA}\renewcommand\colorSYNTAX{\colorSYNTAXA}{{\color{\colorMATH}\ensuremath{s}}}\endgroup }\hspace*{0.33em}c
     }{
     \lambda x.\hspace*{0.33em}f\hspace*{0.33em}x >>= g \mathrel{:} a \rightarrow  {\begingroup\renewcommand\colorMATH{\colorMATHB}\renewcommand\colorSYNTAX{\colorSYNTAXB}{{\color{\colorMATH}\ensuremath{M}}}\endgroup }\hspace*{0.33em}({\begingroup\renewcommand\colorMATH{\colorMATHA}\renewcommand\colorSYNTAX{\colorSYNTAXA}{{\color{\colorMATH}\ensuremath{r}}}\endgroup }\otimes {\begingroup\renewcommand\colorMATH{\colorMATHA}\renewcommand\colorSYNTAX{\colorSYNTAXA}{{\color{\colorMATH}\ensuremath{s}}}\endgroup })\hspace*{0.33em}c
  }
  \hfill\hspace{0pt}
  \inferrule*[lab={\mtextsc{ Return}}
  ]{
     }{
     return \mathrel{:} a \rightarrow  {\begingroup\renewcommand\colorMATH{\colorMATHB}\renewcommand\colorSYNTAX{\colorSYNTAXB}{{\color{\colorMATH}\ensuremath{M}}}\endgroup }\hspace*{0.33em}\varnothing \hspace*{0.33em}a
  }
  \hfill\hspace{0pt}
\end{tabularx}
\end{gather*}\endgroup
\noindent These typing rules describe a formulation of parametric effect monads $M$ which accept an effect index as their first argument. This effect index of some arbitrary type E is a monoid $(E,\otimes ,\varnothing )$.

\subsection{\solo's Privacy Monad}

\solo defines \emph{privacy environments} in the same way as sensitivity environments; instead of tracking a sensitivity with respect to each of the program's data sources, however, a privacy environment tracks a \emph{privacy cost} associated with each data source. Privacy environments for pure {{\color{\colorMATH}\ensuremath{\epsilon }}}-differential privacy are defined as follows:
\begin{minted}{haskell}
-- $\color{black}{\mbox{\textbf{Privacy Environments}}}$
data EpsPrivacyCost = InfEps | EpsCost TLRat    -- values for $\epsilon $
type EpsPrivEnv = [(Source, EpsPrivacyCost)]    -- privacy environments, $\epsilon $-differential privacy
\end{minted}
\inline{TLRat} is a type-level encoding of positive rational numbers by a pair of the numerator and denominator as natural numbers in GCD-reduced form.

%\paragraph{Privacy as an Effect in \solo}
The \emph{sequential composition} theorem for differential privacy (Theorem \ref{thm:sequential-composition}) says that when sequencing {{\color{\colorMATH}\ensuremath{\epsilon }}}-differentially private computations, we can add up their privacy costs. This theorem provides the basis for the definition of a privacy monad.
We observe that our privacy environments have a monoidal structure \inline{(EpsPrivEnv, EpsSeqComp, '[])}, where \inline{EpsSeqComp} is a type family implementing the sequential composition theorem. We derive a privacy monad which is indexed by our privacy environments, in the same style as a notion of effectful monads or \emph{parametric effect monads} given separately by Orchard \cite{orchard2014, orchard2014A} and Katsumata \cite{katsumata2014}. Computations of type \inline{PrivacyMonad} are constructing via these core functions:
\begin{minted}{haskell}
-- $\color{black}{\mbox{\textbf{Privacy Monad for {{\color{\colorMATH}\ensuremath{\epsilon }}}-differential privacy}}}$
return :: a -> EpsPrivacyMonad '[] a
(>>=) :: EpsPrivacyMonad p#$_{1}$# a -> (a -> EpsPrivacyMonad p#$_{2}$# b) -> EpsPrivacyMonad (EpsSeqComp p#$_{1}$# p#$_{2}$#) b
\end{minted}
The \inline{return} operation accepts some value and embeds it in the \inline{PrivacyMonad} without causing any side-effects. The \inline{(>>=)} (\inline{bind}) operation allows us to sequence private computations using differential privacy's sequential composition property, encoded here as the type family \inline{EpsSeqComp}. The implementation of \inline{EpsSeqComp} performs elementwise summation of two privacy environments. In the computation \inline{f>>=g} we execute the private computation \inline{f} for some polymorphic privacy cost \inline{p&$_{1}$&}, pass its result to the private computation \inline{g}, and output the result of \inline{g} at a total privacy cost of the degradation of the \inline{p&$_{1}$&} and \inline{p&$_{2}$&} privacy environments combined according to sequential composition. Note that while \inline{PrivacyMonad} is not a regular monad in Haskell (due to the extra index in its type) we may still make use of \inline{do}-notation in our examples by using Haskell's \inline{RebindableSyntax} language extension.

Note that \inline{return} in \solo's privacy monad is very different from the same operator in \fuzz. The typing rule for {{\color{\colorMATH}\ensuremath{{{\color{\colorSYNTAX}\mtexttt{return}}}}}} in \fuzz scales the sensitivities in the context by {{\color{\colorMATH}\ensuremath{\infty }}}---reflecting the idea that {{\color{\colorMATH}\ensuremath{{{\color{\colorSYNTAX}\mtexttt{return}}}}}}'s argument is revealed with no added noise, incurring infinite privacy cost. However, this definition of {{\color{\colorMATH}\ensuremath{{{\color{\colorSYNTAX}\mtexttt{return}}}}}} does not satisfy the monad laws. The \fuzz privacy monad is described as ``monad-like,'' and is intentionally designed not to satisfy the laws. For example, in \fuzz, {{\color{\colorMATH}\ensuremath{{{\color{\colorSYNTAX}\mtexttt{return}}}\hspace*{0.33em}x \gg = {{\color{\colorSYNTAX}\mtexttt{laplace}}} \neq  {{\color{\colorSYNTAX}\mtexttt{laplace}}}\hspace*{0.33em}x}}}.
%
% {\small \begingroup\color{\colorMATH}\begin{gather*} {{\color{\colorSYNTAX}\mtexttt{return}}}\hspace*{0.33em}x \gg = {{\color{\colorSYNTAX}\mtexttt{laplace}}}\hspace*{1.00em}\hspace*{1.00em} \neq  \hspace*{1.00em}\hspace*{1.00em}{{\color{\colorSYNTAX}\mtexttt{laplace}}}\hspace*{0.33em}x
% \end{gather*}\endgroup}
%
In \solo, the \inline{return} operator attaches an empty privacy environment to the returned value, and does satisfy the monad laws. If a sensitive value is given as the argument to \inline{return}, then \emph{it remains sensitive}, rather than being revealed (as in \fuzz)---so there is no need to assign the value an infinite privacy cost. This approach is not feasible in \fuzz because privacy costs are associated with program variables rather than with values. We can recover \fuzz's {{\color{\colorMATH}\ensuremath{{{\color{\colorSYNTAX}\mtexttt{return}}}}}} behavior (revealing a value without noise, and scaling its privacy cost by infinity) using a \inline{reveal} function with the following type:
\begin{minted}{haskell}
reveal :: SDouble m senv -> EpsPrivacyMonad (ScaleToInfinity senv) Double
\end{minted}

\paragraph{Core Privacy Mechanisms.}
We can define core privacy mechanisms like the Laplace mechanism (described in Section~\ref{sec:background}), which satisfies {{\color{\colorMATH}\ensuremath{\epsilon }}}-differential privacy:
\begin{minted}{haskell}
laplace     :: Proxy #$\epsilon $# -> SDouble s Diff -> EpsPrivacyMonad (TruncateSens #$\epsilon $# s) Double
listLaplace :: Proxy #$\epsilon $# -> L1List (SDouble Diff) s -> EpsPrivacyMonad (TruncateSens #$\epsilon $# s) [Double]
\end{minted}
The first argument to \inline{laplace} is the privacy parameter {{\color{\colorMATH}\ensuremath{\epsilon }}} (as a type-level natural). The second argument is the value we would like to add noise to; it must be a sensitive number with the \inline{Diff} metric. The function's result is a regular Haskell \inline{Double}, in the privacy monad. The \inline{TruncateSens} type family transforms a sensitivity environment into a privacy environment by replacing each sensitivity with the privacy parameter {{\color{\colorMATH}\ensuremath{\epsilon }}}. The function's implementation follows the definition of the Laplace mechanism; it determines the scale of the noise to add using the maximum sensitivity in the sensitivity environment \inline{s} and the privacy parameter {{\color{\colorMATH}\ensuremath{\epsilon }}}.

The \inline{listLaplace} function implements the \emph{vector-valued Laplace mechanism}, which adds noise to each element of a vector based on the vector's {{\color{\colorMATH}\ensuremath{L_{1}}}} sensitivity. Its argument is required to be a \inline{L1List} of sensitive doubles with the \inline{Diff} metric, and its output is a list of Haskell doubles in the privacy monad.
As a simple example, the following function adds noise to its input twice, once with {{\color{\colorMATH}\ensuremath{\epsilon  = 2}}} and once with {{\color{\colorMATH}\ensuremath{\epsilon  = 3}}}, for a total privacy cost of {{\color{\colorMATH}\ensuremath{\epsilon  = 5}}}. If the type annotation is left off, Haskell infers this type.
\begin{minted}{haskell}
addNoiseTwice :: TL.KnownNat (MaxSens s) =>
  SDouble s Diff -> EpsPrivacyMonad (Plus (TruncateSens 2 s) (TruncateSens 3 s)) Double
addNoiseTwice x = do
  y#$_{1}$# <- laplace @2 Proxy x
  y#$_{2}$# <- laplace @3 Proxy x
  return $ y#$_{1}$# + y#$_{2}$#
\end{minted}

\subsection{{{\color{\colorMATH}\ensuremath{(\epsilon , \delta )}}}-Differential Privacy \& Advanced Composition}

The \emph{advanced composition theorem} for differential privacy~\cite{dwork2014algorithmic} provides tighter bounds on the privacy cost of iterative algorithms, but requires the use of {{\color{\colorMATH}\ensuremath{(\epsilon , \delta )}}}-differential privacy.
\begin{theorem}[Advanced composition]
  For {{\color{\colorMATH}\ensuremath{0 < \epsilon ^{\prime} < 1}}} and {{\color{\colorMATH}\ensuremath{\delta ^{\prime} > 0}}}, the class of {{\color{\colorMATH}\ensuremath{(\epsilon , \delta )}}}-differentially private mechanisms satisfies {{\color{\colorMATH}\ensuremath{(\epsilon ^{\prime}, k\delta  + \delta ^{\prime})}}}-differential privacy under {{\color{\colorMATH}\ensuremath{k}}}-fold adaptive composition for:
\begingroup\color{\colorMATH}\begin{gather*} \epsilon ^{\prime} = 2 \epsilon  \sqrt {2k\ln (1/\delta ^{\prime})}
\end{gather*}\endgroup
\end{theorem}
\vspace*{-2mm}
\noindent To support advanced composition in \solo, we first define privacy environments and a privacy monad for {{\color{\colorMATH}\ensuremath{(\epsilon , \delta )}}}-differential privacy as follows:
%
%-- $\color{black}{\mbox{\textbf{Privacy Environments \& Monad for {{\color{\colorMATH}\ensuremath{\hspace*{0.33em}(\epsilon , \delta )}}}-differential privacy}}}$
\begin{minted}{haskell}
data EDPrivacyCost = InfED | EDCost TLReal TLReal
type EDEnv = [(TL.Symbol, EDPrivacyCost)]
return :: a -> EDPrivacyMonad '[] a
(>>=) :: EDPrivacyMonad p#$_{1}$# a -> (a -> EDPrivacyMonad p#$_{2}$# b) -> EDPrivacyMonad (EDSeqComp p#$_{1}$# p#$_{2}$#) b
\end{minted}

\noindent where \inline{EDSeqComp} is a type family that implements sequential composition for {{\color{\colorMATH}\ensuremath{(\epsilon , \delta )}}}-differential privacy (Theorem \ref{thm:sequential-composition}) via elementwise summation of both {{\color{\colorMATH}\ensuremath{\epsilon }}} and {{\color{\colorMATH}\ensuremath{\delta }}} values.
Rational numbers were sufficient to represent privacy costs in pure {{\color{\colorMATH}\ensuremath{\epsilon }}}-differential privacy, but we use a type-level representation of real numbers (\inline{TLReal}) for {{\color{\colorMATH}\ensuremath{(\epsilon , \delta )}}}-differential privacy. For advanced composition, we will need operations like square root and natural logarithm. Haskell avoids supporting doubles at the type level, because equality for doubles does not interact well with the notion of equality required for typing. We therefore implement \inline{TLReal} by building type-level expressions that represent real-valued computations, and interpret those expressions using Haskell's standard double type at the value level.
We can now write the type of a looping combinator for advanced composition:
\begin{minted}{haskell}
advloop :: NatS k -> a -> (a -> EDPrivacyMonad p a) -> EDPrivacyMonad (AdvComp k #$\delta ^{\prime}$# p) a
\end{minted}
%
% \begin{minted}{haskell}
% let x = <sensitive thing>
%     y = <other sensitive thing>
% in advloop 0 (\lambda  y \Rightarrow  laplace x)
% \end{minted}
%
The looping combinator \inline{advloop} is designed to run an {{\color{\colorMATH}\ensuremath{(\epsilon , \delta )}}}-differentially private mechanism {{\color{\colorMATH}\ensuremath{k}}} times, and satisfies {{\color{\colorMATH}\ensuremath{(2\epsilon \sqrt {2k\ln ({1}/{\delta ^{\prime}})}, \delta ^{\prime}{+}k\delta )}}}-differential privacy---which is significantly lower than the standard composition theorem when {{\color{\colorMATH}\ensuremath{k}}} is large. The first argument \inline{k} is the statically known number of iterations. The type family \inline{AdvComp} is a helper to statically compute the appropriate total privacy cost given the privacy parameters of the private function passed as the penultimate parameter to the primitive which satisfies {{\color{\colorMATH}\ensuremath{(\epsilon ,\delta )}}}-differential privacy. \inline{AdvComp} builds a type-level expression containing square roots and logarithms, as described earlier.

% Complex sensitivity and privacy quantities as seen in the type of \inline{advloop} utilize type-level symbolic representations for static tracking such as the symbolic $log$ and $exp$ functions. We represent rational numbers at the type-level as pairs of natural numbers.

% Advanced composition is particularly useful in programs that run many differentially private functions in sequence over high-dimensional data, as is typical in machine learning algorithms. We express the law of advanced composition via the \inline{EDPrivacyMonad}, one of the variants of differential privacy that \solo supports.

\paragraph{The Gaussian Mechanism.}
The Gaussian mechanism (described in Section~\ref{sec:background}) adds Gaussian noise instead of Laplace noise, and ensures {{\color{\colorMATH}\ensuremath{(\epsilon , \delta )}}}-differential privacy (with {{\color{\colorMATH}\ensuremath{\delta  > 0}}}). The primary advantage of the Gaussian mechanism is in the vector setting: the Gaussian mechanism uses {{\color{\colorMATH}\ensuremath{L_{2}}}} sensitivity, which is typically much lower than the {{\color{\colorMATH}\ensuremath{L_{1}}}} sensitivity used by the Laplace mechanism. This requirement is reflected in the type of the Gaussian mechanism in \solo:
\begin{minted}{haskell}
gauss     :: Proxy #$\epsilon $# -> Proxy #$\delta $# -> SDouble s Diff -> EDPrivacyMonad (TruncateSensED #$\epsilon $# #$\delta $# s) Double
listGauss :: Proxy #$\epsilon $# -> Proxy #$\delta $# -> L2List (SDouble Diff) s
                         -> EDPrivacyMonad (TruncateSensED #$\epsilon $# #$\delta $# s) [Double]
\end{minted}

\subsection{Additional Variants \& Converting Between Variants}

\solo provides a type class of privacy monads instantiated for each supported variant of differential privacy. For each privacy variant, the corresponding privacy monad is indexed with a privacy environment that tracks the appropriate privacy parameters, and the bind operation enforces the appropriate form of sequential composition. Conversion operations are provided between variants to enable variant-mixing in programs. For example, the following function converts an {{\color{\colorMATH}\ensuremath{\epsilon }}}-differentially private computation into an {{\color{\colorMATH}\ensuremath{(\epsilon , \delta )}}}-differentially private one, setting {{\color{\colorMATH}\ensuremath{\delta =0}}}.
\begin{minted}{haskell}
conv_eps_to_ed :: EpsPrivacyMonad p1 a -> EDPrivacyMonad (ConvEpstoED p1) a
\end{minted}
%
% \begin{minted}{haskell}
% -- interactive conversion example
% x = conv_eps_to_ed applied#$_{3}$#
% :t x
% -- x :: EDPrivacyMonad '[ '("sensitive_input",5,1)] Double
% \end{minted}
%
\solo currently supports {{\color{\colorMATH}\ensuremath{\epsilon }}}-differential privacy, {{\color{\colorMATH}\ensuremath{(\epsilon , \delta )}}}-differential privacy, and R\'enyi differential privacy (RDP)~\cite{mironov17}. Conversions are possible from {{\color{\colorMATH}\ensuremath{\epsilon }}}-DP to {{\color{\colorMATH}\ensuremath{(\epsilon , \delta )}}}-DP and RDP, and from RDP to {{\color{\colorMATH}\ensuremath{(\epsilon , \delta )}}}-DP. Conversions are not possible from {{\color{\colorMATH}\ensuremath{(\epsilon , \delta )}}}-DP to {{\color{\colorMATH}\ensuremath{\epsilon }}}-DP or RDP.

\section{Formalism}
\label{sec:formalism}
In \solo, we implement a novel static analysis for function sensitivity and differential privacy. Our approach can be seen as a type-and-effect system, which may be embedded in statically typed functional languages with support for monads and type-level arithmetic.

\paragraph{Program Syntax.}
Figure~\ref{fig:syntax} shows a core subset of the syntax for our analysis system.
Our language model includes arithmetic operations ({{\color{\colorMATH}\ensuremath{e \odot  e}}}), pairs ({{\color{\colorMATH}\ensuremath{\langle e,e\rangle }}} and {{\color{\colorMATH}\ensuremath{\pi _{i}(e)}}}), conditionals ({{\color{\colorMATH}\ensuremath{{{\color{\colorSYNTAX}\mtexttt{if0}}}(e)\{ e\} \{ e\} }}}), and functions ({{\color{\colorMATH}\ensuremath{\lambda _{x}x.\hspace*{0.33em}e}}} and {{\color{\colorMATH}\ensuremath{e(e)}}}). Types {{\color{\colorMATH}\ensuremath{\tau }}} presented in the formalism include: base numeric types {{\color{\colorMATH}\ensuremath{{{\color{\colorSYNTAX}\mtexttt{real}}}}}}, singleton numeric types with a known runtime value at compile-time {{\color{\colorMATH}\ensuremath{{{\color{\colorSYNTAX}\mtexttt{real}}}[r]}}}, booleans {{\color{\colorMATH}\ensuremath{{{\color{\colorSYNTAX}\mtexttt{bool}}}}}}, functions {{\color{\colorMATH}\ensuremath{\tau  \rightarrow  \tau }}}, pairs {{\color{\colorMATH}\ensuremath{\tau \times \tau }}}, and the privacy monad {{\color{\colorMATH}\ensuremath{{\scriptstyle \bigcirc }_{\Sigma }(\tau )}}}. Regular types {{\color{\colorMATH}\ensuremath{\tau }}} are accompanied by sensitive types {{\color{\colorMATH}\ensuremath{\sigma }}} which are essentially regular types annotated with static sensitivity analysis information {{\color{\colorMATH}\ensuremath{\Sigma }}}---which is the sensitivity analysis (or sensitivity environment) for the expression which was typed as {{\color{\colorMATH}\ensuremath{\tau }}}.
Senstitive types shown in our formalism include sensitive numeric types {{\color{\colorMATH}\ensuremath{{{\color{\colorSYNTAX}\mtexttt{sreal}}}}}}, sensititve pairs {{\color{\colorMATH}\ensuremath{\sigma  \otimes  \sigma }}}, and sensitive lists {{\color{\colorMATH}\ensuremath{{{\color{\colorSYNTAX}\mtexttt{slist}}}(\sigma )}}}. A metric-carrying singleton numeric type is unnecessary since its value is fixed and cannot vary.
{{\color{\colorMATH}\ensuremath{\Sigma }}}---the {\mtextit{sensitivity/privacy environment}}---is defined as a mapping from sensitive sources {{\color{\colorMATH}\ensuremath{o \in  {\mtext{source}}}}} to scalar values which represent the sensitivity/privacy of the resulting value with respect to that source.

Types/values with standard treatment are not shown in our formalism, but included in our implementation with both regular and metric-carrying versions, include vectors and matrices which have known dimensions at compile-time via singleton natural number indices. Single natural numbers are also used to execute loops with statically known number of iterations and to help contruct sensitivity and privacy quantities.

% {-{ FIGURE: Syntax
\begingroup
\renewcommand\b[1]{\colorbox{blue!10}{{{\color{\colorMATH}\ensuremath{#1}}}}}
\begin{figure}
{\small
\begin{framed}
\begingroup\color{\colorMATH}\begin{gather*}
\begin{tabularx}{\linewidth}{>{\centering\arraybackslash\(}X<{\)}}\hfill\hspace{0pt} b \in  {\mathbb{B}} \hfill\hspace{0pt} r \in  {\mathbb{R}} \hfill\hspace{0pt} \dot r \in  {\mathbb{R}} \mathrel{\Coloneqq } r \mathrel{|} \infty  \hfill\hspace{0pt} x,z \in  {\mtext{var}} \hfill\hspace{0pt} \b{o \in  {\mtext{source}}} \hfill\hspace{0pt}
\cr 
\cr \hfill\hspace{0pt}
  \begin{array}{rclcl@{\hspace*{1.00em}}l
  } \Sigma    &{}\in {}& {\mtext{spenv}}  &{}\triangleq {}& {\mtext{source}} \rightharpoonup  \dot {\mathbb{R}}                               & {{\color{\colorTEXT}\textnormal{sensitivity/privacy environment}}}
  \cr  \tau    &{}\in {}& {\mtext{type}}  &{}\mathrel{\Coloneqq }{}& {{\color{\colorSYNTAX}\mtexttt{bool}}} \mathrel{|}{{\color{\colorSYNTAX}\mtexttt{real}}} \mathrel{|} {{\color{\colorSYNTAX}\mtexttt{real}}}[r]                  & {{\color{\colorTEXT}\textnormal{base and singleton types}}}
  \cr      &{} {}&         &{}\mathrel{|}{}& \tau \times \tau  \mathrel{|} {{\color{\colorSYNTAX}\mtexttt{list}}}(\tau ) \mathrel{|} \tau \rightarrow \tau                        & {{\color{\colorTEXT}\textnormal{connectives}}}
  \cr      &{} {}&         &{}\mathrel{|}{}& {\scriptstyle \bigcirc }_{\Sigma }(\tau ) \mathrel{|} \b{\sigma @\Sigma }                           & {{\color{\colorTEXT}\textnormal{privacy monad and \b{{{\color{\colorTEXT}\textnormal{sensitive types}}}}}}}
  \cr  \b{\sigma } &{}\in {}& \b{{\mtext{stype}}} &{}\mathrel{\Coloneqq }{}& \b{{{\color{\colorSYNTAX}\mtexttt{sreal}}} \mathrel{|} \sigma  \otimes  \sigma  \mathrel{|} {{\color{\colorSYNTAX}\mtexttt{slist}}}(\sigma )}      & \b{{{\color{\colorTEXT}\textnormal{sensitive types}}}}
  \cr  {\odot } &{}\in {}& {\mtext{binop}}       &{}\mathrel{\Coloneqq }{}& {+} \mathrel{|} {\times } \mathrel{|} {\rtimes }                       & {{\color{\colorTEXT}\textnormal{operations}}}
  \cr  e   &{}\in {}& {\mtext{expr}}        &{}\mathrel{\Coloneqq }{}& x \mathrel{|} b \mathrel{|} r \mathrel{|} {{\color{\colorSYNTAX}\mtexttt{sing}}}(r)                 & {{\color{\colorTEXT}\textnormal{variables and literals}}}
  \cr      &{} {}&               &{}\mathrel{|}{}& e \odot  e \mathrel{|} {{\color{\colorSYNTAX}\mtexttt{if}}}(e)\{ e\} \{ e\}                  & {{\color{\colorTEXT}\textnormal{binary operations and conditionals}}}
  \cr      &{} {}&               &{}\mathrel{|}{}& \langle e,e\rangle  \mathrel{|} \pi _{i}(e)                         & {{\color{\colorTEXT}\textnormal{pair creation and access}}}
  \cr      &{} {}&               &{}\mathrel{|}{}& [] \mathrel{|} e\mathrel{:: }e                              & {{\color{\colorTEXT}\textnormal{list creation}}}
  \cr      &{} {}&               &{}\mathrel{|}{}& {{\color{\colorSYNTAX}\mtexttt{case}}}(e)\{ [].e\} \{ x\mathrel{:: }x.e\}                 & {{\color{\colorTEXT}\textnormal{list destruction}}}
  \cr      &{} {}&               &{}\mathrel{|}{}& \lambda _{x}x.\hspace*{0.33em}e \mathrel{|} e(e)                        & {{\color{\colorTEXT}\textnormal{recursive functions}}}
  \cr      &{} {}&               &{}\mathrel{|}{}& \b{{{\color{\colorSYNTAX}\mtexttt{reveal}}}(e)} \mathrel{|} {{\color{\colorSYNTAX}\mtexttt{laplace}}}[e,e](e)   & {{\color{\colorTEXT}\textnormal{privacy operations}}}
  \cr      &{} {}&               &{}\mathrel{|}{}& {{\color{\colorSYNTAX}\mtexttt{return}}}(e) \mathrel{|} x \leftarrow  e \mathrel{;} e               & {{\color{\colorTEXT}\textnormal{privacy monad}}}
  \cr      &{} {}&               &{}\mathrel{|}{}& \b{\hat \langle e,e\hat \rangle  \mathrel{|} \hat \pi _{i}(e)}                     & \b{ {{\color{\colorTEXT}\textnormal{sensitive pair creation and access}}} }
  \cr      &{} {}&               &{}\mathrel{|}{}& \b{\hat {[}\hat {]} \mathrel{|} e\mathrel{\hat {\mathrel{:: }}}e}                          & \b{ {{\color{\colorTEXT}\textnormal{sensitive list creation}}} }
  \cr      &{} {}&               &{}\mathrel{|}{}& \b{{{\color{\colorSYNTAX}\mtexttt{case}}}(e)\{ \hat {[}\hat {]}.e\} \{ x\mathrel{\hat {\mathrel{:: }}}x.e\} }            & \b{ {{\color{\colorTEXT}\textnormal{sensitive list destruction}}} }
  \cr  \gamma    &{}\in {}& {\mtext{venv}}        &{}\triangleq {}& {\mtext{var}} \rightharpoonup  {\mtext{value}}                       & {{\color{\colorTEXT}\textnormal{evaluation environment}}}
  \cr  \rho    &{}\in {}& {\mtext{ddist}}       &{}\triangleq {}& \left\{ f \in  {\mtext{value}} \rightarrow  {\mathbb{R}} \mathrel{|} \sum \limits_{v} f(v) = 1\right\}      & {{\color{\colorTEXT}\textnormal{discrete distributions (PMF)}}}
  \cr  v   &{}\in {}& {\mtext{value}}       &{}\mathrel{\Coloneqq }{}& b \mathrel{|} r                                 & {{\color{\colorTEXT}\textnormal{literals}}}
  \cr      &{} {}&               &{}\mathrel{|}{}& \langle v,v\rangle                                  & {{\color{\colorTEXT}\textnormal{pairs}}}
  \cr      &{} {}&               &{}\mathrel{|}{}& [] \mathrel{|} v\mathrel{:: }v                              & {{\color{\colorTEXT}\textnormal{lists}}}
  \cr      &{} {}&               &{}\mathrel{|}{}& \langle \lambda _{x}x.\hspace*{0.33em}e\mathrel{|}\gamma \rangle                            & {{\color{\colorTEXT}\textnormal{recursive closures}}}
  \cr      &{} {}&               &{}\mathrel{|}{}& \rho                                      & {{\color{\colorTEXT}\textnormal{distributions of values}}}
  \end{array}
\end{tabularx}
\end{gather*}\endgroup
\end{framed}
}
\caption{
Syntax for types, expressions and values. {\color{blue!10}$\blacksquare $} =
sensitivity sources, types and expressions unique to \solo.
}
\label{fig:syntax}
\end{figure}
\endgroup
% }-}

% {-{ FIGURE: Type Rules
\begingroup
\renewcommand\b[1]{\colorbox{blue!10}{{{\color{\colorMATH}\ensuremath{#1}}}}}
\begin{figure}
\smaller
\begin{framed}
\begingroup\color{\colorMATH}\begin{gather*}
\begin{tabularx}{\linewidth}{>{\centering\arraybackslash\(}X<{\)}}\hfill\hspace{0pt} \Gamma  \in  {\mtext{tenv}} \triangleq  {\mtext{var}} \rightharpoonup  {\mtext{type}}
  \hfill\hspace{0pt} {}\rceil \Sigma \lceil {}^{s}(o) \triangleq  {}\rceil \Sigma (o)\lceil {}^{s}
  \hfill\hspace{0pt} {}\rceil s\lceil {}^{s^{\prime}} \triangleq  \left\{ \begin{array}{l@{\hspace*{1.00em}}c@{\hspace*{1.00em}}l } 0 &{}{{\color{\colorTEXT}\textnormal{if}}}{}& s \triangleq  0 \cr  s^{\prime} &{}{{\color{\colorTEXT}\textnormal{if}}}{}& s \neq  0 \end{array}\right.
  \hfill\hspace{0pt}
\cr \hfill\hspace{0pt} {\mathcal{R}}({{\color{\colorSYNTAX}\mtexttt{sreal}}}) \triangleq  {{\color{\colorSYNTAX}\mtexttt{real}}} \hfill\hspace{0pt} {\mathcal{R}}(\sigma  \otimes  \sigma ) \triangleq  {\mathcal{R}}(\sigma ) \times  {\mathcal{R}}(\sigma ) \hfill\hspace{0pt} {\mathcal{R}}({{\color{\colorSYNTAX}\mtexttt{slist}}}(\sigma )) \triangleq  {{\color{\colorSYNTAX}\mtexttt{list}}}({\mathcal{R}}(\sigma )) \hfill\hspace{0pt}
\end{tabularx}
\end{gather*}\endgroup
\hrule
\begingroup\color{\colorMATH}\begin{gather*}
\begin{tabularx}{\linewidth}{>{\centering\arraybackslash\(}X<{\)}} \hfill\hspace{0pt} \begingroup\color{\colorTEXT}\boxed{\begingroup\color{\colorMATH} \Gamma  \vdash  e \mathrel{:} \tau  \endgroup}\endgroup
\end{tabularx}
\end{gather*}\endgroup
\vspace{-4ex}
\begingroup\color{\colorMATH}\begin{mathpar} \inferrule*[lab={\mtextsc{ t-var}}
   ]{ \Gamma (x) = \tau 
      }{
      \Gamma  \vdash  x \mathrel{:} \tau 
   }
\and \inferrule*[lab={\mtextsc{ t-blit}}
   ]{ }{
      \Gamma  \vdash  b \mathrel{:} {{\color{\colorSYNTAX}\mtexttt{bool}}}
   }
\and \inferrule*[lab={\mtextsc{ t-rlit}}
   ]{ }{
      \Gamma  \vdash  r \mathrel{:} {{\color{\colorSYNTAX}\mtexttt{real}}}
   }
\and \inferrule*[lab={\mtextsc{ t-sing}}
   ]{ }{
      \Gamma  \vdash  {{\color{\colorSYNTAX}\mtexttt{sing}}}(r) \mathrel{:} {{\color{\colorSYNTAX}\mtexttt{real}}}[r]
   }
\and \inferrule*[lab={\mtextsc{ t-op}}
   ]{ \Gamma  \vdash  e_{1} \mathrel{:} {{\color{\colorSYNTAX}\mtexttt{real}}}
   \\ \Gamma  \vdash  e_{2} \mathrel{:} {{\color{\colorSYNTAX}\mtexttt{real}}}
   \\ \odot  \in  \{ +,\times \} 
      }{
      \Gamma  \vdash  e_{1} \odot  e_{2} \mathrel{:} {{\color{\colorSYNTAX}\mtexttt{real}}}
   }
\and \inferrule*[lab={\mtextsc{ t-if}}
   ]{ \Gamma  \vdash  e_{1} \mathrel{:} {{\color{\colorSYNTAX}\mtexttt{bool}}}
   \\ \Gamma  \vdash  e_{2} \mathrel{:} \tau 
   \\ \Gamma  \vdash  e_{3} \mathrel{:} \tau 
      }{
      \Gamma  \vdash  {{\color{\colorSYNTAX}\mtexttt{if}}}(e_{1})\{ e_{2}\} \{ e_{3}\}  \mathrel{:} \tau 
   }
\and \inferrule*[lab={\mtextsc{ t-pair}}
   ]{ \Gamma  \vdash  e_{1} \mathrel{:} \tau _{1}
   \\ \Gamma  \vdash  e_{2} \mathrel{:} \tau _{2}
      }{
      \Gamma  \vdash  \langle e_{1},e_{2}\rangle  \mathrel{:} \tau _{1} \times  \tau _{2}
   }
\and \inferrule*[lab={\mtextsc{ t-proj}}
   ]{ \Gamma  \vdash  e \mathrel{:} \tau _{1} \times  \tau _{2}
      }{
      \Gamma  \vdash  \pi _{i}(e) \mathrel{:} \tau _{i}
   }
\and \inferrule*[lab={\mtextsc{ t-nil}}
   ]{ }{
      \Gamma  \vdash  [] \mathrel{:} {{\color{\colorSYNTAX}\mtexttt{list}}}(\tau )
   }
\and \inferrule*[lab={\mtextsc{ t-cons}}
   ]{ \Gamma  \vdash  e_{1} \mathrel{:} \tau 
   \\ \Gamma  \vdash  e_{2} \mathrel{:} {{\color{\colorSYNTAX}\mtexttt{list}}}(\tau )
      }{
      \Gamma  \vdash  e_{1} \mathrel{:: } e_{2} \mathrel{:} {{\color{\colorSYNTAX}\mtexttt{list}}}(\tau )
   }
\and \inferrule*[lab={\mtextsc{ t-case}}
   ]{ \Gamma  \vdash  e_{1} \mathrel{:} {{\color{\colorSYNTAX}\mtexttt{list}}}(\tau )
   \\ \Gamma  \vdash  e_{2} \mathrel{:} \tau ^{\prime}
   \\ \{ x_{1} \mapsto  \tau ,x_{2} \mapsto  {{\color{\colorSYNTAX}\mtexttt{list}}}(\tau )\}  \uplus  \Gamma  \vdash  e_{3} \mathrel{:} \tau ^{\prime}
      }{
      \Gamma  \vdash  {{\color{\colorSYNTAX}\mtexttt{case}}}(e_{1})\{ [].e_{2}\} \{ x_{1}\mathrel{:: }x_{2}.e_{3}\}  \mathrel{:} \tau ^{\prime}
   }
\and \inferrule*[lab={\mtextsc{ t-lam}}
   ]{ \{ x \mapsto  \tau _{1},z \mapsto  \tau _{1} \rightarrow  \tau _{2}\}  \uplus  \Gamma  \vdash  e \mathrel{:} \tau _{2}
      }{
      \Gamma  \vdash  \lambda _{z}x.\hspace*{0.33em} e \mathrel{:} \tau _{1} \rightarrow  \tau _{2}
   }
\and \inferrule*[lab={\mtextsc{ t-app}}
   ]{ \Gamma  \vdash  e_{1} \mathrel{:} \tau _{1} \rightarrow  \tau _{2}
   \\ \Gamma  \vdash  e_{2} \mathrel{:} \tau _{1}
      }{
      \Gamma  \vdash  e_{1}(e_{2}) \mathrel{:} \tau _{2}
   }
\and \b{
   \inferrule*[lab={\mtextsc{ t-reveal}}
   ]{ \Gamma  \vdash  e \mathrel{:} \sigma @\Sigma 
      }{
      \Gamma  \vdash  {{\color{\colorSYNTAX}\mtexttt{reveal}}}(e) \mathrel{:} {\scriptstyle \bigcirc }_{{}\rceil \Sigma \lceil {}^{\infty }}({\mathcal{R}}(\sigma ))
   }
   }
\and \inferrule*[lab={\mtextsc{ t-laplace}}
   ]{ \Gamma  \vdash  e_{1} \mathrel{:} {{\color{\colorSYNTAX}\mtexttt{real}}}[r_{s}]
   \\ \Gamma  \vdash  e_{2} \mathrel{:} {{\color{\colorSYNTAX}\mtexttt{real}}}[r_{\epsilon }]
   \\ \Gamma  \vdash  e_{3} \mathrel{:} {{\color{\colorSYNTAX}\mtexttt{sreal}}}@\Sigma 
   \\ \Sigma  \sqsubseteq  {}\rceil \Sigma \lceil {}^{s}
      }{
      \Gamma  \vdash  {{\color{\colorSYNTAX}\mtexttt{laplace}}}[e_{1},e_{2}](e_{3}) \mathrel{:} {\scriptstyle \bigcirc }_{{}\rceil \Sigma \lceil {}^{\epsilon }}({{\color{\colorSYNTAX}\mtexttt{real}}})
   }
\and \b{
   \inferrule*[lab={\mtextsc{ t-return}}
   ]{ \Gamma  \vdash  e \mathrel{:} \tau 
      }{
      \Gamma  \vdash  {{\color{\colorSYNTAX}\mtexttt{return}}}(e) \mathrel{:} {\scriptstyle \bigcirc }_{\varnothing }(\tau )
   }
   }
\and \inferrule*[lab={\mtextsc{ t-bind}}
   ]{ \Gamma  \vdash  e_{1} \mathrel{:} {\scriptstyle \bigcirc }_{\Sigma _{1}}(\tau _{1})
      \{ x\mapsto \tau _{1}\} \uplus \Gamma  \vdash  e_{2} \mathrel{:} {\scriptstyle \bigcirc }_{\Sigma _{2}}(\tau _{2})
      }{
      \Gamma  \vdash  x \leftarrow  e_{1} \mathrel{;} e_{2} \mathrel{:} {\scriptstyle \bigcirc }_{\Sigma _{1}+\Sigma _{2}}(\tau _{2})
   }
\and \b{
   \inferrule*[lab={\mtextsc{ t-splus}}
   ]{ \Gamma  \vdash  e_{1} \mathrel{:} {{\color{\colorSYNTAX}\mtexttt{sreal}}}@\Sigma _{1}
   \\ \Gamma  \vdash  e_{2} \mathrel{:} {{\color{\colorSYNTAX}\mtexttt{sreal}}}@\Sigma _{2}
      }{
      \Gamma  \vdash  e_{1} + e_{2} \mathrel{:} {{\color{\colorSYNTAX}\mtexttt{sreal}}}@(\Sigma _{1}+\Sigma _{2})
   }
   }
\and \b
   {\inferrule*[lab={\mtextsc{ t-stimes}}
   ]{ \Gamma  \vdash  e_{1} \mathrel{:} {{\color{\colorSYNTAX}\mtexttt{real}}}[r]
   \\ \Gamma  \vdash  e_{2} \mathrel{:} {{\color{\colorSYNTAX}\mtexttt{sreal}}}@\Sigma 
      }{
      \Gamma  \vdash  e_{1} \ltimes  e_{2} \mathrel{:} {{\color{\colorSYNTAX}\mtexttt{sreal}}}@r\Sigma 
   }
   }
\and \b{
   \inferrule*[lab={\mtextsc{ t-spair}}
   ]{ \Gamma  \vdash  e_{1} \mathrel{:} \sigma _{1}@\Sigma _{1}
   \\ \Gamma  \vdash  e_{2} \mathrel{:} \sigma _{2}@\Sigma _{2}
      }{
      \Gamma  \vdash  \hat \langle e_{1},e_{2}\hat \rangle  \mathrel{:} (\sigma _{1} \otimes  \sigma _{2})@(\Sigma _{1} \sqcup  \Sigma _{2})
   }
   }
\and \b{
   \inferrule*[lab={\mtextsc{ t-sproj}}
   ]{ \Gamma  \vdash  e \mathrel{:} (\sigma _{1} \otimes  \sigma _{2})@\Sigma 
      }{
      \Gamma  \vdash  \hat \pi _{i}(e) \mathrel{:} \sigma _{i}@\Sigma 
   }
   }
\and \b{
   \inferrule*[lab={\mtextsc{ t-snil}}
   ]{ }{
      \Gamma  \vdash  \hat {[}\hat {]} \mathrel{:} {{\color{\colorSYNTAX}\mtexttt{slist}}}(\sigma )@\varnothing 
   }
   }
\and \b{
   \inferrule*[lab={\mtextsc{ t-scons}}
   ]{ \Gamma  \vdash  e_{1} \mathrel{:} \sigma @\Sigma _{1}
   \\ \Gamma  \vdash  e_{2} \mathrel{:} {{\color{\colorSYNTAX}\mtexttt{slist}}}(\sigma )@\Sigma _{2}
      }{
      \Gamma  \vdash  e_{1} \mathrel{\hat {\mathrel{:: }}} e_{2} \mathrel{:} {{\color{\colorSYNTAX}\mtexttt{slist}}}(\tau )@(\Sigma _{1} \sqcup  \Sigma _{2})
   }
   }
\and \b{
   \inferrule*[lab={\mtextsc{ t-scase}}
   ]{ \Gamma  \vdash  e_{1} \mathrel{:} {{\color{\colorSYNTAX}\mtexttt{slist}}}(\sigma )@\Sigma 
   \\ \Gamma  \vdash  e_{2} \mathrel{:} \tau ^{\prime}
   \\ \{ x_{1} \mapsto  \sigma @\Sigma ,x_{2} \mapsto  {{\color{\colorSYNTAX}\mtexttt{slist}}}(\sigma )@\Sigma \}  \uplus  \Gamma  \vdash  e_{3} \mathrel{:} \tau ^{\prime}
      }{
      \Gamma  \vdash  {{\color{\colorSYNTAX}\mtexttt{case}}}(e_{1})\{ \hat {[}\hat {]}.e_{2}\} \{ x_{1}\mathrel{\hat {\mathrel{:: }}}x_{2}.e_{3}\}  \mathrel{:} \tau ^{\prime}
   }
   }
\end{mathpar}\endgroup
\end{framed}
\caption{
The type system. {\color{blue!10}$\blacksquare $} =
type rules unique to \solo.
}
\label{fig:type-system}
\end{figure}
\endgroup
% }-}

% {-{ FIGURE: Semantics
\begin{figure}
\begin{framed}
\begingroup\color{\colorMATH}\begin{gather*}
\begin{tabularx}{\linewidth}{>{\centering\arraybackslash\(}X<{\)}}\hfill\hspace{0pt} \bar \rho  \in  {\mtext{value}} \rightarrow  {\mtext{ddist}} \hfill\hspace{0pt} \bar n \in  {\mtext{value}} \rightarrow  {\mathbb{N}} \hfill\hspace{0pt} \begingroup\color{\colorTEXT}\boxed{\begingroup\color{\colorMATH} \gamma  \vdash  e \Downarrow _{n} v \endgroup}\endgroup
\end{tabularx}
\end{gather*}\endgroup
\begingroup\color{\colorMATH}\begin{mathpar} \inferrule*[lab={\mtextsc{ e-var}}
   ]{ \gamma (x) = v
      }{
      \gamma  \vdash  x \Downarrow _{0} v
   }
\and \inferrule*[lab={\mtextsc{ e-blit}}
   ]{ }{
      \gamma  \vdash  b \Downarrow _{0} b
   }
\and \inferrule*[lab={\mtextsc{ e-rlit}}
   ]{ }{
      \gamma  \vdash  r \Downarrow _{0} r
   }
\and \inferrule*[lab={\mtextsc{ e-plus}}
   ]{ \gamma  \vdash  e_{1} \Downarrow _{n_{1}} r_{1}
   \\ \gamma  \vdash  e_{2} \Downarrow _{n_{2}} r_{2}
      }{
      \gamma  \vdash  e_{1} + e_{2} \Downarrow _{n_{1}+n_{1}} r_{1}+r_{2}
   }
\and \inferrule*[lab={\mtextsc{ e-times}}
   ]{ \gamma  \vdash  e_{1} \Downarrow _{n_{1}} r_{1}
   \\ \gamma  \vdash  e_{2} \Downarrow _{n_{2}} r_{2}
      }{
      \gamma  \vdash  e_{1}\times e_{2} \Downarrow _{n_{1}+n_{2}} r_{1}r_{2}
   \\\\ \gamma  \vdash  e_{1}\ltimes e_{2} \Downarrow _{n_{1}+n_{2}} r_{1}r_{2}
   }
\and \inferrule*[lab={\mtextsc{ e-if-true}}
   ]{ \gamma  \vdash  e_{1} \Downarrow  {\mtext{true}}
      \gamma  \vdash  e_{2} \Downarrow  v
      }{
      \gamma  \vdash  {{\color{\colorSYNTAX}\mtexttt{if}}}(e_{1})\{ e_{2}\} \{ e_{3}\}  \Downarrow  v
   }
\and \inferrule*[lab={\mtextsc{ e-if-false}}
   ]{ \gamma  \vdash  e_{1} \Downarrow  {\mtext{false}}
      \gamma  \vdash  e_{3} \Downarrow  v
      }{
      \gamma  \vdash  {{\color{\colorSYNTAX}\mtexttt{if}}}(e_{1})\{ e_{2}\} \{ e_{3}\}  \Downarrow  v
   }
\and \inferrule*[lab={\mtextsc{ e-pair}}
   ]{ \gamma  \vdash  e_{1} \Downarrow _{n_{1}} v_{1}
   \\ \gamma  \vdash  e_{2} \Downarrow _{n_{2}} v_{2}
      }{
      \gamma  \vdash  \langle e_{1},e_{2}\rangle  \Downarrow _{n_{1}+n_{2}} \langle v_{1},v_{2}\rangle 
   \\\\ \gamma  \vdash  \hat \langle e_{1},e_{2}\hat \rangle  \Downarrow _{n_{1}+n_{2}} \langle v_{1},v_{2}\rangle 
   }
\and \inferrule*[lab={\mtextsc{ e-proj}}
   ]{ \gamma  \vdash  e \Downarrow _{n} \langle v_{1},v_{2}\rangle 
      }{
      \gamma  \vdash  \pi _{i} \Downarrow _{n} v_{i}
   \\\\ \gamma  \vdash  \hat \pi _{i} \Downarrow _{n} v_{i}
   }
\and \inferrule*[lab={\mtextsc{ e-nil}}
   ]{ }{
      \gamma  \vdash  [] \Downarrow _{0} []
   \\\\ \gamma  \vdash  \hat {[}\hat {]} \Downarrow _{0} []
   }
\and \inferrule*[lab={\mtextsc{ e-cons}}
   ]{ \gamma  \vdash  e_{1} \Downarrow _{n_{1}} v_{1}
   \\ \gamma  \vdash  e_{2} \Downarrow _{n_{2}} v_{2}
      }{
      \gamma  \vdash  e_{1}\mathrel{:: }e_{2} \Downarrow _{n_{1}+n_{2}} v_{1}\mathrel{:: }v_{2}
   \\\\ \gamma  \vdash  e_{1}\mathrel{\hat {\mathrel{:: }}}e_{2} \Downarrow _{n_{1}+n_{2}} v_{1}\mathrel{:: }v_{2}
   }
\and \inferrule*[lab={\mtextsc{ e-case-nil}}
   ]{ \gamma  \vdash  e_{1} \Downarrow _{n_{1}} []
   \\ \gamma  \vdash  e_{2} \Downarrow _{n_{2}} v
      }{
      \gamma  \vdash  {{\color{\colorSYNTAX}\mtexttt{case}}}(e_{1})\{ [].e_{2}\} \{ x_{1}\mathrel{:: }x_{2}.e_{3}\}  \Downarrow _{n_{1}+n_{2}} v
   \\\\ \gamma  \vdash  {{\color{\colorSYNTAX}\mtexttt{case}}}(e_{1})\{ \hat {[}\hat {]}.e_{2}\} \{ x_{1}\mathrel{\hat {\mathrel{:: }}}x_{2}.e_{3}\}  \Downarrow _{n_{1}+n_{2}} v
   }
\and \inferrule*[lab={\mtextsc{ e-case-cons}}
   ]{ \gamma  \vdash  e_{1} \Downarrow _{n_{1}} v_{1}\mathrel{:: }v_{2}
   \\ \{ x_{1}\mapsto v_{1},x_{2}\mapsto v_{2}\}  \uplus  \gamma  \vdash  e_{3} \Downarrow _{n_{2}} v_{3}
      }{
      \gamma  \vdash  {{\color{\colorSYNTAX}\mtexttt{case}}}(e_{1})\{ [].e_{2}\} \{ x_{1}\mathrel{:: }x_{2}.e_{3}\}  \Downarrow _{n_{1}+n_{2}} v_{3}
   \\\\ \gamma  \vdash  {{\color{\colorSYNTAX}\mtexttt{case}}}(e_{1})\{ \hat {[}\hat {]}.e_{2}\} \{ x_{1}\mathrel{\hat {\mathrel{:: }}}x_{2}.e_{3}\}  \Downarrow _{n_{1}+n_{2}} v_{3}
   }
\and \inferrule*[lab={\mtextsc{ e-lam}}
   ]{ }{
      \gamma  \vdash  \lambda _{z}x.\hspace*{0.33em}e \Downarrow _{0} \langle \lambda _{z}x.\hspace*{0.33em}e\mathrel{|}\gamma \rangle 
   }
\and \inferrule*[lab={\mtextsc{ e-app}}
   ]{ \gamma  \vdash  e_{1} \Downarrow _{n_{1}} \langle \lambda _{z}x.\hspace*{0.33em}e^{\prime}\mathrel{|}\gamma ^{\prime}\rangle 
   \\ \gamma  \vdash  e_{2} \Downarrow _{n_{2}} v_{1}
   \\ \{ x\mapsto v_{1},z\mapsto \langle \lambda _{z}x.\hspace*{0.33em}e^{\prime}\mathrel{|}\gamma ^{\prime}\rangle \}  \uplus  \gamma ^{\prime}\vdash  e^{\prime} \Downarrow _{n_{3}} v_{2}
      }{
      \gamma  \vdash  e_{1}(e_{2}) \Downarrow _{n_{1}+n_{2}+n_{3}+1} v_{2}
   }
\and \inferrule*[lab={\mtextsc{ e-reveal}}
   ]{ \gamma  \vdash  e \Downarrow _{n} v
      }{
      \gamma  \vdash  {{\color{\colorSYNTAX}\mtexttt{reveal}}}(e) \Downarrow _{n} \{ v\mapsto 1\} 
   \\\\ \gamma  \vdash  {{\color{\colorSYNTAX}\mtexttt{return}}}(e) \Downarrow _{n} \{ v\mapsto 1\} 
   }
\and \inferrule*[lab={\mtextsc{ e-laplace}}
   ]{ \gamma  \vdash  e_{1} \Downarrow _{n} s
   \\ \gamma  \vdash  e_{2} \Downarrow _{n} \epsilon 
   \\ \gamma  \vdash  e_{3} \Downarrow _{n} r
      }{
      \gamma  \vdash  {{\color{\colorSYNTAX}\mtexttt{laplace}}}[e_{1}, e_{2}](e_{3}) \Downarrow _{n} {\mtext{laplace}}(r, s/\epsilon )
   }
\and \inferrule*[lab={\mtextsc{ e-bind}}
   ]{ \gamma  \vdash  e_{1} \Downarrow _{n_{1}} \rho _{1}
   \\ \forall  v.\hspace*{0.33em} \{ x\mapsto v\}  \uplus  \gamma  \vdash  e_{2} \Downarrow _{\bar n_{2}(v)} \bar \rho _{2}(v)
      }{
      \gamma  \vdash  x \leftarrow  e_{1} \mathrel{;} e_{2} \Downarrow _{\left(n_{1} + \bigsqcup \limits_{v}\bar n_{2}(v)\right)} \left\{ v \mapsto  \sum \limits_{v^{\prime}}\rho _{1}(v^{\prime})\bar \rho _{2}(v^{\prime})(v)\right\} 
   }
\end{mathpar}\endgroup
\end{framed}
\caption{Step-indexed big-step evaluation semantics.}
\label{fig:semantics}
\end{figure}
% }-}

\paragraph{Typing Rules.}
Figure~\ref{fig:type-system} shows typing rules in our system used to reason about the sensitivity of computations. The majority of these rules are standard, and modeled on the corresponding rules in our implementation language (Haskell). In particular, the rule for function introduction ({{\color{\colorMATH}\ensuremath{{\mtextsc{ t-lam}}}}}) does not mention sensitivities or sensitive types.

The rules with a shaded background ({\color{blue!10}$\blacksquare $}) are unique to \solo, and model the primitives described earlier in the paper. For example, the rule {{\color{\colorMATH}\ensuremath{{\mtextsc{ t-splus}}}}} models the addition operator, which adds the sensitivity environments attached to its arguments ({{\color{\colorMATH}\ensuremath{\Sigma _{1} + \Sigma _{2}}}}). Addition of sensitivity environments is identical to \fuzz~\cite{reed2010distance} and \duet~\cite{near2019duet}, and models the implementation described earlier:
%
% {}\rceil s\lceil {}^{s^{\prime}} \triangleq  \left\{ \begin{array}{l@{\hspace*{1.00em}}c@{\hspace*{1.00em}}l } 0 &{}{{\color{\colorTEXT}\textnormal{if}}}{}& s \triangleq  0 \cr  s^{\prime} &{}{{\color{\colorTEXT}\textnormal{if}}}{}& s \neq  0 \end{array}\right.
%
\begingroup\color{\colorMATH}\begin{gather*}
(\Sigma _{1} + \Sigma _{2})(o) \triangleq  \left\{ \begin{array}{l@{\hspace*{1.00em}}c@{\hspace*{1.00em}}l } \Sigma _{1}(o) + \Sigma _{2}(o) &{}{{\color{\colorTEXT}\textnormal{if}}}{}& o \in  \Sigma _{1} \hspace*{0.33em}{{\color{\colorTEXT}\textnormal{and}}}\hspace*{0.33em}o \in  \Sigma _{2}
                 \cr  \Sigma _{1}(o) &{}{{\color{\colorTEXT}\textnormal{if}}}{}& o \in  \Sigma _{1} \hspace*{0.33em}{{\color{\colorTEXT}\textnormal{but}}}\hspace*{0.33em}o \notin  \Sigma _{2}
                 \cr  \Sigma _{2}(o) &{}{{\color{\colorTEXT}\textnormal{if}}}{}& o \in  \Sigma _{2} \hspace*{0.33em}{{\color{\colorTEXT}\textnormal{but}}}\hspace*{0.33em}o \notin  \Sigma _{1} \end{array}\right.
\end{gather*}\endgroup
The pointwise maximum of two sensitivity environments ({{\color{\colorMATH}\ensuremath{\Sigma _{1} \sqcup  \Sigma _{2}}}}) is defined analogously, but with the numeric maximum instead of addition; it is used in the rule {{\color{\colorMATH}\ensuremath{{\mtextsc{ t-spair}}}}} for pairs.
The rule {{\color{\colorMATH}\ensuremath{{\mtextsc{ t-stimes}}}}} describes multiplication of a sensitive value by a statically-known number, which \emph{scales} the associated sensitivity environment. Sensitivity environment scaling {{\color{\colorMATH}\ensuremath{{\begingroup\renewcommand\colorMATH{\colorMATHA}\renewcommand\colorSYNTAX{\colorSYNTAXA}{{\color{\colorSYNTAX}\mtexttt{{\ensuremath{s(}}}}}\endgroup } \Sigma  {\begingroup\renewcommand\colorMATH{\colorMATHA}\renewcommand\colorSYNTAX{\colorSYNTAXA}{{\color{\colorSYNTAX}\mtexttt{{\ensuremath{)}}}}}\endgroup }}}} is defined as:
\begingroup\color{\colorMATH}\begin{gather*}
{\begingroup\renewcommand\colorMATH{\colorMATHA}\renewcommand\colorSYNTAX{\colorSYNTAXA}{{\color{\colorSYNTAX}\mtexttt{{\ensuremath{s(}}}}}\endgroup } \Sigma  {\begingroup\renewcommand\colorMATH{\colorMATHA}\renewcommand\colorSYNTAX{\colorSYNTAXA}{{\color{\colorSYNTAX}\mtexttt{{\ensuremath{)}}}}}\endgroup }(o) \triangleq  {\begingroup\renewcommand\colorMATH{\colorMATHA}\renewcommand\colorSYNTAX{\colorSYNTAXA}{{\color{\colorSYNTAX}\mtexttt{{\ensuremath{s(}}}}}\endgroup } \Sigma (o) {\begingroup\renewcommand\colorMATH{\colorMATHA}\renewcommand\colorSYNTAX{\colorSYNTAXA}{{\color{\colorSYNTAX}\mtexttt{{\ensuremath{)}}}}}\endgroup }
\end{gather*}\endgroup
The truncation operation {{\color{\colorMATH}\ensuremath{{}\rceil {\begingroup\renewcommand\colorMATH{\colorMATHA}\renewcommand\colorSYNTAX{\colorSYNTAXA}{{\color{\colorMATH}\ensuremath{ {\begingroup\renewcommand\colorMATH{\colorMATHB}\renewcommand\colorSYNTAX{\colorSYNTAXB}{{\color{\colorMATH}\ensuremath{\Sigma }}}\endgroup } }}}\endgroup }\lceil {}^{{\begingroup\renewcommand\colorMATH{\colorMATHC}\renewcommand\colorSYNTAX{\colorSYNTAXC}{{\color{\colorMATH}\ensuremath{\epsilon }}}\endgroup }}}}} is also defined as seen in prior work \cite{near2019duet}. This operation converts sensitivity environments to privacy environments, by replacing each sensitivity in the environment with a consistent privacy cost (i.e. pointwise). While typing rules for arithmetic operations vary for the several permutations of static(singleton)/dynamic arguments, we only show the interesting cases for the multiplication operator.
% Although most privacy examples reference just {\begingroup\renewcommand\colorMATH{\colorMATHC}\renewcommand\colorSYNTAX{\colorSYNTAXC}{{\color{\colorMATH}\ensuremath{\epsilon }}}\endgroup }-differential privacy for simplicity, our system is capable of several other variants.
{\mtextsc{ t-return}}, {\mtextsc{ t-reveal}}, and {\mtextsc{ t-laplace}} model the privacy primitives described in Section~\ref{sec:privacy}.
%The {\mtextsc{ t-return}} rule is the first interface between sensitivity and privacy: its value is the value of its argument in the privacy monad with the empty privacy context i.e. no privacy afforded. {\mtextsc{ t-laplace}} encodes the Laplace mechanism, and represents any valid privacy mechanism discussed throughout this paper.
{\mtextsc{ t-bind}} encodes Theorem~\ref{thm:sequential-composition} (sequential composition).  In general, the typing rules are similar to previous work, except that the sensitivity and privacy environments are properties of (and embedded in) the types themselves, rather than being a property of the program context.

\paragraph{Dynamic Semantics.}
Figure~\ref{fig:semantics} shows a core subset of the standard dynamic semantics that accompanies the syntax for our analysis system. Our semantics largely follows the structure of \fuzz~\cite{reed2010distance}, and model the evaluation of expressions to \emph{discrete distributions} of values (since our privacy mechanisms are randomized). Distributions are represented as mappings from values to their probabilities (i.e. probability mass functions). The rule {{\color{\colorMATH}\ensuremath{{\mtextsc{ e-reveal}}}}} says that a deterministic reveal of a value produces a point distribution ({{\color{\colorMATH}\ensuremath{\{ v \mapsto  1\} }}}); the rule {{\color{\colorMATH}\ensuremath{{\mtextsc{ e-bind}}}}} encodes sequential composition. The rule {{\color{\colorMATH}\ensuremath{{\mtextsc{ e-laplace}}}}} returns a discrete Laplace distribution centered at {{\color{\colorMATH}\ensuremath{r}}} with scale {{\color{\colorMATH}\ensuremath{s / \epsilon }}}.

\paragraph{Type Soundness.}
The property of type soundness in our system is defined (as in prior work) as the \emph{metric preservation} theorem. Essentially, metric preservation dictates a maximum variation which is possible when a sensitive open term is closed over by two distinct but related sensitive closure environments. This means that given related initial well-typed configurations, we expect the outputs to be related by some level of variation. Specifically: given two well-typed environments which are related by the logical relation (values may be apart by distance {{\color{\colorMATH}\ensuremath{\Sigma }}}, for {{\color{\colorMATH}\ensuremath{n}}} steps), and a well typed term, then each evaluation of that term in each environment is related by the relation, that is, when one side terminates in {{\color{\colorMATH}\ensuremath{<n}}} steps to a value, the other side will deterministically terminate to a related value.
Similar to prior work, in order to state and prove the metric preservation theorem, we define the notion of function sensitivity as a (step-indexed) logical relation. Figure~\ref{fig:logical-relations} shows the step-indexed logical relation used to define function sensitivity. We briefly describe the logical relations seen in this figure, then state the metric preservation theorem formally.

\begin{enumerate}\item  Two real numbers are related {{\color{\colorMATH}\ensuremath{r_{1} \sim ^{r} r_{2}}}} at type {{\color{\colorMATH}\ensuremath{{\mathbb{R}}}}} and distance {{\color{\colorMATH}\ensuremath{r}}} when the absolute difference between real numbers {{\color{\colorMATH}\ensuremath{r_{1}}}} and {{\color{\colorMATH}\ensuremath{r_{2}}}} is less than {{\color{\colorMATH}\ensuremath{r}}}.
\item  Two values are related {{\color{\colorMATH}\ensuremath{v_{1} \sim  v_{2}}}} in {{\color{\colorMATH}\ensuremath{{\mathcal{V}}_{\Sigma }\llbracket \tau \rrbracket }}} when {{\color{\colorMATH}\ensuremath{v_{1}}}} and {{\color{\colorMATH}\ensuremath{v_{2}}}} are related at type {{\color{\colorMATH}\ensuremath{\tau }}} for
   initial distance {{\color{\colorMATH}\ensuremath{\Sigma }}}. We may define relatedness for the syntactic category of
   values via case analysis as follows:
   \begin{enumerate}\item  Base numeric values are related {{\color{\colorMATH}\ensuremath{r_{1} \sim ^{\Sigma } r_{2}}}} at type {{\color{\colorMATH}\ensuremath{{\mathbb{R}}}}}  in {{\color{\colorMATH}\ensuremath{{\mathcal{V}}_{\Sigma _{1}}\llbracket \tau \rrbracket }}}
      when {{\color{\colorMATH}\ensuremath{r_{1}}}} and {{\color{\colorMATH}\ensuremath{r_{2}}}} are related by {{\color{\colorMATH}\ensuremath{\Sigma \mathord{\cdotp }\Sigma _{1}}}}, where {{\color{\colorMATH}\ensuremath{\Sigma }}} is the
      initial distances between each input source {{\color{\colorMATH}\ensuremath{o}}}, and {{\color{\colorMATH}\ensuremath{\Sigma _{1}}}} describes
      how much these values may wiggle as function arguments i.e. the maximum permitted argument variation. {{\color{\colorMATH}\ensuremath{\mathord{\cdotp }}}} is defined as the vector dot product.
   \item  Function values {{\color{\colorMATH}\ensuremath{\langle \lambda x.\hspace*{0.33em}e_{1}\mathrel{|}\gamma _{1}\rangle  \sim  \langle \lambda x.\hspace*{0.33em}e_{2}\mathrel{|}\gamma _{2}\rangle }}} are related at type {{\color{\colorMATH}\ensuremath{(\tau \rightarrow \tau )}}} in {{\color{\colorMATH}\ensuremath{{\mathcal{V}}_{\Sigma }\llbracket \tau \rrbracket }}}
      when given related inputs, they produce related computations.
   \item  Pair values {{\color{\colorMATH}\ensuremath{\langle v_{1 1},v_{1 2}\rangle  \sim  \langle v_{2 1},v_{2 2}\rangle }}} are related at type {{\color{\colorMATH}\ensuremath{\langle \tau ,\tau \rangle }}} in {{\color{\colorMATH}\ensuremath{{\mathcal{V}}_{\Sigma }\llbracket \tau \rrbracket }}}
      when they are elementwise related.
   \item  {{\color{\colorMATH}\ensuremath{\gamma _{1},e_{1} \sim  \gamma _{2},e_{2}}}} are related at type {{\color{\colorMATH}\ensuremath{\tau }}} and distance {{\color{\colorMATH}\ensuremath{\Sigma }}} in {{\color{\colorMATH}\ensuremath{{\mathcal{E}}_{\Sigma }\llbracket \tau \rrbracket }}} when the input doubles {{\color{\colorMATH}\ensuremath{\gamma _{1},e_{1}}}} and {{\color{\colorMATH}\ensuremath{\gamma _{2},e_{2}}}} evaluate to output values which are related by {{\color{\colorMATH}\ensuremath{\Sigma }}}.
   \end{enumerate}
\item  Two value environments {{\color{\colorMATH}\ensuremath{\gamma _{1} \sim  \gamma _{2}}}} are related at type environment {{\color{\colorMATH}\ensuremath{\Gamma }}} and sensitivity environment {{\color{\colorMATH}\ensuremath{\Sigma }}} in {{\color{\colorMATH}\ensuremath{{\mathcal{G}}_{\Sigma }\llbracket \Gamma \rrbracket }}} if value environments {{\color{\colorMATH}\ensuremath{\gamma _{1}}}} and {{\color{\colorMATH}\ensuremath{\gamma _{2}}}} both map each variable in the type environment {{\color{\colorMATH}\ensuremath{\Gamma }}} to related values at a matching type at distance {{\color{\colorMATH}\ensuremath{\Sigma }}}.
\end{enumerate}

\begin{theorem}[Metric Preservation]\label{thm:metric-preservation}\ \\
  \begin{itemize}[label={},leftmargin=0pt]\item  \begin{tabular}{r@{\hspace*{1.00em}}l@{\hspace*{1.00em}}l
     } If:      & {{\color{\colorMATH}\ensuremath{\gamma _{1} \sim  \gamma _{2} \in  {\mathcal{G}}_{n}^{\Sigma }\llbracket \Gamma \rrbracket }}}     & {{\color{\colorTEXT}\textnormal{{\mtextit{(H1)}}}}}
     \cr  And:     & {{\color{\colorMATH}\ensuremath{ \Gamma  \vdash  e \mathrel{:} \tau }}}             & {{\color{\colorTEXT}\textnormal{{\mtextit{(H2)}}}}}
     \cr  Then:    & {{\color{\colorMATH}\ensuremath{\gamma _{1},e \sim  \gamma _{2},e \in  {\mathcal{E}}_{n}^{\Sigma }\llbracket \tau \rrbracket }}} &
     \end{tabular}
  \item 
  \item  That is, either {{\color{\colorMATH}\ensuremath{n=0}}}, or {{\color{\colorMATH}\ensuremath{n = n^{\prime}+1}}} and\ldots 
  \item 
  \item  \begin{tabular}{r@{\hspace*{1.00em}}l@{\hspace*{1.00em}}l
     } If:      & {{\color{\colorMATH}\ensuremath{n^{\prime \prime} \leq  n^{\prime}}}}                   & {{\color{\colorTEXT}\textnormal{{\mtextit{(H3)}}}}}
     \cr  And:     & {{\color{\colorMATH}\ensuremath{\gamma _{1} \vdash  e \Downarrow _{n^{\prime \prime}} v_{1}}}}           & {{\color{\colorTEXT}\textnormal{{\mtextit{(H4)}}}}}
     \cr  Then:    & {{\color{\colorMATH}\ensuremath{\exists !v_{2}.\hspace*{0.33em} \gamma _{2} \vdash  e \Downarrow _{n^{\prime \prime}} v_{2}}}}     & {{\color{\colorTEXT}\textnormal{{\mtextit{(C1)}}}}}
     \cr  And:     & {{\color{\colorMATH}\ensuremath{v_{1} \sim  v_{2} \in  {\mathcal{V}}_{n^{\prime}-n^{\prime \prime}}^{\Sigma }\llbracket \tau \rrbracket }}}  & {{\color{\colorTEXT}\textnormal{{\mtextit{(C2)}}}}}
     \end{tabular}
  \end{itemize}
\end{theorem}

The proofs appear in Appendix~\ref{sec:proof} in the supplemental material.

\begin{figure}
\begin{framed}
\begingroup\color{\colorMATH}\begin{gather*}
\begin{tabularx}{\linewidth}{>{\centering\arraybackslash\(}X<{\)}}\hfill\hspace{0pt}
  \begin{array}[t]{rclrlrlrl
  } \gamma _{1},e_{1} &{}\sim {}& \gamma _{2},e_{2} &{}\in {}& {\mathcal{E}}_{n}^{\Sigma }\llbracket \tau \rrbracket  &{}\overset \vartriangle \iff {}& n=0    &{}\implies {}& {\mtext{true}}
  \cr        &{} {}&       &{} {}&            &{}\wedge {}& n=n^{\prime}+1 &{}\implies {}& \forall  n^{\prime \prime}\leq n^{\prime}, v_{1}.\hspace*{0.33em} \gamma _{1} \vdash  e_{1} \Downarrow _{n^{\prime \prime}} v_{1}
  \cr        &{} {}&       &{} {}&            &{} {}&        &{}\Rightarrow {}& \exists ! v_{2}.\hspace*{0.33em} \gamma _{2} \vdash  e_{2} \Downarrow _{n^{\prime \prime}} v_{2} \hspace*{0.33em}\wedge \hspace*{0.33em}v_{1} \sim  v_{2} \in  {\mathcal{V}}_{n^{\prime}-n^{\prime \prime}}^{\Sigma }\llbracket \tau \rrbracket 
  \end{array}
  \hfill\hspace{0pt}
  \mathllap{\begingroup\color{\colorTEXT}\boxed{\begingroup\color{\colorMATH} \gamma ,e \sim  \gamma ,e \in  {\mathcal{E}}_{n}^{\Sigma }\llbracket \tau \rrbracket  \endgroup}\endgroup}
\cr 
\cr \hfill\hspace{0pt}
  r_{1} \sim ^{r} r_{2} \overset \vartriangle \iff  |r_{1} - r_{2}| \leq  r
  \hfill\hspace{0pt}
  \begingroup\color{\colorTEXT}\boxed{\begingroup\color{\colorMATH} r \sim ^{r} r \endgroup}\endgroup
\cr 
\cr \hfill\hspace{0pt}
  \begin{array}[t]{rclrl
  } b_{1} \sim  b_{2} &{}\in {}& {\mathcal{V}}_{n}^{\Sigma }\llbracket {{\color{\colorSYNTAX}\mtexttt{bool}}}\rrbracket                      &{}\overset \vartriangle \iff {}& b_{1} = b_{2}
  \cr  r_{1} \sim  r_{2} &{}\in {}& {\mathcal{V}}_{n}^{\Sigma }\llbracket {{\color{\colorSYNTAX}\mtexttt{real}}}\rrbracket                      &{}\overset \vartriangle \iff {}& r_{1} = r_{2}
  \cr  r_{1} \sim  r_{2} &{}\in {}& {\mathcal{V}}_{n}^{\Sigma }\llbracket {{\color{\colorSYNTAX}\mtexttt{real}}}[r]\rrbracket                   &{}\overset \vartriangle \iff {}& r_{1} = r_{2} = r
  \cr  r_{1} \sim  r_{2} &{}\in {}& {\mathcal{V}}_{n}^{\Sigma }\llbracket {{\color{\colorSYNTAX}\mtexttt{sreal}}}@\Sigma ^{\prime}\rrbracket                  &{}\overset \vartriangle \iff {}& r_{1}\sim ^{\Sigma \mathord{\cdotp }\Sigma ^{\prime}} r_{2}
  \cr  \langle v_{1 1},v_{1 2}\rangle  \sim  \langle v_{2 1},v_{2 2}\rangle  &{}\in {}& {\mathcal{V}}_{n}^{\Sigma }\llbracket \tau _{1} \times  \tau _{2}\rrbracket       &{}\overset \vartriangle \iff {}& v_{1 1} \sim  v_{2 1} \in  {\mathcal{V}}_{n}^{\Sigma }\llbracket \tau _{1}\rrbracket 
  \cr                        &{} {}&                       &{}\wedge {}& v_{1 2} \sim  v_{2 2} \in  {\mathcal{V}}_{n}^{\Sigma }\llbracket \tau _{2}\rrbracket 
  \cr  \langle v_{1 1},v_{1 2}\rangle  \sim  \langle v_{2 1},v_{2 2}\rangle  &{}\in {}& {\mathcal{V}}_{n}^{\Sigma }\llbracket (\sigma _{1}{\otimes }\sigma _{2})@\Sigma ^{\prime}\rrbracket  &{}\overset \vartriangle \iff {}& v_{1 1} \sim  v_{2 1} \in  {\mathcal{V}}_{n}^{\Sigma }\llbracket \sigma _{1}@\Sigma ^{\prime}\rrbracket 
  \cr                        &{} {}&                       &{}\wedge {}& v_{1 2} \sim  v_{2 2} \in  {\mathcal{V}}_{n}^{\Sigma }\llbracket \sigma _{2}@\Sigma ^{\prime}\rrbracket 
  \cr  v_{1 1}\mathrel{:: }v_{1 2} \sim  v_{2 1}\mathrel{:: }v_{2 2} &{}\in {}& {\mathcal{V}}_{n}^{\Sigma }\llbracket {{\color{\colorSYNTAX}\mtexttt{list}}}(\tau )\rrbracket         &{}\overset \vartriangle \iff {}& v_{1 1} \sim  v_{2 1} \in  {\mathcal{V}}_{n}^{\Sigma }\llbracket \tau \rrbracket 
  \cr                        &{} {}&                       &{}\wedge {}& v_{1 2} \sim  v_{2 2} \in  {\mathcal{V}}_{n}^{\Sigma }\llbracket {{\color{\colorSYNTAX}\mtexttt{list}}}(\tau )\rrbracket 
  \cr  v_{1 1}\mathrel{\hat {\mathrel{:: }}}v_{1 2} \sim  v_{2 1}\mathrel{\hat {\mathrel{:: }}}v_{2 2} &{}\in {}& {\mathcal{V}}_{n}^{\Sigma }\llbracket {{\color{\colorSYNTAX}\mtexttt{slist}}}(\sigma )@\Sigma ^{\prime}\rrbracket     &{}\overset \vartriangle \iff {}& v_{1 1} \sim  v_{2 1} \in  {\mathcal{V}}_{n}^{\Sigma }\llbracket \sigma @\Sigma ^{\prime}\rrbracket 
  \cr                        &{} {}&                       &{}\wedge {}& v_{1 2} \sim  v_{2 2} \in  {\mathcal{V}}_{n}^{\Sigma }\llbracket {{\color{\colorSYNTAX}\mtexttt{slist}}}(\sigma )@\Sigma ^{\prime}\rrbracket 
  \cr  \langle \lambda _{z}x.e_{1}{\mathrel{|}}\gamma _{1}\rangle  \sim  \langle \lambda _{z}x.e_{2}{\mathrel{|}}\gamma _{2}\rangle  &{}\in {}& {\mathcal{V}}_{n}^{\Sigma }\llbracket \tau _{1}\rightarrow \tau _{2}\rrbracket  &{}\overset \vartriangle \iff {}& \forall  n^{\prime}\leq n,v_{1},v_{2}.\hspace*{0.33em} v_{1} \sim  v_{2} \in  {\mathcal{V}}_{n^{\prime}}^{\Sigma }\llbracket \tau _{1}\rrbracket 
  \cr                          &{} {}&                   &{}  \Rightarrow {}& \{ x\mapsto v_{1},z\mapsto \langle \lambda _{z}x.e_{1}\mathrel{|}\gamma _{1}\rangle \} \uplus \gamma _{1},e_{1}
  \cr                          &{} {}&                     &{}\sim {}& \{ x\mapsto v_{2},z\mapsto \langle \lambda _{z}x.e_{2}\mathrel{|}\gamma _{2}\rangle \} \uplus \gamma _{2},e_{2}
  \cr                          &{} {}&                     &{}\in {}& {\mathcal{E}}_{n^{\prime}}^{\Sigma }\llbracket \tau _{2}\rrbracket 
  \cr  \rho _{1} \sim  \rho _{2} &{}\in {}& {\mathcal{V}}_{n}^{\Sigma }\llbracket {\scriptstyle \bigcirc }_{\Sigma ^{\prime}}(\tau )\rrbracket  &{}\overset \vartriangle \iff {}& \forall v.\hspace*{0.33em} \rho _{1}(v) \leq  {\textit{e}}^{|{}\rceil \Sigma \lceil {}^{1}\times \Sigma ^{\prime}|_{L\infty }}\rho _{2}(v)
  \end{array}
  \hfill\hspace{0pt}
  \mathllap{\begingroup\color{\colorTEXT}\boxed{\begingroup\color{\colorMATH} v \sim  v \in  {\mathcal{V}}_{n}^{\Sigma }\llbracket \tau \rrbracket  \endgroup}\endgroup}
\cr 
\cr \hfill\hspace{0pt} \gamma _{1} \sim  \gamma _{2} \in  {\mathcal{G}}_{n}^{\Sigma }\llbracket \Gamma \rrbracket  \overset \vartriangle \iff  \forall  x \in  {\mtext{dom}}(\gamma _{1}\cup \gamma _{2}).\hspace*{0.33em} \gamma _{1}(x) \sim  \gamma _{2}(x) \in  {\mathcal{V}}_{n}^{\Sigma }\llbracket \tau \rrbracket 
  \hfill\hspace{0pt}
  \begingroup\color{\colorTEXT}\boxed{\begingroup\color{\colorMATH} \gamma  \sim  \gamma  \in  {\mathcal{G}}_{n}^{\Sigma }\llbracket \Gamma \rrbracket  \endgroup}\endgroup
\end{tabularx}
\end{gather*}\endgroup
\end{framed}
\caption{Step-indexed Logical Relation.}
\label{fig:logical-relations}
\end{figure}
% }-}

\section{Implementation \& Case Studies}
\label{sec:case}

We have implemented \solo as a Haskell library in about 600 lines of code; it will be made open-source upon publication and will be submitted as an artifact. For our case studies, we introduce sensitive matrices \inline{SMatrix &$\sigma $& m r c a}, sensitive key-value mappings (dictionaries) \inline{SDict &$\sigma $& m a b}, and sensitive sets \inline{SSet &$\sigma $& a}, as well as sound primitive operations over these values. We provide the usual primitives over these types seen in prior work \cite{reed2010distance,near2019duet}.

We have implemented four case studies in \solo, to validate its applicability to real differentially private algorithms. Each case study algorithm has been previously verified using specialized type systems, but ours is the first static approach embedded in a mainstream language with this capability. Due to space limitations, we include the complete code listings for the case studies in Appendix~\ref{sec:case_appendix} in the supplemental material; our case studies are summarized as follows:
\begin{itemize}[leftmargin=5mm, itemsep=0pt]
\item \textbf{K-Means clustering} is an iterative clustering algorithm previously verified in \fuzz. \solo infers that one iteration is {{\color{\colorMATH}\ensuremath{3\epsilon }}}-DP, and can use advanced composition for total privacy cost.
\item \textbf{Cumulative distribution function}, originally verified in \dfuzz, uses a loop to form a CDF. The \solo version leverages our looping combinators.
\item \textbf{Gradient descent}, originally verified in \duet, demonstrates \solo's ability to use the Gaussian mechanism and R\'enyi differential privacy.
\item \textbf{Multiplicative-Weights Exponential Mechanism}, previously verified in \dduo, uses the exponential mechanism in addition to looping combinators.
\end{itemize}

\section{Related Work}
\label{sec:related}

%\subsection{Languages for Static Verification of Differential Privacy.}

\paragraph{Lightweight Static Analysis for Differential Privacy.}
The \dpella~\cite{lobo2020programming} system is closest to our work. Like \solo, \dpella uses Haskell's type system for sensitivity analysis, but \dpella implements a custom dynamic analysis of programs to compute privacy and accuracy information. \solo goes beyond \dpella by supporting calculation of privacy costs using Haskell's type system, in addition to sensitivity information.

\paragraph{Linear Types.}
\fuzz was the first language and type system designed to verify differential privacy
costs of a program, and did so by modeling sensitivity using linear types \cite{reed2010distance}. \dfuzz extended \fuzz with dependent types and automation aided by SMT solvers \cite{gaboardi2013linear}.
The \duet language extends \fuzz with support for advanced variants of differential privacy such as {{\color{\colorMATH}\ensuremath{(\epsilon , \delta )}}}-differential privacy \cite{near2019duet}.
Adaptive Fuzz embeds a static sensitivity analysis within a dynamic privacy analysis using privacy odometers and filters \cite{Winograd-CortHR17}.
The above approaches all require linear types, which are typically not available in mainstream programming languages. The Granule language~\cite{orchard2019quantitative} is specifically designed to support linear types, but has not yet been widely adopted by programmers.

\paragraph{Indexed Monadic Types.}
Azevedo de Amorim et al \cite{amorim2018} introduce a path construction to embed relational tracking for {{\color{\colorMATH}\ensuremath{(\epsilon , \delta )}}}-differential privacy within the \fuzz type system.
This technique internalizes \emph{group privacy} and can produce non-optimal privacy bounds for multi-argument programs.

\paragraph{Program Logics, Randomness Alignments, \& Probabilistic Couplings}
Program logics such as \apRHL~\cite{Barthe:POPL12,Barthe:TOPLAS:13} are very flexible and expressive but difficult to automate. \fuzzi~\cite{zhang2019fuzzi} combines the \fuzz type system (for composition of sensitivity and privacy operations) with \apRHL (for proofs of basic mechanisms) to eliminate the need for trusted primitives like the Laplace mechanism.
Approaches based on randomness alignments, such as LightDP \cite{zhang2017lightdp} and ShadowDP \cite{wang2019proving} are suitable for verifying low level techniques such as the sparse vector technique \cite{privacybook} but not for sensitivity analysis.
Barthe et al introduce an approach for proving differential privacy using a generalization of probabilistic couplings. They present several case studies in the \apRHLplus~\cite{Barthe:LICS16} language which extends program logics with approximate couplings. The technique of aligning randomness is also used in the coupling method.
Albarghouthi and Hsu~\cite{DBLP:journals/pacmpl/AlbarghouthiH18} use an alternative
approach based on randomness alignments as well as approximate couplings.
None of these approaches can be easily embedded in mainstream languages like Haskell.

%\subsection{Dynamic Enforcement of Differential Privacy}

\paragraph{Dynamic Analyses.}
\pinq~\cite{mcsherry2009} pioneered dynamic enforcement of differential privacy for a subset of relational database query tasks.
Featherweight PINQ~\cite{ebadi2015featherweight} is a framework which models \pinq and presents a proof that any programs which use its API are differentially private. ProPer~\cite{ebadi2015} is also based on \pinq, but is primarily designed to maintain a privacy budget for each individual in a database system. ProPer operates by silently dropping records from queries when their privacy budget is exceeded. UniTrax~\cite{munz2018} improves on ProPer by allowing per-user budgets without silently dropping records. UniTrax operates by tracking queries against an abstract database as opposed to the actual database records. Diffprivlib~\cite{holohan2019diffprivlib} (for Python) and
Google's library~\cite{wilson2020differentially} (for several
languages) provide differentially private algorithms for modern machine-learning and general data analysis.
{{\color{\colorMATH}\ensuremath{\epsilon }}}ktelo~\cite{zhang2018ektelo} describes differentially private programs as \emph{plans} over high level libraries of \emph{operators} which have classes for data transformation, reduction, inference and other tasks.
\dduo~\cite{abuah2021dduo} extends \pinq-style dynamic analysis to general-purpose Python programs.
Dynamic approaches require running the program in order to verify differential privacy, and in many cases add significant runtime overhead.

\paragraph{Dynamic Testing.}
Recent work by Bichsel et al.~\cite{bichsel2018dp}, Ding et al.~\cite{ding2018detecting}, Wang et al.~\cite{wang2020checkdp}, and Wilson et al.~\cite{wilson2020differentially} have given rise to a set of techniques which facilitate testing for differential privacy. These approaches work for arbitrary programs written in any language, but they typically involve evaluating a program many times on neighboring inputs to check for possible violations of differential privacy---which can be intractable for complex algorithms.

%\subsection{Security as a Library/Language Extension}

\paragraph{Static Taint Analysis and IFC.}
Li et al \cite{li2006} present an embedded security sublanguage in Haskell using the arrows combinator interface. Russo et al introduce a monadic library for light-weight information flow security in Haskell \cite{russo2008}. Crockett et al propose a domain specific language for safe homomorphic encryption in Haskell \cite{crockett2018}. Safe Haskell \cite{terei2012} is a Haskell language extension which implements various security policies as monads. Parker et al \cite{parker2019} introduce a Haskell framework for enforcing information flow control policies in database-oriented web applications.

\solo's sensitivity tracking is similar to approaches for tracking information flow, but it is more quantitative and follows a probabilistic programming structure (e.g. sampling from distributions). \solo thus has the structure of a taint analysis, but is refined to capture the specific information flow property of differential privacy. In particular, the sensitivity and privacy environments that we attach to the types of values can be seen as similar to IFC labels~\cite{myers1999, buiras2015, yang2012, tripp2009, li2014, wang2008, arzt2014, sridharan2011}

\section{Conclusion}
\label{sec:conclusion}

We have presented \solo, a lightweight static analysis approach for differential privacy. \solo can be embedded in mainstream functional languages, without the need for a specialized type system. We have proved the soundness (metric preservation) of \solo using a logical relation to establish function sensitivity. We have presented several case studies verifying differentially private algorithms seen in related work.

\clearpage

\bibliographystyle{ACM-Reference-Format}
\bibliography{refs,gdp}

\appendix
\newpage
\section{Type-Level Function Definitions}
\label{sec:type_level_defs}

This section contains the complete definitions of type-level functions on sensitivity environments. The \inline{Plus} function is defined in terms of the infix operator \inline{+++}.

\begin{minted}{haskell}
type family MaxSens (s :: SEnv) :: TL.Nat where
  MaxSens '[] = 0
  MaxSens ('(_,n)':s) = MaxNat n (MaxSens s)

type family ScaleSens (n :: TL.Nat) (s :: SEnv) :: SEnv where
  ScaleSens _ '[] = '[]
  ScaleSens n1 ('(o,n2) ': s) = '(o,n1 TL.* n2) ': ScaleSens n1 s

type family TruncateSens (n :: TL.Nat) (s :: SEnv) :: SEnv where
  TruncateSens _ '[] = '[]
  TruncateSens n1 ('(o,n2) ': s) = '(o,TruncateNat n1 n2) ': TruncateSens n1 s

-- compute the sum of two sensitivity environments by traversing each
-- association list, adding values that have the same key (third equation), and
-- keeping things in order when keys don't overlap (fourth equation)
type family (+++) (s1 :: SEnv) (s2 :: SEnv) :: SEnv where
  '[]            +++ s2             = s2
  s1             +++ '[]            = s1
  ('(o,n1)':s1)  +++ ('(o,n2)':s2)  = '(o,n1 TL.+ n2) ': (s1 +++ s2)
  ('(o1,n1)':s1) +++ ('(o2,n2)':s2) =
    Cond (IsLT (TL.CmpSymbol o1 o2)) ('(o1,n1) ': (s1 +++ ('(o2,n2)':s2)))
                                     ('(o2,n2) ': (('(o1,n1)':s1) +++ s2))
\end{minted}

\section{Case Studies}
\label{sec:case_appendix}

For our case studies, we introduce sensitive matrices \inline{SMatrix &$\sigma $& m r c a}, sensitive key-value mappings (dictionaries) \inline{SDict &$\sigma $& m a b}, and sensitive sets \inline{SSet &$\sigma $& a}, as well as sound primitive operations over these values. \inline{r c} are matrix dimensions, \inline{a b} are type parameters representing the contents of the compound types. \inline{&$\sigma $&} represents the sensitivity environments as usual, and \inline{m} represents the distance metric. We provide the usual primitives over these types seen in prior work \cite{reed2010distance,near2019duet}. Sets are assumed to use the Hamming metric, while matrices and key-value maps use the standard compound metrics discussed earlier: \inline{L1 | L2 | LInf}.
Recall that \inline{NatS} is a type for singleton naturals and \inline{natS @ 5} creates a singleton for the value \inline{5}.

\paragraph{Case study: k-means clustering.}
We present a case study based on the privacy-preserving implementation of the k-means clustering algorithm seen originally in Blum et al, as well as in the presentation of the \fuzz language. The goal of the k-means clustering algorithm is to iteratively find a set of $k$ clusters to which $n$ datapoints can be partitioned, where each datapoint belongs to the cluster with the nearest \emph{center} or \emph{centroid} to it.

The algorithm operates by beginning from an initial guess at the list of cluster centroids which it iteratively improves on. A single iteration consists of grouping each datapoint with the centroid it is closest to, then recalculating the mean of each group to initialize the next round's list of centroids. The algorithm applies the Laplace mechanism three times per iteration; \solo infers the privacy cost of one iteration to be {{\color{\colorMATH}\ensuremath{3\epsilon }}}, and we can use the advanced composition combinator introduced earlier to obtain privacy bounds when the algorithm runs for many iterations.

The \inline{assign} function is responsible for pairing each initial datapoint with the index of the centroid it is closest to in the initial centroid list. The \inline{partition} function then groups the set of datapoints into a list of sets, where each set represents a cluster. The rest of the algorithm proceeds to compute the private new center of each cluster. Given that our datapoints are two-dimensional, \inline{totx} and \inline{toty} sum the \inline{x} and \inline{y} coordinates of each cluster of datapoint. After we compute the size of each cluster, the \inline{avg} function calculates the new mean of each cluster with the three-element tuple of coordinate and size data zipped together for each cluster.

\vspace{1em}
\begin{minted}{haskell}
type Pt = (Double, Double)
-- helpers
assign :: [Pt] -> SSet #$\sigma $# Pt -> SSet #$\sigma $# (Pt,Integer)
ppartition :: SSet #$\sigma $# (Pt,Integer) -> SList #$\sigma $# m SSet (Set Pt)
totx :: SSet #$\beta $# Pt -> SDouble #$\beta $# 'AbsoluteM
toty :: SSet #$\beta $# Pt -> SDouble #$\beta $# 'AbsoluteM
size :: SSet #$\beta $# Pt -> SDouble #$\beta $# 'AbsoluteM
avg :: ((Double, Double), Double) ->  (Double, Double)

-- kmeans: 3$\epsilon $-private
iterate :: #$\forall $# m #$\sigma $#. (TL.KnownNat (MaxSens #$\sigma $#)) => SSet #$\sigma $# Pt -> [Pt] -> _
iterate b ms = do
  let b' = ppartition (assign ms b)
  tx <- vector_laplace @1 Proxy $ map#$_{0}$# totx b'
  ty <- vector_laplace @1 Proxy $ map#$_{0}$# toty b'
  t <- vector_laplace @1 Proxy $ map#$_{0}$# size b'
  let stats = zip (zip tx ty) t
  return $ (map avg stats)
\end{minted}

The Haskell typechecker can infer the privacy cost of one iteration of the k-means algorithm as $3\epsilon $.

\paragraph{Case study: Cumulative Distribution Function.}
Our next case study implements the private \inline{cdf} function as seen in \dfuzz \cite{mcsherry2010,gaboardi2013linear}. Given a database of numeric records, and a set of buckets associated with cutoff values, the \inline{cdf} function privately partitions each record to its respective bucket. As in \dfuzz, this case study demonstrates the ability of \solo to verify privacy costs which depend on a program input, in this case the symbolic number of buckets \inline{m}. However, our approach to achieve this feature relies on singleton types in Haskell, and does not require a \emph{true} dependent type system.

\vspace{1em}
\begin{minted}{haskell}
cdf :: #$\forall $# m o s #$\epsilon $#. (TL.KnownNat m,TL.KnownNat #$\epsilon $#) =>
     NatS m
  -> NatS #$\epsilon $#
  -> Matrix m 1 Double -- buckets
  -> SSet #$\sigma $# Double -- db
  -> EpsPrivacyMonad (ScalePriv m (TruncateSens #$\epsilon $# #$\sigma $#)) [Double]
cdf m t buckets db = do
  let f :: Double -> SSet #$\sigma $# Double
           -> _
      f = \x -> \db1 ->
        let (lt,gt) = bag_split (\k -> k < x) db1 in
          (laplace @#$\epsilon $# Proxy (natS @5) $ (bag_size lt), db)
      z = mloop#$_{1}$# m buckets db f $ return []
  z
\end{minted}

\paragraph{Case study: Gradient Descent.}

We now present a case study (Figure~\ref{fig:gd}) based on a simple machine learning algorithm \cite{BST} which performs gradient descent.

As inputs, the \inline{gd} algorithm accepts a list of feature vectors \inline{xs} representing sensitive user data, a set of corresponding classifier labels \inline{ys}, a number of iterations to run \inline{k} and the desired privacy cost per iteration \inline{&$\epsilon $&}. Gradient descent also requires a loss function which describes the accuracy of the current model in predicting the correct classification of user examples. The algorithm works by moving the current model in the opposite direction of the gradient of the loss function. In order to preserve privacy for this algorithm, we may introduce noise at the point where user data is exposed: the gradient calculation. The let-bound function \inline{f} in the \inline{gd} algorithm contains the workload of a single iteration of the program: in which we perform the gradient calculation and introduce noise using the vector-based Laplacian mechanism.

\begin{figure}
\begin{minted}{haskell}

-- sequential composition privacy loop over a matrix
mloop :: NatS k
  -> SMatrix #$\sigma $# LInf 1 n SDouble
  -> (SMatrix #$\sigma $# LInf 1 n SDouble ->
      EpsPrivacyMonad (TruncateSens #$\epsilon $# #$\sigma $#) (Matrix 1 n Double))
  -> EpsPrivacyMonad (ScalePriv k (TruncateSens #$\epsilon $# #$\sigma $#)) (Matrix 1 n Double)

-- gradient descent algorithm
gd :: NatS k
  -> NatS #$\epsilon $#
  -> SMatrix #$\sigma $# LInf m n SDouble
  -> SMatrix #$\sigma $# LInf m 1 SDouble
  -> EpsPrivacyMonad (ScalePriv k (TruncateSens #$\epsilon $# #$\sigma $#)) (Matrix 1 n Double)
gd k t xs ys = do
  let #$m_{0}$# = matrix (sn32 @ 1) (sn32 @ n) $ \ i j -> 0
      cxs = mclip xs (natS @ 1)
  let f :: SMatrix #$\sigma _{1}$# LInf 1 n SDouble
    -> EpsPrivacyMonad (TruncateSens #$\epsilon $# #$\sigma _{1}$#) (Matrix 1 n Double)
      f = \#$\theta $# -> let g = mlaplace @#$\epsilon $# Proxy (natS @5) $ xgradient #$\theta $# cxs ys
      in msubM (return #$\theta $#) g
      z = mloop @(TruncateNat t 1) k (sourceM $ xbp #$m_{0}$#) f
  z
\end{minted}
\caption{Gradient Descent}
\label{fig:gd}
\end{figure}

\paragraph{Case study: Multiplicative-Weights Exponential Mechanism.}

Our final case study, the MWEM algorithm \cite{hardt2012simple}, builds a differentially private synthetic dataset which approximates some sensitive real dataset with some level of accuracy. The algorithm combines usage of the Exponential Mechanism, Laplacian noise, and the multiplicative-weights update rule to construct a noisy synthetic dataset over several iterations with competitive privacy leakage bounds via composition.

\inline{mwem} (Figure~\ref{fig:mwem}) takes the following inputs: a number of iterations \inline{k}, a privacy cost \inline{&$\epsilon $&} to be used by the exponential mechanism and Laplace, \inline{real_data} the sensitive information dataset, a query workload \inline{queries} over the sensitive dataset, and lastly \inline{syn_data} which represents a uniform or random distribution over the domain of the real dataset.

Each iteration, the \inline{mwem} algorithm selects a query from the query workload privately using the exponential mechanism. The query selected is selected by virtue of a scoring function which determines that the result of the query on the synthetic dataset greatly differs from its result on the real dataset (more so than other queries in the workload, with some amount of error). The algorithm updates the synthetic dataset using the multiplicative weights update rule, based on the query result on the real dataset with some noise added. This process continues over several iterations until the synthetic dataset reaches some some level of accuracy relative to the real dataset.

\begin{figure}
\vspace{1em}
\begin{minted}{haskell}

-- exponential mechanism
expmech :: [(Double,Double)]
  -> NatS #$\epsilon $#
  -> SDict #$\sigma $# LInf SDouble SDouble
  -> EpsPrivacyMonad (TruncateSens #$\epsilon $# #$\sigma $#) Int

-- exponential mechanism + laplace loop
expnloop :: NatS k
  -> NatS #$\epsilon $#
  -> [(Double,Double)]
  -> SDict #$\sigma $# LInf SDouble SDouble
  -> Map.Map Double Double
  -> EpsPrivacyMonad (ScalePriv (2 TL.* k) (TruncateSens #$\epsilon $# #$\sigma $#)) (Map.Map Double Double)

-- multiplicative-weights exponential mechanism
mwem :: NatS k
  -> NatS #$\epsilon $#
  -> [(Double,Double)]
  -> SDict #$\sigma $# LInf SDouble SDouble
  -> Map Double Double
  -> EpsPrivacyMonad (ScalePriv (2 TL.* k) (TruncateSens #$\epsilon $# #$\sigma $#)) (Map Double Double)
mwem k #$\epsilon $# queries real_data syn_data =
  expnloop k #$\epsilon $# queries real_data syn_data
\end{minted}
\caption{Multiplicative Weights Exponential Mechanism.}
\label{fig:mwem}
\end{figure}

\section{Lemmas,  Theorems \& Proofs}
\label{sec:proof}

\begingroup
\setlength{\parindent}{0pt}

% {-{ LEMMA - Plus Respects - thm:plus-respects
\begin{lemma}[Plus Respects]\label{thm:plus-respects}\ \\
  If {{\color{\colorMATH}\ensuremath{r_{1} \sim ^{r} r_{2}}}} then {{\color{\colorMATH}\ensuremath{r_{1} + r_{3} \sim ^{r} r_{2} + r_{3}}}}.
\end{lemma}
\begin{proof}\ \\
  By {{\color{\colorMATH}\ensuremath{|r_{1} - r_{2}| \leq  r \implies  |(r_{1} + r_{3}) - (r_{2} + r_{3})| \leq  r}}}.
\end{proof}
% }-}

% {-{ LEMMA - Times Respects - thm:times-respects
\begin{lemma}[Times Respects]\label{thm:times-respects}\ \\
  If {{\color{\colorMATH}\ensuremath{r_{1} \sim ^{r} r_{2}}}} then {{\color{\colorMATH}\ensuremath{r_{3}r_{1} \sim ^{r_{3}r} r_{3}r_{2}}}}.
\end{lemma}
\begin{proof}\ \\
  By {{\color{\colorMATH}\ensuremath{|r_{1} - r_{2}| \leq  r \implies  |r_{3}r_{1} - r_{3}r_{2}| \leq  r}}}.
\end{proof}
% }-}

% {-{ LEMMA - Triangle - thm:triangle
\begin{lemma}[Triangle]\label{thm:triangle}\ \\
  If {{\color{\colorMATH}\ensuremath{r_{1} \sim ^{r_{A}} r_{2}}}} and {{\color{\colorMATH}\ensuremath{r_{2} \sim ^{r_{B}} r_{3}}}} then {{\color{\colorMATH}\ensuremath{r_{1} \sim ^{r_{A} + r_{B}} r_{3}}}}.
\end{lemma}
\begin{proof}\ \\
  By the classic triangle inequality lemma for real numbers.
\end{proof}
% }-}

% {-{ LEMMA - Step-index Weakening - thm:step-index-weakening
\begin{lemma}[Step-index Weakening]\label{thm:step-index-weakening}\ \\
  For {{\color{\colorMATH}\ensuremath{n' \leq  n}}}: (1) If {{\color{\colorMATH}\ensuremath{\gamma _{1} \sim  \gamma _{2} \in  {\mathcal{G}}_{n}^{\Sigma }\llbracket \Gamma \rrbracket }}} then {{\color{\colorMATH}\ensuremath{\gamma _{1} \sim  \gamma _{2} \in  {\mathcal{G}}_{n^{\prime}}^{\Sigma }\llbracket \Gamma \rrbracket }}}; and
  (2) If {{\color{\colorMATH}\ensuremath{v_{1} \sim  v_{2} \in  {\mathcal{V}}_{n}^{\Sigma }\llbracket \tau \rrbracket }}} then {{\color{\colorMATH}\ensuremath{v_{1} \sim  v_{2} \in  {\mathcal{V}}_{n^{\prime}}^{\Sigma }\llbracket \tau \rrbracket }}}; and (3) If {{\color{\colorMATH}\ensuremath{\gamma _{1},e_{1} \sim 
  \gamma _{2},e_{2} \in  {\mathcal{E}}_{n}^{\Sigma }\llbracket \tau \rrbracket }}} then {{\color{\colorMATH}\ensuremath{\gamma _{1},e_{1} \sim  \gamma _{2},e_{2} \in  {\mathcal{E}}_{n}^{\Sigma }\llbracket \tau \rrbracket }}}.
\end{lemma}
\begin{proof}\ \\
  By induction on {{\color{\colorMATH}\ensuremath{n}}} mutually for all properties; case analysis on {{\color{\colorMATH}\ensuremath{v_{1}}}} and
  {{\color{\colorMATH}\ensuremath{v_{2}}}} for property (2), and case analysis on {{\color{\colorMATH}\ensuremath{e_{1}}}} and {{\color{\colorMATH}\ensuremath{e_{2}}}} for property (3).
\end{proof}
% }-}

% {-{ THEOREM - Metric Preservation - thm:metric-preservation
\begin{theorem}[Metric Preservation]\label{thm:metric-preservation}\ \\
  \begin{itemize}[label={},leftmargin=0pt]\item  \begin{tabular}{r@{\hspace*{1.00em}}l@{\hspace*{1.00em}}l
     } If:      & {{\color{\colorMATH}\ensuremath{\gamma _{1} \sim  \gamma _{2} \in  {\mathcal{G}}_{n}^{\Sigma }\llbracket \Gamma \rrbracket }}}     & {{\color{\colorTEXT}\textnormal{{\mtextit{(H1)}}}}}
     \cr  And:     & {{\color{\colorMATH}\ensuremath{ \Gamma  \vdash  e \mathrel{:} \tau }}}             & {{\color{\colorTEXT}\textnormal{{\mtextit{(H2)}}}}}
     \cr  Then:    & {{\color{\colorMATH}\ensuremath{\gamma _{1},e \sim  \gamma _{2},e \in  {\mathcal{E}}_{n}^{\Sigma }\llbracket \tau \rrbracket }}} &
     \end{tabular}
  \item 
  \item  That is, either {{\color{\colorMATH}\ensuremath{n=0}}}, or {{\color{\colorMATH}\ensuremath{n = n^{\prime}+1}}} and\ldots 
  \item 
  \item  \begin{tabular}{r@{\hspace*{1.00em}}l@{\hspace*{1.00em}}l
     } If:      & {{\color{\colorMATH}\ensuremath{n^{\prime \prime} \leq  n^{\prime}}}}                   & {{\color{\colorTEXT}\textnormal{{\mtextit{(H3)}}}}}
     \cr  And:     & {{\color{\colorMATH}\ensuremath{\gamma _{1} \vdash  e \Downarrow _{n^{\prime \prime}} v_{1}}}}           & {{\color{\colorTEXT}\textnormal{{\mtextit{(H4)}}}}}
     \cr  Then:    & {{\color{\colorMATH}\ensuremath{\exists !v_{2}.\hspace*{0.33em} \gamma _{2} \vdash  e \Downarrow _{n^{\prime \prime}} v_{2}}}}     & {{\color{\colorTEXT}\textnormal{{\mtextit{(C1)}}}}}
     \cr  And:     & {{\color{\colorMATH}\ensuremath{v_{1} \sim  v_{2} \in  {\mathcal{V}}_{n^{\prime}-n^{\prime \prime}}^{\Sigma }\llbracket \tau \rrbracket }}}  & {{\color{\colorTEXT}\textnormal{{\mtextit{(C2)}}}}}
     \end{tabular}
  \end{itemize}
\end{theorem}
% {-{ PROOF - Fundamental Property
\begin{proof}\ \\
  By strong induction on {{\color{\colorMATH}\ensuremath{n}}} and case analysis on {{\color{\colorMATH}\ensuremath{e}}} and {{\color{\colorMATH}\ensuremath{\tau }}}:
  \begin{itemize}[label=\textbf{-},leftmargin=*]\item  \begin{itemize}[label={},leftmargin=0pt]\item  {\mtextbf{Case}} {{\color{\colorMATH}\ensuremath{n=0}}}: % {-{
        Trivial by definition.
     \end{itemize} % }-}
  \item  \begin{itemize}[label={},leftmargin=0pt]\item  {\mtextbf{Case}} {{\color{\colorMATH}\ensuremath{n=n^{\prime}+1}}} {\mtextbf{and}} {{\color{\colorMATH}\ensuremath{e=x}}}: % {-{
     \item  By inversion on {\mtextit{(H4)}} we have:
        {{\color{\colorMATH}\ensuremath{n^{\prime} = 0}}} and
        {{\color{\colorMATH}\ensuremath{v_{1} = \gamma _{1}(x)}}}.
        Instantiate {{\color{\colorMATH}\ensuremath{v_{2} = \gamma _{2}(x)}}} in the conclusion.
        To show:
        {\mtextit{(C1)}}: {{\color{\colorMATH}\ensuremath{\gamma _{2} \vdash  x \Downarrow _{0} \gamma _{2}(x)}}} unique; and
        {\mtextit{(C2)}}: {{\color{\colorMATH}\ensuremath{\gamma _{1}(x) \sim  \gamma _{2}(x) \in  {\mathcal{V}}_{n^{\prime}}^{\Sigma }\llbracket \tau \rrbracket }}}.
        {\mtextit{(C1)}} is by {\mtextsc{ e-var}} application and inversion.
        {\mtextit{(C2)}} is by {\mtextit{(H1)}} and \nameref{thm:step-index-weakening}.
     \end{itemize} % }-}
  \item  \begin{itemize}[label={},leftmargin=0pt]\item  {\mtextbf{Case}} {{\color{\colorMATH}\ensuremath{n=n^{\prime}+1}}} {\mtextbf{and}} {{\color{\colorMATH}\ensuremath{e=r}}} {\mtextbf{and}} {{\color{\colorMATH}\ensuremath{\tau  = {{\color{\colorSYNTAX}\mtexttt{real}}}}}}: % {-{
     \item  By inversion on {\mtextit{(H4)}} we have:
        {{\color{\colorMATH}\ensuremath{n^{\prime} = 0}}} and
        {{\color{\colorMATH}\ensuremath{v_{1} = r}}}.
        Instantiate {{\color{\colorMATH}\ensuremath{v_{2} = r}}} in the conclusion.
        To show:
        {\mtextit{(C1)}}: {{\color{\colorMATH}\ensuremath{\gamma _{2} \vdash  r \Downarrow _{0} r}}} unique; and
        {\mtextit{(C2)}}: {{\color{\colorMATH}\ensuremath{r = r}}}.
        {\mtextit{(C1)}} is by {\mtextsc{ e-real}} application and inversion.
        {\mtextit{(C2)}} is trivial.
     \end{itemize} % }-}
  \item  \begin{itemize}[label={},leftmargin=0pt]\item  {\mtextbf{Case}} {{\color{\colorMATH}\ensuremath{n=n^{\prime}+1}}} {\mtextbf{and}} {{\color{\colorMATH}\ensuremath{e=r}}} {\mtextbf{and}} {{\color{\colorMATH}\ensuremath{\tau  = {{\color{\colorSYNTAX}\mtexttt{sreal}}}@\varnothing }}}: % {-{
     \item  By inversion on {\mtextit{(H4)}} we have:
        {{\color{\colorMATH}\ensuremath{n^{\prime} = 0}}} and
        {{\color{\colorMATH}\ensuremath{v_{1} = r}}}.
        Instantiate {{\color{\colorMATH}\ensuremath{v_{2} = r}}} in the conclusion.
        To show:
        {\mtextit{(C1)}}: {{\color{\colorMATH}\ensuremath{\gamma _{2} \vdash  r \Downarrow _{0} r}}} unique; and
        {\mtextit{(C2)}}: {{\color{\colorMATH}\ensuremath{r \sim ^{0} r}}}.
        {\mtextit{(C1)}} is by {\mtextsc{ e-sreal}} application and inversion.
        {\mtextit{(C2)}} is immediate by {{\color{\colorMATH}\ensuremath{|r-r| = 0 \leq  0}}}.
     \end{itemize} % }-}
  \item  \begin{itemize}[label={},leftmargin=0pt]\item  {\mtextbf{Case}} {{\color{\colorMATH}\ensuremath{n=n^{\prime}+1}}} {\mtextbf{and}} {{\color{\colorMATH}\ensuremath{e={\mtext{sing}}(r)}}} {\mtextbf{and}} {{\color{\colorMATH}\ensuremath{\tau  = {{\color{\colorSYNTAX}\mtexttt{real}}}[r]}}}: % {-{
     \item  By inversion on {\mtextit{(H4)}} we have:
        {{\color{\colorMATH}\ensuremath{n^{\prime} = 0}}} and
        {{\color{\colorMATH}\ensuremath{v_{1} = r}}}.
        Instantiate {{\color{\colorMATH}\ensuremath{v_{2} = r}}} in the conclusion.
        To show:
        {\mtextit{(C1)}}: {{\color{\colorMATH}\ensuremath{\gamma _{2} \vdash  {\mtext{sing}}(r) \Downarrow _{0} r}}} unique; and
        {\mtextit{(C2)}}: {{\color{\colorMATH}\ensuremath{r = r}}}.
        {\mtextit{(C1)}} is by {\mtextsc{ e-sing}} application and inversion.
        {\mtextit{(C2)}} is immediate.
     \end{itemize} % }-}
  \item  \begin{itemize}[label={},leftmargin=0pt]\item  {\mtextbf{Case}} {{\color{\colorMATH}\ensuremath{n=n^{\prime}+1}}} {\mtextbf{and}} {{\color{\colorMATH}\ensuremath{e=e_{1}+e_{2}}}} {\mtextbf{and}} {{\color{\colorMATH}\ensuremath{\tau ={{\color{\colorSYNTAX}\mtexttt{real}}}}}}: % {-{
     \item  By inversion on {\mtextit{(H4)}}:
     \item  \ {{\color{\colorMATH}\ensuremath{\begin{array}{rclcl@{\hspace*{1.00em}}l
           } \gamma _{1} &{}\vdash {}& e_{1} &{}\Downarrow _{n_{1}}{}& r_{1 1} & {{\color{\colorTEXT}\textnormal{{\mtextit{(H4.1)}}}}}
           \cr  \gamma _{1} &{}\vdash {}& e_{2} &{}\Downarrow _{n_{2}}{}& r_{1 2} & {{\color{\colorTEXT}\textnormal{{\mtextit{(H4.2)}}}}}
           \end{array}}}}
     \item  and we also have:
        {{\color{\colorMATH}\ensuremath{n^{\prime} = n_{1} + n_{2}}}},
        {{\color{\colorMATH}\ensuremath{v_{1} = r_{1 1}+r_{1 2}}}} and
        By IH ({{\color{\colorMATH}\ensuremath{n = n_{i}}}} decreasing), {\mtextit{(H1)}}, {\mtextit{(H2)}}, {\mtextit{(H3)}}, {\mtextit{(H4.1)}} and {\mtextit{(H4.2)}} we have:
        {{\color{\colorMATH}\ensuremath{\gamma _{2} \vdash  e_{1} \Downarrow _{n_{1}} r_{2 1}}}} (unique) {\mtextit{(IH.C1.1)}};
        {{\color{\colorMATH}\ensuremath{\gamma _{2} \vdash  e_{2} \Downarrow _{n_{2}} r_{2 2}}}} (unique) {\mtextit{(IH.C1.2)}};
        {{\color{\colorMATH}\ensuremath{r_{1 1} = r_{2 1}}}} {\mtextit{(IH.C2.1)}} ; and
        {{\color{\colorMATH}\ensuremath{r_{1 2} = r_{2 2}}}} {\mtextit{(IH.C2.2)}} .
        Instantiate {{\color{\colorMATH}\ensuremath{v_{2} = r_{2 1} + r_{2 2}}}}.
        To show:
        {\mtextit{(C1)}}: {{\color{\colorMATH}\ensuremath{\gamma _{2} \vdash  e_{1} + e_{2} \Downarrow _{n_{1}+n_{2}} r_{2 1} + r_{2 2}}}} (unique); and
        {\mtextit{(C2)}}: {{\color{\colorMATH}\ensuremath{r_{1 1} + r_{1 2} = r_{2 1} + r_{2 2}}}}.
        {\mtextit{(C1)}} is by {{\color{\colorMATH}\ensuremath{(IH.C1.1)}}}, {{\color{\colorMATH}\ensuremath{(IH.C1.2)}}}, and {\mtextsc{ e-plus}} application and inversion.
        {\mtextit{(C2)}} is by {\mtextit{(IH.C2.1)}} and {\mtextit{(IH.C2.2)}}.
     \end{itemize} % }-}
  \item  \begin{itemize}[label={},leftmargin=0pt]\item  {\mtextbf{Case}} {{\color{\colorMATH}\ensuremath{n=n^{\prime}+1}}} {\mtextbf{and}} {{\color{\colorMATH}\ensuremath{e=e_{1}+e_{2}}}} {\mtextbf{and}} {{\color{\colorMATH}\ensuremath{\tau  = {{\color{\colorSYNTAX}\mtexttt{sreal}}}@\Sigma ^{\prime}}}}: % {-{
     \item  By inversion on {\mtextit{(H2)}} and {\mtextit{(H4)}} we have:
     \item  \ {{\color{\colorMATH}\ensuremath{\begin{array}{rclcl@{\hspace*{1.00em}}l
           } \Gamma   &{}\vdash {}& e_{1} &{}\mathrel{:}{}& {{\color{\colorSYNTAX}\mtexttt{sreal}}}@\Sigma _{1} & {{\color{\colorTEXT}\textnormal{{\mtextit{(H2.1)}}}}}
           \cr  \Gamma   &{}\vdash {}& e_{2} &{}\mathrel{:}{}& {{\color{\colorSYNTAX}\mtexttt{sreal}}}@\Sigma _{2} & {{\color{\colorTEXT}\textnormal{{\mtextit{(H2.2)}}}}}
           \end{array}}}}
     \item  \ {{\color{\colorMATH}\ensuremath{\begin{array}{rclcl@{\hspace*{1.00em}}l
           } \gamma _{1} &{}\vdash {}& e_{1} &{}\Downarrow _{n_{1}}{}& r_{1 1}    & {{\color{\colorTEXT}\textnormal{{\mtextit{(H4.1)}}}}}
           \cr  \gamma _{1} &{}\vdash {}& e_{2} &{}\Downarrow _{n_{2}}{}& r_{1 2}    & {{\color{\colorTEXT}\textnormal{{\mtextit{(H4.2)}}}}}
           \end{array}}}}
     \item  and we also have:
        {{\color{\colorMATH}\ensuremath{\Sigma ^{\prime} = \Sigma _{1} + \Sigma _{2}}}},
        {{\color{\colorMATH}\ensuremath{n^{\prime} = n_{1} + n_{2}}}},
        {{\color{\colorMATH}\ensuremath{v_{1} = r_{1 1}+r_{1 2}}}} and
        By IH ({{\color{\colorMATH}\ensuremath{n = n_{i}}}} decreasing), {\mtextit{(H1)}}, {\mtextit{(H2)}}, {\mtextit{(H3)}}, {\mtextit{(H4.1)}} and {\mtextit{(H4.2)}} we have:
        {{\color{\colorMATH}\ensuremath{\gamma _{2} \vdash  e_{1} \Downarrow _{n_{1}} r_{2 1}}}} (unique) {\mtextit{(IH.C1.1)}};
        {{\color{\colorMATH}\ensuremath{\gamma _{2} \vdash  e_{2} \Downarrow _{n_{2}} r_{2 2}}}} (unique) {\mtextit{(IH.C1.2)}};
        {{\color{\colorMATH}\ensuremath{r_{1 1} \sim ^{\Sigma \mathord{\cdotp }\Sigma _{1}} r_{2 1}}}} {\mtextit{(IH.C2.1)}}; and
        {{\color{\colorMATH}\ensuremath{r_{2 1} \sim ^{\Sigma \mathord{\cdotp }\Sigma _{2}} r_{2 2}}}} {\mtextit{(IH.C2.2)}}.
        Instantiate {{\color{\colorMATH}\ensuremath{v_{2} = r_{2 1} + r_{2 2}}}}.
        To show:
        {\mtextit{(C1)}}: {{\color{\colorMATH}\ensuremath{\gamma _{2} \vdash  e_{1} + e_{2} \Downarrow _{n_{1}+n_{2}} r_{2 1} + r_{2 2}}}} (unique); and
        {\mtextit{(C2)}}: {{\color{\colorMATH}\ensuremath{r_{1 1} + r_{1 2} \sim ^{\Sigma \mathord{\cdotp }(\Sigma _{1} + \Sigma _{2})} r_{2 1} + r_{2 2}}}}.
        {\mtextit{(C1)}} is by {\mtextit{(IH.C1.1)}}, {\mtextit{(IH.C1.2)}}, and {\mtextsc{ e-plus}} application and inversion.
        {\mtextit{(C2)}} is by {\mtextit{(IH.C2.1)}}, {\mtextit{(IH.C2.2)}}, \nameref{thm:plus-respects} and \nameref{thm:triangle}.
     \end{itemize} % }-}
  \item  \begin{itemize}[label={},leftmargin=0pt]\item  {\mtextbf{Case}} {{\color{\colorMATH}\ensuremath{n=n^{\prime}+1}}} {\mtextbf{and}} {{\color{\colorMATH}\ensuremath{e=e_{1}\ltimes e_{2}}}} {\mtextbf{and}} either {{\color{\colorMATH}\ensuremath{\tau  = {{\color{\colorSYNTAX}\mtexttt{real}}}}}} {\mtextbf{or}} {{\color{\colorMATH}\ensuremath{\tau  = {{\color{\colorSYNTAX}\mtexttt{sreal}}}@\Sigma ^{\prime}}}}: % {-{
     \item  Similar to previous two cases, using \nameref{thm:times-respects}
        instead of \nameref{thm:plus-respects}.
     \end{itemize} % }-}
  \item  \begin{itemize}[label={},leftmargin=0pt]\item  {\mtextbf{Case}} {{\color{\colorMATH}\ensuremath{n=n^{\prime}+1}}} {\mtextbf{and}} {{\color{\colorMATH}\ensuremath{e = {{\color{\colorSYNTAX}\mtexttt{if0}}}(e_{1})\{ e_{2}\} \{ e_{3}\} }}}: % {-{
     \item  By inversion on {\mtextit{(H4)}} we have 2 subcases, each which induce:
     \item  \ {{\color{\colorMATH}\ensuremath{\begin{array}{rclcl@{\hspace*{1.00em}}l
           } \gamma _{1} &{}\vdash {}& e_{1} &{}\Downarrow _{n_{1}}{}& b_{1} & {{\color{\colorTEXT}\textnormal{{\mtextit{(H4.1)}}}}}
           \end{array}}}}
     \item  By IH ({{\color{\colorMATH}\ensuremath{n = n_{1}}}} decreasing), {\mtextit{(H1)}}, {\mtextit{(H2)}}, {\mtextit{(H3)}} and {\mtextit{(H4.1)}} we have:
        {{\color{\colorMATH}\ensuremath{\gamma _{2} \vdash  e_{1} \Downarrow _{n_{1}} b_{2}}}} (unique) {\mtextit{(IH.1.C1)}}; and
        {{\color{\colorMATH}\ensuremath{b_{1} = b_{2}}}} {\mtextit{(IH.1.C2)}}.
     \item  \begin{itemize}[label=\textbf{-},leftmargin=*]\item  \begin{itemize}[label={},leftmargin=0pt]\item  {\mtextbf{Subcase}} {{\color{\colorMATH}\ensuremath{b_{1} = b_{2} = {\mtext{true}}}}}:
           \item  From prior inversion on {\mtextit{(H4)}} we also have:
           \item  \ {{\color{\colorMATH}\ensuremath{\begin{array}{rclcl@{\hspace*{1.00em}}l
                 } \gamma _{1} &{}\vdash {}& e_{2} &{}\Downarrow _{n_{2}}{}& v_{1}             & {{\color{\colorTEXT}\textnormal{{\mtextit{(H4.2)}}}}}
                 \end{array}}}}
           \item  By IH ({{\color{\colorMATH}\ensuremath{n = n_{2}}}} decreasing), {\mtextit{(H1)}}, {\mtextit{(H2)}}, {\mtextit{(H3)}} and {\mtextit{(H4.2)}} we have:
              {{\color{\colorMATH}\ensuremath{\gamma _{2} \vdash  e_{2} \Downarrow _{n_{2}} v_{2}}}} (unique) {\mtextit{(IH.2.C1)}}; and
              {{\color{\colorMATH}\ensuremath{v_{1} \sim  v_{2} \in  {\mathcal{V}}_{n_{1}^{\Sigma } + n_{2}}\llbracket \tau \rrbracket }}} {\mtextit{(IH.2.C2)}}.
              Instantiate {{\color{\colorMATH}\ensuremath{v_{2} = v_{2}}}}.
              To show:
              {\mtextit{(C1)}}: {{\color{\colorMATH}\ensuremath{\gamma _{2} \vdash  {{\color{\colorSYNTAX}\mtexttt{if}}}(e_{1})\{ e_{2}\} \{ e_{3}\}  \Downarrow _{n_{1}+n_{2}} v_{2}}}}; and
              {\mtextit{(C2)}}: {{\color{\colorMATH}\ensuremath{v_{1} \sim  v_{2} \in  {\mathcal{V}}_{n_{1}^{\Sigma } + n_{2}}\llbracket \tau \rrbracket }}}.
              {\mtextit{(C1)}} is by {\mtextit{(IH.1.C1)}}, {\mtextit{(IH.2.C2)}} and {\mtextsc{ e-if-true}} application and inversion.
              {\mtextit{(C2)}} is by {\mtextit{(IH.2.C2)}}.
           \end{itemize}
        \item  \begin{itemize}[label={},leftmargin=0pt]\item  {\mtextbf{Subcase}} {{\color{\colorMATH}\ensuremath{b_{1} = b_{2} = {\mtext{false}}}}}:
           \item  Analogous to case {{\color{\colorMATH}\ensuremath{b_{1} = b_{2} = {\mtext{true}}}}}.
           \end{itemize}
        \end{itemize}
     \end{itemize} % }-}
  \item  \begin{itemize}[label={},leftmargin=0pt]\item  {\mtextbf{Case}} {{\color{\colorMATH}\ensuremath{n=n^{\prime}+1}}} {\mtextbf{and}} either {{\color{\colorMATH}\ensuremath{e=\langle e_{1},e_{2}\rangle }}} {\mtextbf{and}} {{\color{\colorMATH}\ensuremath{\tau  = \tau _{1} \times  \tau _{2}}}} {\mtextbf{or}} % {-{
        {{\color{\colorMATH}\ensuremath{e = \hat \langle e_{1},e_{2}\hat \rangle }}} {\mtextbf{and}} {{\color{\colorMATH}\ensuremath{\tau  = (\sigma _{1} \otimes  \sigma _{2})@\Sigma ^{\prime}}}}:
     \item  Analogous to cases for {{\color{\colorMATH}\ensuremath{e = e_{1} + e_{2}}}} where {{\color{\colorMATH}\ensuremath{\tau  = {{\color{\colorSYNTAX}\mtexttt{real}}}}}} or {{\color{\colorMATH}\ensuremath{\tau  =
        {{\color{\colorSYNTAX}\mtexttt{sreal}}}@\Sigma ^{\prime}}}}, and instead of appealing to \nameref{thm:triangle},
        appealing to the definition of the logical relation.
     \end{itemize} % }-}
  \item  \begin{itemize}[label={},leftmargin=0pt]\item  {\mtextbf{Case}} {{\color{\colorMATH}\ensuremath{n=n^{\prime}+1}}} {\mtextbf{and}} either {{\color{\colorMATH}\ensuremath{e=\pi _{i}(e)}}} {\mtextbf{or}} {{\color{\colorMATH}\ensuremath{e=\hat \pi _{i}(e)}}}): % {-{
     \item  Analogous to cases for {{\color{\colorMATH}\ensuremath{e = e_{1} + e_{2}}}} where {{\color{\colorMATH}\ensuremath{\tau  = {{\color{\colorSYNTAX}\mtexttt{real}}}}}} or {{\color{\colorMATH}\ensuremath{\tau  =
        {{\color{\colorSYNTAX}\mtexttt{sreal}}}@\Sigma ^{\prime}}}}, and instead of appealing to \nameref{thm:triangle},
        appealing to the definition of the logical relation.
     \end{itemize} % }-}
  \item  \begin{itemize}[label={},leftmargin=0pt]\item  {\mtextbf{Case}} {{\color{\colorMATH}\ensuremath{n=n^{\prime}+1}}} {\mtextbf{and}} either {{\color{\colorMATH}\ensuremath{e=e_{1}\mathrel{:: }e_{2}}}} {\mtextbf{and}} {{\color{\colorMATH}\ensuremath{\tau  = {{\color{\colorSYNTAX}\mtexttt{list}}}(\tau )}}} {\mtextbf{or}} % {-{
        {{\color{\colorMATH}\ensuremath{e = e_{1}\mathrel{\hat {\mathrel{:: }}}e_{2}}}} {\mtextbf{and}} {{\color{\colorMATH}\ensuremath{\tau  = {{\color{\colorSYNTAX}\mtexttt{slist}}}(\sigma )@\Sigma ^{\prime}}}}:
     \item  Analogous to cases for {{\color{\colorMATH}\ensuremath{e = e_{1} + e_{2}}}} where {{\color{\colorMATH}\ensuremath{\tau  = {{\color{\colorSYNTAX}\mtexttt{real}}}}}} or {{\color{\colorMATH}\ensuremath{\tau  =
        {{\color{\colorSYNTAX}\mtexttt{sreal}}}@\Sigma ^{\prime}}}}, and instead of appealing to \nameref{thm:triangle},
        appealing to the definition of the logical relation.
     \end{itemize} % }-}
  \item  \begin{itemize}[label={},leftmargin=0pt]\item  {\mtextbf{Case}} {{\color{\colorMATH}\ensuremath{n=n^{\prime}+1}}} {\mtextbf{and}} either {{\color{\colorMATH}\ensuremath{e={{\color{\colorSYNTAX}\mtexttt{case}}}(e_{1})\{ [].e_{2}\} \{ x_{1}\mathrel{:: }x_{2}.e_{3}\} }}} {\mtextbf{or}} % {-{
        {{\color{\colorMATH}\ensuremath{e = {{\color{\colorSYNTAX}\mtexttt{case}}}(e_{1})\{ \hat {[}\hat {]}.e_{2}\} \{ x_{1}\mathrel{\hat {\mathrel{:: }}}x_{2}.e_{3}\} }}}:
     \item  Analogous to cases for {{\color{\colorMATH}\ensuremath{e = {{\color{\colorSYNTAX}\mtexttt{if}}}(e_{1})\{ e_{2}\} \{ e_{3}\} }}}.
     \end{itemize} % }-}
  \item  \begin{itemize}[label={},leftmargin=0pt]\item  {\mtextbf{Case}} {{\color{\colorMATH}\ensuremath{n=n^{\prime}+1}}} {\mtextbf{and}} {{\color{\colorMATH}\ensuremath{e=\lambda _{z}x.\hspace*{0.33em}e}}} {\mtextbf{and}} {{\color{\colorMATH}\ensuremath{\tau  = \tau _{1} \rightarrow  \tau _{2}}}}: % {-{
        By inversion on {\mtextit{(H4)}} we have:
        {{\color{\colorMATH}\ensuremath{n^{\prime} = 0}}}, and
        {{\color{\colorMATH}\ensuremath{v_{1} = \langle \lambda x_{z}.\hspace*{0.33em}e\mathrel{|}\gamma _{1}\rangle }}}.
        Instantiate {{\color{\colorMATH}\ensuremath{v_{2} = \langle \lambda _{z}x.\hspace*{0.33em}e\mathrel{|}\gamma _{2}\rangle }}}.
        To show:
        {\mtextit{(C1)}}: {{\color{\colorMATH}\ensuremath{\gamma _{2} \vdash  \lambda _{z}x.\hspace*{0.33em} e \Downarrow  \langle \lambda _{z}x.\hspace*{0.33em}e\mathrel{|}\gamma _{2}\rangle }}} unique; and
        {\mtextit{(C2)}}: {{\color{\colorMATH}\ensuremath{\langle \lambda _{z}x.\hspace*{0.33em}e\mathrel{|}\gamma _{1}\rangle  \sim  \langle \lambda _{z}x.\hspace*{0.33em}e\mathrel{|}\gamma _{2}\rangle  \in  {\mathcal{V}}_{n^{\prime}}^{\Sigma }\llbracket \tau _{1} \rightarrow  \tau _{2}\rrbracket }}}.
        Unfolding the definition, we must show: {{\color{\colorMATH}\ensuremath{\forall  n^{\prime \prime}\leq n^{\prime},v_{1},v_{2},.\hspace*{0.33em} v_{1} \sim  v_{2} \in 
        {\mathcal{V}}_{n^{\prime \prime}}^{\Sigma }\llbracket \tau _{1}\rrbracket  \Rightarrow \{ x\mapsto v_{1},z\mapsto \langle \lambda _{z}x.\hspace*{0.33em}e\mathrel{|}\gamma _{1}\rangle \} \uplus \gamma _{1},e \sim  \{ x\mapsto v_{2},z\mapsto \langle \lambda _{z}x.\hspace*{0.33em}e\mathrel{|}\gamma _{2}\rangle \} \uplus \gamma _{2},e \in 
        {\mathcal{E}}_{n^{\prime \prime}}^{\Sigma }\llbracket \tau _{2}\rrbracket }}}.
        To show, we assume:
        {{\color{\colorMATH}\ensuremath{v_{1} \sim  v_{2} \in  {\mathcal{V}}_{n^{\prime}}^{\Sigma }\llbracket \tau _{1}\rrbracket }}} {\mtextit{(C2.H1)}}.
        Note the following facts:
        {{\color{\colorMATH}\ensuremath{\gamma _{1} \sim  \gamma _{2} \in  {\mathcal{G}}_{n^{\prime}}^{\Sigma }\llbracket \Gamma \rrbracket }}}                             {\mtextit{(F1)}}; and
        {{\color{\colorMATH}\ensuremath{\{ x\mapsto v_{1}\} \uplus \gamma _{1} \sim  \{ x\mapsto v_{2}\} \uplus \gamma _{2} \in  {\mathcal{G}}_{n^{\prime}}^{\Sigma }\llbracket \{ x\mapsto \tau _{1},z\mapsto \tau _{1}\rightarrow \tau _{2}\} \uplus \Gamma \rrbracket }}} {\mtextit{(F2)}}.
        {\mtextit{(F1)}} holds from {\mtextit{H1}} and {\nameref{thm:step-index-weakening}.1}.
        {\mtextit{(F2)}} holds from {\mtextit{(F1)}}, {\mtextit{(C2.H1)}} and the definition of {{\color{\colorMATH}\ensuremath{\gamma  \sim  \gamma  \in  {\mathcal{G}}_{n}^{\Sigma }\llbracket \Gamma \rrbracket }}}.
        Conclusion holds by IH ({{\color{\colorMATH}\ensuremath{n = n^{\prime}}}} decreasing), {{\color{\colorMATH}\ensuremath{F2}}} and {{\color{\colorMATH}\ensuremath{C2.H1}}}.
     \end{itemize} % }-}
  \item  \begin{itemize}[label={},leftmargin=0pt]\item  {\mtextbf{Case}} {{\color{\colorMATH}\ensuremath{n=n^{\prime}+1}}} {\mtextbf{and}} {{\color{\colorMATH}\ensuremath{e=e_{1}(e_{2})}}}: % {-{
     \item  By inversion on {\mtextit{(H4)}} we have:
     \item  \ {{\color{\colorMATH}\ensuremath{\begin{array}{rclcl@{\hspace*{1.00em}}l
           } \gamma _{1}         &{}\vdash {}& e_{1}  &{}\Downarrow _{n_{1}}{}& \langle \lambda _{z}x.\hspace*{0.33em}e_{1}^{\prime}\mathrel{|}\gamma _{1}^{\prime}\rangle        & {{\color{\colorTEXT}\textnormal{{\mtextit{(H4.1)}}}}}
           \cr  \gamma _{1}         &{}\vdash {}& e_{2}  &{}\Downarrow _{n_{2}}{}& v_{1}                    & {{\color{\colorTEXT}\textnormal{{\mtextit{(H4.2)}}}}}
           \cr  \{ x\mapsto v_{1},z\mapsto \langle \lambda _{z}x.\hspace*{0.33em}e_{1}^{\prime}\mathrel{|}\gamma _{1}^{\prime}\rangle \} \uplus \gamma _{1}^{\prime} &{}\vdash {}& e_{1}^{\prime} &{}\Downarrow _{n_{3}}{}& v_{1}^{\prime} & {{\color{\colorTEXT}\textnormal{{\mtextit{(H4.3)}}}}}
           \end{array}}}}
     \item  and we also have:
        {{\color{\colorMATH}\ensuremath{n^{\prime} = n_{1} + n_{2} + n_{3} + 1}}}, and
        {{\color{\colorMATH}\ensuremath{v_{1} = v_{1}^{\prime}}}}.
        By IH ({{\color{\colorMATH}\ensuremath{n = n^{\prime}}}} decreasing), {\mtextit{(H1)}}, {\mtextit{(H2)}}, {\mtextit{(H3)}}, {\mtextit{(H4.1)}} and {\mtextit{(H4.2)}} we have:
        {{\color{\colorMATH}\ensuremath{\gamma _{2}         \vdash  e_{1}  \Downarrow _{n_{1}} \langle \lambda _{z}x.\hspace*{0.33em}e_{2}^{\prime}\mathrel{|}\gamma _{2}^{\prime}\rangle }}} {\mtextit{(IH.1.C1)}},
        {{\color{\colorMATH}\ensuremath{\gamma _{2}         \vdash  e_{2}  \Downarrow _{n_{2}} v_{2}}}}              {\mtextit{(IH.2.C1)}},
        {{\color{\colorMATH}\ensuremath{\langle \lambda _{z}x.\hspace*{0.33em}e_{1}^{\prime}\mathrel{|}\gamma _{1}^{\prime}\rangle  \sim  \langle \lambda _{z}x.\hspace*{0.33em}e_{2}^{\prime}\mathrel{|}\gamma _{2}^{\prime}\rangle  \in  {\mathcal{V}}_{n^{\prime}-n_{1}^{\Sigma }}\llbracket \tau _{1} \rightarrow  \tau _{2}\rrbracket }}}  {\mtextit{(IH.1.C2)}}, and
        {{\color{\colorMATH}\ensuremath{v_{1} \sim  v_{2} \in  {\mathcal{V}}_{n^{\prime}-n_{2}^{\Sigma }}\llbracket \tau _{1}\rrbracket }}} {\mtextit{(IH.2.C2)}}.
        Note the following facts:
        {{\color{\colorMATH}\ensuremath{n_{3} \leq  n^{\prime}-n_{1}-n_{2}}}}  {\mtextit{(F1)}};
        {{\color{\colorMATH}\ensuremath{\gamma _{1} \sim  \gamma _{2} \in  {\mathcal{G}}_{n-n_{1}^{\Sigma }-n_{2}}\llbracket \Gamma \rrbracket }}} {\mtextit{(F2)}}; and
        {{\color{\colorMATH}\ensuremath{v_{1} \sim  v_{2} \in  {\mathcal{V}}_{n-n_{1}^{\Sigma }-n_{2}}\llbracket \tau _{2}\rrbracket }}} {\mtextit{(F3)}}.
        {\mtextit{(F1)}} follows from {\mtextit{(H3)}} and {{\color{\colorMATH}\ensuremath{n^{\prime} = n_{1} + n_{2} + n_{3} + 1}}}.
        {\mtextit{(F2)}} and {\mtextit{(F3)}} follow from {\mtextit{(H1)}}, {\mtextit{(IH.2.C2)}} and \nameref{thm:step-index-weakening}.
        By IH ({{\color{\colorMATH}\ensuremath{n = n^{\prime}-n_{1}-n_{2}}}} decreasing), {\mtextit{(H2)}}, {\mtextit{(IH.1.C2)}}, {\mtextit{(IH.2.C2)}}, {\mtextit{(F1)}}, {\mtextit{(F2)}},
        {\mtextit{(F3)}} and {\mtextit{(H4.3)}} we have:
        {{\color{\colorMATH}\ensuremath{\{ x\mapsto v_{2},z\mapsto \langle \lambda _{z}x.\hspace*{0.33em}e_{2}^{\prime}\mathrel{|}\gamma _{2}^{\prime}\rangle \} \uplus \gamma _{2}^{\prime} \vdash  e_{2}^{\prime} \Downarrow _{n_{3}} v_{2}^{\prime}}}}  ({\mtextit{(IH.3.C1)}}) and
        {{\color{\colorMATH}\ensuremath{v_{1}^{\prime} \sim  v_{2}^{\prime} \in  {\mathcal{V}}_{n-n_{1}^{\Sigma }-n_{2}-n_{3}}\llbracket \tau _{2}\rrbracket }}}   {\mtextit{(IH.3.C2)}}.
        Instantiate {{\color{\colorMATH}\ensuremath{v_{2} = v_{2}^{\prime}}}}.
        To show:
        {\mtextit{(C1)}}: {{\color{\colorMATH}\ensuremath{\gamma _{2} \vdash  e_{1}(e_{2}) \Downarrow _{n_{1}+n_{2}+n_{3}+1} v_{2}^{\prime}}}} (unique); and
        {\mtextit{(C2)}}: {{\color{\colorMATH}\ensuremath{v_{1}^{\prime} \sim  v_{2}^{\prime} \in  {\mathcal{V}}_{n-n_{1 1}^{\Sigma }-n_{1 2}-n_{1 3}-1}\llbracket \tau _{2}\rrbracket }}}.
        {\mtextit{(C1)}} is immediate from {\mtextit{(IH.1.C1)}}, {\mtextit{(IH.2.C1)}} and {\mtextit{(IH.3.C1)}}.
        {\mtextit{(C2)}} is immediate from {\mtextit{(IH.3.C2)}} and {\nameref{thm:step-index-weakening}.2}.
     \end{itemize} % }-}
  \item  \begin{itemize}[label={},leftmargin=0pt]\item  {\mtextbf{Case}} {{\color{\colorMATH}\ensuremath{n=n^{\prime}+1}}} {\mtextbf{and}} either {{\color{\colorMATH}\ensuremath{e = {{\color{\colorSYNTAX}\mtexttt{reveal}}}(e^{\prime})}}} {\mtextbf{or}} {{\color{\colorMATH}\ensuremath{e = % {-{
        {{\color{\colorSYNTAX}\mtexttt{return}}}(e^{\prime})}}} {\mtextbf{or}} {{\color{\colorMATH}\ensuremath{e = {{\color{\colorSYNTAX}\mtexttt{laplace}}}[e_{1},e_{2}](e_{3})}}} {\mtextbf{or}} {{\color{\colorMATH}\ensuremath{e = x \leftarrow  e_{1} \mathrel{;} e_{2}}}}:
     \item  Follows from inductive hypothesis, post processing (for {{\color{\colorMATH}\ensuremath{{{\color{\colorSYNTAX}\mtexttt{laplace}}}}}}) and
        sequential composition (for {{\color{\colorMATH}\ensuremath{x \leftarrow  e_{1} \mathrel{;} e_{2}}}}) theorems from the
        differential privacy literature~\cite{privacybook}.
     \end{itemize}
     % }-}
  \end{itemize}
\end{proof}
% }-}
% }-}

\endgroup

\end{document}